\documentclass{amsart}
\usepackage{graphics,epsfig,epsf,psfrag}
\usepackage{url}
\usepackage{color, appendix}
\usepackage{graphicx}
\usepackage{amsmath, amsfonts, amsthm,amssymb}

\usepackage{epstopdf}
\usepackage{hyperref}
\usepackage{enumerate, framed}
\usepackage{graphicx,lscape,rotate}
\usepackage{algorithmic, algorithm}
\usepackage{epsfig}
\usepackage{mathrsfs}
\usepackage{amsfonts}
\usepackage{mathtools}
\usepackage{tasks}
\usepackage{ifthen}
\usepackage{bbm}
\graphicspath{{eps/}}
\usepackage{amsthm}
\usepackage[utf8]{inputenc}
\usepackage[english]{babel}
\usepackage{bm}
\usepackage[round]{natbib}
\usepackage{comment}
\numberwithin{equation}{section}
\theoremstyle{plain}
\newtheorem{theorem}{Theorem}[section]
\newtheorem{assumption}{Assumption}[section]
\newtheorem{remark}{Remark}
\newtheorem{proposition}{Proposition}[section]

\newtheorem{lemma}{Lemma}[section]
\newtheorem{definition}{Definition}
\newtheorem{example}[theorem]{Example}
\theoremstyle{remark}
\newcommand{\eq}{\begin{equation}}
\newcommand{\eeq}{\end{equation}}
\newcommand{\R}{ \mathbb{R}}

\newcommand\vu{\vec{u}}
\newcommand{\s}{ L}
\newcommand{\Pp}{ \mathscr{P}}

\def\braket#1#2{\langle{#1}|{#2}\rangle}

\title{A New Paradigm For Scattering Theory of Linear And Nonlinear Waves: Review And Open Problems}
 \author[Avy Soffer]{Avy Soffer}
 \address[Avy Soffer]{\newline
        Department of Mathematics, \newline
         Rutgers University, New Brunswick, NJ 08903 USA.}
  \email[]{soffer@math.rutgers.edu}
\date{May 5, 2024}
\begin{document}
\maketitle
\begin{abstract}
I present a review of the recent advancements in scattering theory, which provides a unified approach to studying dispersive and hyperbolic equations with general interaction terms and data. These equations encompass time-dependent potentials, as well as NLS, NLKG, and NLW equations. Additionally, I discuss a series of open problems, along with their significance and potential future applications in scattering and inverse scattering.

\end{abstract}
\tableofcontents
\section{Introduction}
To provide context, I briefly trace the history of scattering theory: Leonardo Da Vinci extensively examined the phenomenon of wave scattering off objects. In his summation of observations, he wrote:

\emph{"The air that is between bodies is full of the intersections formed by the radiating images of these bodies."}

Remarkably, this statement encapsulates the essence of the two fundamental meta-theorems of scattering theory.

Firstly, there is direct scattering: systems, however complicated, break into subsystems with simpler and independent dynamics as time approaches infinity. Secondly, all the information about the system is carried into the asymptotic states. Consider light arriving without any prior information from infinity. It interacts with the painting of the Mona Lisa. As time approaches positive infinity, the light becomes free, and the painting returns to a static state. The reflected light contains all the information about the painting extending to positive infinity in space.

A few comments are necessary: How can one prove such a result? In contemporary science, this entails solving the complicated Matter+Radiation system. The "air" state carries the information. It wasn't known then that light propagates through vacuum as well, but Da Vinci understood it as a form of wave propagation. (The existence and concept of vacuum emerged 150 years later in the works of Torricelli). Air has other flows besides acoustic waves. Are they necessary to perceive an object with ultrasound? This will be discussed in the section on open problems.

Moving a few hundred years forward, we encounter the notion of Limiting Absorption \cite{v1905reflexion,sommerfeld1912greensche}, introduced by Sommerfeld, which includes the principles of limiting absorption and Sommerfeld's Radiation condition, as a means to study scattering.

These concepts played a pivotal role: they are associated with transforming the scattering problem into a time-independent one by considering the Green's function of the equation, thus reducing the problem to solving elliptic equations. Consequently, it serves as a crucial tool for calculations. These works primarily apply to the wave equation.

In 1911, Rutherford revolutionized the field of atomic physics by introducing the scattering approach, earning him the title of the father of the atomic era. In his seminal 1911 paper, often hailed as the "most beautiful paper in physics" \cite{rutherford1911lxxix}, he tackled the classical Coulomb scattering problem and introduced the concept of the differential cross-section.

Following Da Vinci's methodology, Rutherford utilized the solution of the scattering problem to compute the back-scattering from two atomic models: one featuring nuclei surrounded by diffuse opposite charges and another portraying the atom as a uniform distribution of charges without a nucleus.

Subsequently, he conducted experiments, bombarding a thin gold layer with $\alpha$ particles, and measured the resulting back-scattering. Through probabilistic analysis, he conclusively demonstrated the existence of atomic nuclei.

Until 1977, these concepts and techniques flourished, culminating in the refinement of scattering theory into a highly sophisticated discipline, which included the proof of Asymptotic Completeness of three-body quantum scattering by Faddeev. Asymptotic Completeness, representing the first meta-theorem of scattering theory, asserts that all initial conditions lead to independently moving subsystems with simpler dynamics.

In 1978, V. Enss \cite{E1978}, following preliminary studies by Ruelle and Haag, introduced a groundbreaking, \emph{time-dependent} approach to scattering theory. This novel methodology directly analyzes the flow in phase space , initially applied to one-body quantum scattering involving a spatially localized and time-independent potential. Essentially, this approach reduced the problem to stationary phase estimates of the free flow for micro-localized initial conditions.

This innovative technique sparked numerous investigations, ultimately leading to a comprehensive resolution of the Asymptotic Completeness problem for short and long-range three-particle scattering.

The next significant development in this review is the phase-space approach introduced by Sigal-Soffer in 1987 \cite{SS1987}. This approach led to  the solution of the N-body scattering problem in quantum mechanics.  This approach involves decoupling the various scattering channels through a partition of the phase space.

In this partition, each element segregates the N particles into distinct conical sets in the configuration space ($x$ variable) and also decomposes the momentum space (the dual space, Fourier dual, etc.). This dual decomposition ensures that the \emph{boundary of each partition element} is supported in the classically forbidden region. This region denotes where either $x$ is not parallel to $p$, or the total energy deviates from  the   (conserved) energy of the initial data. Subsequently, it is demonstrated that there is no propagation of the solution within these forbidden regions. However, these last estimates are established for the full (interacting) dynamics using the method of micro-local propagation estimates.

Regrettably, there is no apparent method to extend this approach to nonlinear equations with multiple channels of propagation. Nonetheless, these methods have been successfully applied to time-independent systems, including N-body long-range scattering and certain aspects of quantum field theory.

Moving forward, let's delve into Nonlinear Dispersive Dynamics. A significant advancement, though modest, was the exploration of scattering in the nonlinear case concerning small perturbations of exact solutions, particularly Solitons of the Nonlinear Schrödinger Equation (NLS). Soffer-Weinstein (\cite{soffer1990multichannel}) introduced and utilized a method for deriving modulation equations, enabling the separation of a soliton from the radiation part for non-completely integrable systems. These equations could then be solved up to infinity by imposing a smallness condition on the radiation part. This paved the way for a broad class of asymptotic stability theories for coherent states of various dispersive and hyperbolic equations. In cases where the coherent structure is small and unstable, it became feasible to address it using the notion of nonlinear Fermi's Golden Rule \cite{S-Wei5,Ts2002, Sig}. However, it's crucial to note the emphasis on smallness here.

Consequently, the lingering issue pertains to addressing problems involving large data.

In 2004, T. Tao introduced a novel approach to address nonlinear scattering. Its inception is rooted in building on the strategy of the Enss approach, initially in the  paper by Rodnianski and Tao \cite{rodnianski2004longtime} for a linear problem, and subsequently expanded upon in a series of works. This culminated in the proof of the decomposition of solutions of the Nonlinear Schrödinger Equation (NLS) with inter-critical interactions, in three or more dimensions in the radial case.

Tao's works \cite{T2014,T2008,T2007,T2006,T2004} demonstrate that:

\begin{theorem}
Let $\psi(t)$ be a solution of the NLS with inter-critical interaction term. Suppose the solution is uniformly bounded in $H^{1}_{rad}(\mathbb{R}^n)$, where $n\geq 3$. Then, the solution strongly converges in $H^1$ to a solution of the free Schrödinger equation plus a weakly localized part around the origin. The weakly localized part is smooth (in the case of quadratic nonlinearity) and lies in the domain of the dilation operator $A=(-i/2)(x\cdot \nabla_x+\nabla_x\cdot x)$ to an arbitrary power, where $x\in \mathbb{R}^n$.
\end{theorem}

Here $H^1_{rad}$ stands for $L^2$ functions $f$, which are spherically symmetric and such that $\nabla_x f \in L^2.$

The term "inter-critical nonlinear" implies mass super-critical and energy sub-critical conditions. Despite the large initial data, such equations exhibit numerous global solutions, including in the focusing case, as seen in \cite{beceanu2021large} and cited references.

In the Enss method, Ruelle's theorem is employed, stating that any state in the continuous spectrum of the Hamiltonian spreads in space on average over time. This principle essentially derives from a Wiener's theorem concerning the Fourier transform of a finite measure. This spreading phenomenon implies that the solution remains away from the interaction (localized around the origin), enabling the utilization of propagation estimates of the free flow, acting separately on incoming and outgoing parts of the solution. Unfortunately, there is no  theorem analogous to Ruelle's, that is available in the nonlinear case.

Instead, Tao used the weak convergence of the M\"oller wave operators to peel off the free part of the solution.
Then, the use of the decomposition to incoming and outgoing waves is used to estimate the Duhamel terms.
As described, this works for localised interactions around the origin.

 The localization of the interaction term is derived from the assumption of spherical symmetry and the $H^1$ boundedness, using the Radial Sobolev Embedding Theorems:
\begin{align}\label{radial}
&|\psi(r)| \lesssim r^{-(n-1)/2}\|\psi\|_{H^1}, \quad r\geq 1.\\
&|\psi(r)| \lesssim r^{-n/p_c(n)}\|\psi\|_{\dot H^1}, \quad r<1, \quad p_c(n)=2n/(n-2).
\end{align}

An essential aspect of this analysis involves the construction and utilization of incoming and outgoing projections, with many such constructions independently introduced over the years.

Another facet of Tao's work revolves around comprehending the structure of the non-radiative part, referred to as "weakly localized" in the case of localized interactions. A significant advancement in this direction, based on the profile decomposition, was carried out by T. Roy \cite{roy2017weak} for the NLS with bi-harmonic Laplacian, within the Sobolev space $H^2(\mathbb{R}^5)$. In this case, it is demonstrated that the solution belongs to a G-compact set, which is compact in the space of localized functions acted upon by the Galilean Group.

Tao also established such a result for NLS with a smooth compactly supported potential and a defocusing nonlinear term. These works also give rise to a standard conjecture regarding Asymptotic Completeness (AC) for NLS, termed {\bf Soliton resolution.}

As noted by Tao, this conjecture can only be expected to hold in a generic sense, as there exist many coherent states that are not solitons (such as breathers, lumps of various types, vortices, kinks, and their combinations). This complexity complicates the proofs since it's not known a priori what to prove in a specific case. Subsequent remarkable works by Killip, Visan, Murphy, and others provide sharp size estimates on the initial data of NLS that lead to either scattering or blow-up, without yielding weakly localized states. However, in further studies, when the nonlinearity is not a monomial and has more than one (formal) minimum, it is shown that a potential minimum state is a scaled soliton but not a soliton. Its large-time behavior differs in each time direction. The analysis in the aforementioned works, heavily relies on the use of Dilation and Morawetz-type identities.

The overarching philosophy and strategy described above, were subsequently utilized by Kenig-Merle and collaborators\cite{collot2024soliton} to examine the large-time dynamics of the energy-critical wave equation in three or higher dimensions. Their series of works led to a comprehensive solution to such problems, demonstrating that solutions of initial data in $\dot{H}^1$ ( the Homogeneous Sobolev Space), are asymptotic to a free wave plus a finite sum of scaled solitons. These results also hold for solutions that blow up in finite time, see \cite{duyckaerts2012profiles} and cited references. See also \cite{duyckaerts2024profile}.

\section{ Microlocal Propagation Estimates}

\subsection{Introduction- Scattering in General}

To establish a scattering theory applicable to both linear and nonlinear dispersive and hyperbolic equations, we require a principle that can be universally applied across all equations. Let's examine the scenario where an interaction term is spatially localized (albeit weakly, exhibiting some polynomial decay for large $x$) and possibly time-dependent. This situation is quite general, encompassing both linear and nonlinear equations, provided that purely nonlinear terms (Interaction terms that depend only on $\psi.$) are constrained by the assumption that the solution is spherically symmetric in some $H^s$ space and in dimensions larger than 1.

\begin{remark}
The solutions of nonlinear equations or time dependent potential problems can blow-up in a finite time.
In these cases, there is no scattering, so it is necessary to restrict attention to initial data which lead to a global solution, and that some relevant norm (typically the conserved quantities) remain uniformly bounded in time.
By adding a defocusing nonlinear term to the equation, or by choosing a saturated  nonlinearity, we can ensure global existence for all initial data in some Banach space. Such saturated equations are widely used, as they are more physically (and numerically) applicable, see e.g. \cite{sulem2007nonlinear,fibich2015nonlinear,rodnianski2003asymptotic}.
\end{remark}

Remarkably, this scenario lacks treatment in the physics literature. While special cases like impulse, adiabatic, or small perturbations were addressed in the early days of Quantum Theory, the general case remains largely unexplored. Ideally, one would intuitively argue in a circular manner: suppose the solution (or a portion of it) behaves like a free wave for large time intervals. Calculate the interaction term with this free wave. If it is integrable over time under the free flow, we can discard the interaction, leaving us with a free dynamics, thus ensuring self-consistency.

While this line of reasoning often yields the expected results, particularly in predicting in the long-range potential case the corrections to free dynamics at infinity, it is less useful for nonlinear equations or for formal proofs. It primarily works well within the realm of perturbation theory. Therefore, a different paradigm is necessary to tackle such equations rigorously.

\subsection{The setup of the problem}

The starting point is a general argument: let's consider the behavior of the solution far from the interaction. Since the interaction is localized in space, we only need to consider the region exterior to a ball. If we are far away, we expect only free waves to be present, along with the tail of weakly localized states. Therefore, we should attempt to prove \textbf{Propagation Estimates} for the \emph{interacting, full} flow only in that region, with the expectation that they are the same as the free flow plus integrable corrections. Thus, we refer to these estimates as \emph{Exterior Propagation Estimates}.

Of course, there is a price to pay due to the flow from the interior to the exterior and back. Now comes the challenging part: microlocalization rather than just localizing  outside the ball. This microlocalization must ensure that the flow through the boundary is unidirectional (up to higher-order corrections). This means the sign of the derivative with respect to time of the mass outside the ball (in space) should be non-negative.  Such derivatives are easily calculated for a given dynamics and a given micro-local operator. We call an operator with such a good derivative a \textbf{Propagation Observable (PROB)}.

Here is the relevant formal identity.

 Consider the following class of systems

Given a self-adjoint operator family $H(t)$ acting on a Hilbert space $\mathcal{H},$
we consider the solutions (in the weak sense in general) of
\begin{align}\label{dynamics}
& i\frac{\partial \psi}{\partial t}=H(t)\psi\\
& H(t)=\omega(p)+ \mathcal{N}(|x|,t,|\psi(t)|), \quad  \mathcal{N}= \mathcal{N}^* \\
& p\equiv -i\nabla_x\\
&\psi(t=0), \nabla \psi(t=0) \in \mathcal{H}.
\end{align}

We assume that the initial data is such that global existence holds, and that the $H^1$ norm is uniformly bounded in t.

This class of equations cover NLS (Nonlinear Schroedinger) NLKG (Nonlinear Klein-Gordon) NLW(Nonlinear Wave), N-body QM, Charge Transfer Hamiltonians (generalized to be  with time dependent potentials), as well as equations in which the dispersion relation is arbitrary.

Given a family of self adjoint operators $B(t)$
we compute the \emph{expectation} of $B(t)$ in the state $\psi(t)$ by

\eq
(\psi(t),B(t)\psi(t))_{\mathcal{H}}\equiv \langle B(t)\rangle.
\eeq
Here 
$$
(f,g)_{\mathcal{H}}
$$
is the scalar product of $f$ with $g$ in the Hilbert space $\mathcal{H}.$

Then, it follows that

\eq\label{ Heisenberg}
\partial_t (\psi(t),B(t)\psi(t))_{\mathcal{H}}= \langle \{i[H(t),B(t)]+ \frac{\partial B(t)}{\partial t}\}\rangle.
\eeq

Suppose now that $B(t)$ is a good Propagation Observable:
\eq\label{prob}
 i[H(t),B(t)]+ \frac{\partial B(t)}{\partial t} \geq C^*C +R(t),
\eeq
 with $R(t)$ has an $L^1(dt)$ bounded expectation.
 
 Then the following Propagation Estimate follows:
 \eq\label{pres}
 \int_1^{\infty} \|C(t)\psi(t)\|^2 dt \leq 2\sup_t \langle B(t)\rangle + Const.
\eeq

See \cite{SS1987,SS1988,HSS1999}.

Before we show applications, we introduce another important notion: \newline {\bf Channel Wave Operators}.

 A fundamental tool to compare two flows in scattering theory was introduced in the 1940's, called M\"oller wave operators.
 Say one wants to compare the flow given by a Unitary family of operators $U(t)$ to (say the asymptotic family) $U_0(t).$

 Then, we construct the following M\"oller wave operators:
 \eq\label{moller}
s-\lim\limits_{t \to \pm \infty} U(-t)U_0(t)\equiv \Omega_{\pm}.
\eeq
The adjoint operator is  

\eq
s-\lim\limits_{t \to \pm \infty} U_0(-t)U(t)\equiv \Omega_{\pm}^*P_c.
\eeq
 The limit is the strong operator limit in the Hilbert space. Typically, the wave operator exits for all initial data,
 but the adjoint only expected to exist on the states which are asymptotic to the $U_0$ dynamics. In the time independent case $P_c$ is the operator that projects on the continuous spectral part of the Hamiltonian $H.$ There is no analog of that for the general case, and it is one of the issues, to construct such projections for time dependent potentials.

 When the problem is multichannel, as in the N-body scattering problem, the Physics-based analysis introduced the notion of Channel Wave Operators.
  These operators were constructed for each channel of scattering with a different $U_0.$
  Now the various channels are indexed by a letter, and we replace $U_0$ by $U_a.$
  Moreover, to get the right answer one needs to project the full solution on the relevant subspace, before we compare to the channel $U_a.$
  This was done by using for example the projection on an eigenfunction of a subsystem and then $U_a$ would be the free dynamics of the center of masses of the moving bound clusters.
  This is not a very effective decomposition, since the projections are unstable under the full flow.
  These projections then were replaced by a geometric decomposition of the phase-space in stead. In this way the projections were chosen to decouple the channels, by making the boundaries classically forbidden regions \cite{SS1987}.
  As explained before, we cannot use this for nonlinear problems, since we do not know what are the channels.
  The new approach I describe below is based on finding new classes of decompositions of the phase-space that work for nonlinear and time dependent equations.

  The general form of Channel wave operators that we use is:
  \eq\label{channel}
s-\lim\limits_{t \to \pm \infty} U_0(-t)J(x,p,t)U(t)\equiv \Omega_{J}^*P_c.
\eeq
The limit is in the strong sense in the Hilbert Space.
 \section{ Existence of the Free Channel Wave Operator-Radial or Localized Interactions} 

  In the paper \cite{liu2023large}, the starting point was to construct the asymptotic radial velocity operator for localized interactions. 
  It is a version of the exterior Morawetz estimate, formally. The kind of estimates one gets are different from the usual Morawetz estimate.
  A key (class of) operator is a generator of a radial vector field, we denote collectively by $\gamma.$

  \begin{definition}
Let $\Vec{f}$ be a radial vector-field, typically given by $\nabla g(|x|).$
Then we define the generator of the flow by 
$$
\gamma_f=(-i/2)(\Vec{f}\cdot \nabla_x+\nabla_x\cdot \Vec{f}).
$$
\end{definition}
For example for $\Vec{f}=x/|x|$ we get the standard radial derivative, and the corresponding $\gamma$ leads to the well known Morawetz estimate (in 3 or more dimensions) upon using it as a Propagation Observable.

Strictly speaking, the formula above for Propagation Observable does not apply, since this $\gamma$ is not self adjoint. Luckily, the boundary term has a good sign in 3 or higher dimensions.

A second example is of the type we use the most. In this case the vector-field is either equal to $x/|x|$
for $r\equiv |x|>2$, converges polynomially fast to $x/|x|$ or very slowly, depending on the applications.

In these forms, besides the original works of Morawetz, it was used by many authors for different problems, see e.g. \cite{SS1987,graf1990asymptotic,yafaev1993radiation,blue2003semilinear,blue20066,tataru2013local,killip2021scattering,T2004}.
\begin{remark}
    Note the important identity
    $$
    2\gamma_f=i[-\Delta, g], \quad f=\nabla_x g.
    $$
\end{remark}
In all of these works it was shown that in some sense one gets a Propagation Estimate by the form of the commutator of $\gamma$ with the Laplacian.
However in the work \cite{liu2023large} it is used differently; we use it to localize the direction of the flow in the phase-space.
We take the vector-field to be equal to $x/|x|$ for $r>2$, and also localize in the region of large $r$, so in fact the commutator of the radial Laplacian with $\gamma$ is zero!

Our first Propagation Observable is then given by
$$
B_1= F_1\left (\frac{|x|}{t^{\alpha}}\geq 1\right)\gamma F_1\left (\frac{|x|}{t^{\alpha}}\geq 1\right).
$$
A direct calculation of the commutator of $B_1$ with the radial Laplacian,  and estimating the time derivative by a higher order term, shows that $B_1$ is a PROB.

We find that the leading term, obtained after symmetrization of the leading term (by two more commutators) that:

The expectation of $\sqrt{ F_1 F'_1} \gamma^2\sqrt {F_1 F'_1} t^{-\alpha}$ is integrable w.r.t. time.

The commutator with the potential needs to give an integrable contribution. This is achieved by assuming the radial derivative of interactions decay faster than $<x>^{-3}$ at infinity, and $\alpha >1/3.$

The contribution of the symmetrization, due to two more commutators is $t^{-3\alpha}$ and therefore also integrable.

The limit of the expectation of $B_1$ therefore exists as time goes to infinity. It is a real number. So it is positive, negative or zero.
If it is negative it can be used to show the solution blows up in a finite time. It means there are incoming waves at time plus infinity.

If the limit is positive there should be free waves at infinity, since some part is moving out with a constant speed.

If the limit is zero, it cannot have a free wave. So, it is a weakly localized state. Further analysis of the above Propagation Estimate implies that the weakly localized state cannot spread faster that $\sqrt t.$ 
$$
\langle |x|\rangle \lesssim \sqrt t.
$$
The next class of Propagation Observables introduced are more complicated:
\eq\label{prob2}
B_2= F_1\left (\frac{|x|}{t^{\alpha}}\geq 1\right)F_2(t^{\beta}\gamma\geq 1)+F_2F_1.
\eeq
Here $0\leq\beta < \alpha.$
To prove that these operators are Propagation Observables one needs to use the phase-space formalism developed to work with operators which are not pseudo-differential.

The construction of the operators $F_2(\gamma)$ uses the spectral theorem in an essential way.
First, note that we use a smooth vector field to define $\gamma;$  therefore $\gamma$ is an unbounded self adjoint operator on the Hilbert Space $L^2(\R^n).$

By the spectral theorem, functions of $\gamma$ are well defined and are self adjoint if the function is real, and they are bounded if the function is uniformly bounded. While $\gamma$ is a PDO (pseudo-differential operator), functions of it are not. (the same is true for $A$, the dilations operator). However, since $\gamma$ is a generator of a vector-field flow, we can solve for the flow, and then we know in some sense how the operator $e^{i\lambda \gamma}$ acts on functions.
Then integrating over $\lambda$ the function $\hat F(\lambda)e^{i\lambda \gamma}$ defines the operator $F(\gamma).$

The key tool for commuting is the \emph{commutator expansion lemma} \cite{SS1987,HSS1999,amrein1996commutator,helffer1989operateurs}.
 Commutator Expansion Lemma from \cite{SS1987,HS2000}:
 \begin{proposition}[Commutator Expansion Lemma]\label{prop: c}Let $A,B$ be self-adjoint operators, and $B$ is bounded. Let $f\in \mathcal{B}_m$, then we have the following commutator formulas:
 \begin{subequations}
 \eq
 [B,f(A)]=i\int_{R}\hat{f}(s)\int_0^s e^{i(s-u)A} [B,A] e^{iuA} duds,
 \eeq
 \eq
 [B,f(A)]=\sum\limits_{k=1}^{m-1} \frac{1}{k!} f^{(k)}(A) ad_A^{(k)}(B)+R_m,
 \eeq
 and
 \eq
 [B,f(A)]=\sum\limits_{k=1}^{m-1} \frac{1}{k!} (-1)^{k-1} ad_A^{(k)}(B)f^{(k)}(A)-R_m^*.
 \eeq
 \end{subequations}
 Here $ad_A^{(k)}(B)$ is the higher commutators:
 \begin{subequations}
 \eq
 ad_A^{(1)}(B)=[B,A],
 \eeq
 and
 \eq
 ad_A^{(k)}(B)=[ad_A^{(k-1)}(B), A],
 \eeq
 and the remainder term can be estimated as 
 \eq
 \|R_m\|_{\s^2_x(\R^n)\to \s^2_x(\R^n)}\leq C\| ad_A^{(k)}(B)\|\int_\R |f(\hat{\xi})||s|^m ds
 \eeq
 and
 \eq
 \|R_m^*\|_{\s^2_x(\R^n)\to \s^2_x(\R^n)}\leq C\| ad_A^{(k)}(B)\|\int_\R |f(\hat{\xi})||s|^m ds
\eeq
\end{subequations}
\end{proposition}

Using this expansion, we see that the leading term in the  \emph{Heisenberg Derivative} $i[-\Delta,B_2]+ \dot{B_2}$ is non-negative:
$$
t^{-\alpha}\sqrt{F'_1}\left(\gamma F_2+ t^{-1}F'_2\right )\sqrt{F'_1}.
$$
The resulting integrability of this term implies that the free channel wave operator exists with 
$$
J\equiv B_2, \quad \alpha>1/3, 3\alpha-\beta >1.
$$
This is proved by using Cook's argument applied to the Channel wave operator, and deducing the integrability from the Propagation Estimate above.
Similar estimates work with projection on $\gamma$ negative.
This implies that the weakly localized part of the solution concentrates (up to log corrections) in the region of phase space
$$
\bigcup_{\alpha \leq 1/2} \{ |x|\sim t^{\alpha}\cap |p|\sim t^{-\alpha}\}.
$$
This basically shows that the weakly localized part, which is supported away from the origin is self similar function.
The rest of the work \cite{liu2023large} focused on deriving estimates on the weakly localized part.
An important part of that is the proof that the weakly localized part is a function in the domain of the unbounded operator $A,$ the Dilation operator. Few other ways of proving such estimates were developed later in the works \cite{Sof-W2,SW2020,SW20221,Sof-W3,Sof-W5}.

\section{ The Properties of weakly localized states-I}

The properties of the weakly localized part of the solution, are mostly studied in the radial and localized cases.
Tao \cite{T2004} showed that this part of the solution is smooth for quadratic nonlinear term (in three or more dimensions)
and that the solution must have a heavy part near the origin. Furthermore, this function is in the domain of $A^n, \, n\in \mathrm{N}^+$, where $A$ is the dilation generator. These results are for inter-critical monomial term as a nonlinear interaction.

More general result was proven for the non-radial case by T. Roy, for the bi-harmonic NLS, where it is shown that the weakly localized part is a G-compact set. That is, the trajectory is in a compact set translated by the Galilean group. \cite{roy2017weak}.

In \cite{liu2023large} it is shown for the radial case in three dimensions and general interaction that the weakly localized part is regular, with a core near the origin and essentially self similar. This was also proved by later works of Xiaoxu Wu and me that will be described later. Here I will then present first some of the tools used, which are of general interest.

\subsection{Incoming/Outgoing Decompositions}

The decomposition of dispersive waves to incoming and outgoing waves is non-trivial.
In the case of Hyperbolic equations, it is possible to write down an explicit differential operator that can do the job.
For the Schr\"odinger equation it is not possible. The deep reason is that the uncertainty principle makes it impossible to localize the direction of the frequency at a given point in space. (one should not take the word "impossible" literally: Using the Madelung formalism of QM, one can talk about the velocity of the "fluid" at a given point!)
The introduction of such decompositions was made in many ways for different works and purposes.

In scattering theory it was used by Enss \cite{E1978,PSS}. He used objects of the form $F(x)G(p)$ to get the desired microlocalization.
Outgoing would be if $x\cdot \xi >0$ for all $(x,\xi)$ in the support of $F(x)G(\xi)$ as a function on the phase-space.
Inspired by this work of Enss,\cite{E1978} E. Mourre\cite{M1979,M1981} came up with a different decomposition, which is global. He then developed a remarkable spectral theory of general self-adjoint operators, which allowed proving local decay estimates, decay of eigenfunctions, absence of singular spectrum and more.

Mourre's construction begins with the Dilation operator $A.$ Formally, the projection on incoming/outgoing waves should be $P^{\pm}(A),$ where the function $P^{\pm}(z)$ is the characteristic function on the intervals $[0,\infty)$ for outgoing, $(-\infty,0]$ for incoming.
 In fact that is what Mourre used, and these operators are defined by the spectral theorem. These operators and in fact general functions of $A$ are badly behaved. Mourre used a Contour integral representation for these operators.

 A time dependent approach was later developed by Sigal-Soffer \cite{SS1987,SS1988} used a smooth version of these cutoff functions and this approach led to optimal propagation estimates for general Hamiltonian that satisfy the crucial Mourre estimate:
 Let $H,A$ be self-adjoint operators such that
\begin{align}
&\textbf{Mourre Estimate}\\
&\quad F(H\sim E)i[H,A]F(H\sim E) \geq \theta F(H\sim E)^2 +K,\\
&\text{with compact} \, K.
\end{align}
Here we assume that $\theta>0.$ If $K=0,$ it is the Strong Mourre estimate.
$A$ is called \emph{conjugate operator for $H$ at energy $E$}.
Under technical domain assumptions, and regularity, we can use this estimate to prove minimal and maximal velocity bounds of the form 

\eq\label{PES}
\|F(|\frac{|x|}{t}-v|\geq \epsilon)F(H\sim E)e^{-iHt}\psi(0)\|_{L^2_x} \lesssim <t>^{-m}.
\eeq
$v$ is the group velocity of the wave at energy $E.$
To derive the estimate one needs to verify technical conditions, since the operators in question are not bounded in general.

Let $A$ be a self-adjoint operator on a Hilbert space $\mathcal{H}$.  If $S$
is a bounded operator on $\mathcal{H}$ then we denote $[A,S]_\circ$ the
sesquilinear form on $D(A)$ defined by
$[A,S]_\circ(u,v)=\braket{Au}{Sv}-\braket{u}{SAv}$. As usual, we set
$[S,A]_\circ=-[A,S]_\circ$, $[S,i A]_\circ=i [S,A]_\circ$, etc.  We
say that \emph{$S$ is of class $C^1(A)$}, and we write $S\in C^1(A)$,
if $[A,S]_\circ$ is continuous for the topology induced by $\mathcal{H}$ on
$D(A)$ and then we denote $[A,S]$ the unique bounded operator on $\mathcal{H}$
such that $\braket{u}{[A,S]v}=\braket{Au}{Sv}-\braket{u}{SAv}$ for all
$u,v\in D(A)$.  It is easy to show that $S\in C^1(A)$ if and only if
$SD(A)\subset D(A)$ and the operator $SA-AS$ with domain $D(A)$
extends to a bounded operator $[A,S]\in B(\mathcal{H})$. Moreover, $S$ is of
class $C^1(A)$ if and only if the following equivalent conditions are
satisfied
\begin{assumption}
\item the function $t \mapsto e^{-i t A}S e^{i t A}$ is Lipschitz in
  the norm operator topology
\item the function $t\mapsto e^{-i t A}S e^{i t A}$ is of class
  $C^{1}$ in the strong operator topology
\end{assumption}
and then we have $[S,i A]=\frac{ d}{ dt}e^{-i t A}Se^{i t
  A}|_{t=0}$.

Clearly $C^1(A)$ is a $*$-subalgebra of $ B(\mathcal{H})$ and the usual
commutator rules hold true: for any $S, T\in C^{1}(A)$ we have
$[A,S]^*=-[A,S^*]$ and $[A,ST]= [A,S]T+ S[A,T]$, and if $S$ is
bijective then $S^{-1}\in C^1(A)$ and $ [A,S^{-1}]= -
S^{-1}[A,S]S^{-1}$.

%
%

We consider now the rather subtle case of unbounded operators. Note
that we always equip the domain of an operator with its graph
topology.  If $H$ is a self-adjoint operator on $\mathcal{H}$ then
$[A,H]_\circ$ is the sesquilinear form on $D(A)\cap D(H)$ defined by
$[A,H]_\circ(u,v)=(Au,Hv)-(Hu,Av)$. By analogy with
the bounded operator case, one would expect that requiring denseness
of $D(A)\cap D(H)$ in $D(H)$ and continuity of $[A,H]_\circ$ for the
graph topology of $D(H)$ would give a good $C^1(A)$ notion. For
example, this should imply the validity of the Virial theorem, nice
functions of $H$ (at least the resolvent) should also be of class
$C^1$, etc. However this is not true, as the  example from
\cite{Ger,georgescu1999virial} shows.

One can prove sharper localization and decay estimates, following the above bounds. \cite{larenas2015abstract, HSS1999}
 In the works discussed here, these estimates are used for free Hamiltonians of the form
 $H=\omega(p).$ Many of these estimates can be obtained by stationary phase methods as well.
 However, the extension of such estimates to time dependent potentials and non-linear equations requires other approaches. The need for localization of the Hamiltonian in the Mourre estimate, requires new steps in the time dependent case.

Still another localization of the radial velocity was introduced and used in\cite{SS1987}. Here the smooth functions with compactly supported derivative are used to localize the operator $\gamma$  chosen to be self-adjoint, $F(\gamma)$,
 with $2\gamma=\frac{x}{<x>}\cdot p + c.c., $ where $c.c.$ stands for complex conjugate.
 These operators have bad behavior as well. But functions of $\gamma$ are more informative than that of $A.$
 We only use smooth functions of $\gamma$ and in order to approach the crucial point of zero frequency, we do that in a second microlocal way, by zooming onto neighborhoods of size $t^{-\beta}.$
 The next development is due to Tao\cite{T2004}. Tao introduced a Fourier integral operator, the symbol of which is supported where $x\cdot \xi$ is positive or negative:
  The version introduced by Tao \cite{T2004} is similar to to the construction below:
\begin{definition}
For any function $f\in L^1_{loc}(\R^3)$, if $f$ is radial, then we denote $\mathcal F f$ as
\begin{align}
\mathcal F f(\rho)=
\int_0^{+\infty}\Big(e^{-2\pi i \rho r}-e^{2\pi i \rho r}\Big)f(r) \,dr, \label{def:Fourier}
\end{align}
in the distribution sense.
\end{definition}

Now inspired by Tao \cite{T2014}, we define the outgoing and incoming components for the radial function.
\begin{definition}\label{def:outgong-incoming}
Let the function $f\in L^1_{loc}(\R^3)$ be radial with $f(0)=0$,
We define the outgoing component of $f$ as
$$
f_{+}(r)=\int^{+\infty}_{0} e^{2\pi i \rho r} \mathcal F f(\rho)\,d\rho;
$$
the incoming component of $f$ as
$$
f_{-}(r)=-\int^{+\infty}_{0} e^{-2\pi i \rho r}\mathcal F f(\rho)\,d\rho.
$$
Moreover, fixing $\rho_0>0$, we define the ``strict'' outgoing component of $f$ as
$$
f_{++}(r)=\int^{+\infty}_{0} e^{2\pi i \rho r}\chi_{\ge \rho_0}(\rho) \mathcal F f(\rho)\,d\rho;
$$
correspondingly, the ``strict''  incoming part of $f$ as
$$
f_{--}(r)=-\int^{+\infty}_{0} e^{-2\pi i \rho r}\chi_{\ge \rho_0}(\rho) \mathcal F f(\rho)\,d\rho;
$$
\end{definition}
From the definitions, we have
$$
f(r)=f_+(r)+f_-(r).
$$

This is the version used in \cite{beceanu2021large}. There are other versions, essentially dealing with the behavior at $r=0.$ 
 One other version is also based on functions of the Dilations operator, but has very different properties. It is a version that we will use to construct Propagation Observables.
 First we note that a function of $A$  can be represented in terms of an integral against the group generated by $A:$
 \eq
F(A) =\int \hat F(a) e^{iaA} da.
 \eeq
 Hence the commutator of functions of $A$ can be read from commuting with the group.
 The action of this group is simple 
$$
  e^{iaA}x e^{-iaA}=e^a x; \quad e^{iaA}p e^{-iaA}=pe^{-a}.
 $$
 Therefore, commuting with $F(A)$ requires that $e^a\hat F(a)$ be integrable in $a.$ Therefore $F$ should be analytic.
 In \cite{S2011} Analytic functions of $A$ are used to microlocalize.
 In particular the projections on incoming/outgoing are constructed as follows:
 Since $A$, the dilation generator, is a self-adjoint operator, we can construct the operator $F(A / R)$ :

\begin{equation*}\tag{3.1}
F\left(\frac{A}{R}\right) \equiv \tanh \frac{A}{R} 
\end{equation*}

by the spectral theorem.

We show that $F(A / R)$ has a positive commutator with $H=-\Delta$, and find lower bounds for it, if $R$ is sufficiently large.

Then, this can be extended to $H=-\Delta+V$ for certain classes of potentials $V$.

Note that the analysis works in any dimension, and we specify to one dimension, which is the more difficult case.

To proceed, recall the Commutator Expansion Lemma( Sigal-Soffer)

Let

$$
a d_{A}^{n}(B) \equiv\left[a d_{A}^{n-1}(B), A\right] ; \quad a d_{A}^{1}=[B, A] .
$$

Lemma 3.1. Commutator Expansion Lemma

\begin{align}
& i[B, f(A)]=\int \hat{f}(\lambda) e^{i \lambda A}\left[e^{-i \lambda A} B e^{i \lambda A}-B\right] d \lambda  \tag{3.2}\\
& =f^{\prime}(A) i[B, A]+\frac{1}{2 !} f^{\prime \prime}(A) i[[B, A], A]+\cdots R_{n}
\end{align}

\begin{equation}
R_{n}=\frac{1}{n !} \int \hat{f}(\lambda) e^{i \lambda A} \int_{0}^{\lambda} e^{-i s A} \int_{0}^{s} e^{-i \mu A} \ldots \int_{0}^{t} e^{-i u A}(-\imath)^{n} a d_{B}^{n}(A) e^{+i u A} d u \ldots d \lambda 
\end{equation}

\section{Functions of $A$}
\begin{theorem}
Let $A$ be the dilation generator as defined before, on $L^{2}\left(\mathbb{R}^{n}\right)$.

$$
\text { For } \quad R>2 / \pi
$$
we have:
$$
\tanh A / R: D(-\Delta) \rightarrow D(-\Delta)
$$
\end{theorem}
\begin{proof} Commuting $\Delta$ through $e^{i \lambda A / R}$, we have:

$$
e^{i \lambda A / R}\left[\Delta, e^{-i \lambda A / R}\right]=e^{i \lambda A / R} \Delta e^{-i \lambda A / R}-\Delta=\left(e^{-2 \lambda / R}-1\right) \Delta: D(\Delta) \rightarrow L^{2}
$$

Therefore, using the Commutator Expansion Lemma with $n=1$, and the property (3.6) of the Fourier Transform of the tanh function, the result follows.
\end{proof}
\begin{theorem}
 $i[-\Delta, \tanh (A / R)]=2 p g^{2}(A / R) p \geq 0$, for $R>2 / \pi$. Here,

$$
g^{2}(A / R)=\frac{\sin (2 / R)}{\cosh \frac{2 A}{R}+2 \cosh \frac{2}{R}}
$$
\end{theorem}
 These are two typical results:  It is possible to modify by exponentially small corrections at infinity the projection $P^{-}(A),$ so that the solution given by the Schr\"odinger flow, decays monotonically on its range for  repulsive potentials:

\eq
\left\langle\psi(t), F_{M}^{-}(A) \psi(t)\right\rangle \downarrow 0, \text { as } t \rightarrow+\infty
\eeq
and

\eq
\int_{0}^{T}\left\|\langle x\rangle^{-1} F(A \leq-M) \psi(t)\right\|^{2} d t \leq\left\langle\psi(0), 2 F_{M}^{-}(A) \psi(0)\right\rangle.
\eeq

 The first part shows that the flow from incoming waves to outgoing is monotonic, with no restriction on the initial data! The second estimate shows, that at least locally in space, the incoming part is controlled, integrably in time, by the size of the incoming waves part of the initial data. So, in particular, no incoming wave can reappear locally, including zero energy and high energy contributions.

 It should be noticed that the notion of incoming/outgoing is used relative to a point in space.
 But, it maybe of interest to have microlocalizations with respect to a line or a domain.
 Two such examples come to mind.  A variant of the Morawetz type estimates in \cite{SX2016}, where the repulsiveness relative to a line is used.

 The other example has to do with numerical schemes. The seminal work of Enquist and Majda (1977) used the decomposition of the waves (of the wave equation) to incoming/outgoing on the boundary of the domain of the computation. Then, by eliminating the outgoing waves, reduced the reflected waves from the boundary into the domain of the computation. This idea is the fundamental tool behind (ABS) absorbing boundary conditions in numerical schemes until these days.
 
 In the case of the wave equation, the incoming/outgoing projection can be done locally, using differential operators.
For the Schr\"odinger case, it is not possible. A new, general approach to doing such 
 decompositions was introduced in \cite{soffertime} and cited references.
 This approach allows decomposition w.r.t. to a domain; a wave in any compact domain is expended in terms of a unique mother function, typically a Gaussian, and its translates and a phase shift $e^{iv\cdot x}.$
 Then each such term is moving under the free flow from its localization in space, with velocity $v.$
 In this way, it is easy to determine for any such term if it is incoming/outgoing relative to a domain.
 
 \subsection{Applications of Incoming/outgoing decpomsitions- I}

 Besides the application to numerical schemes mentioned above, these decompositions are used in the study of scattering and global existence of dispersive PDEs.
 Enss  and later followed by the works on Nonlinear equations by Tao, Kenig-Merle and many others, used the following key idea: Any part of the solutions located far away, can be decomposed into an incoming and outgoing parts.
 In a perturbative argument, via Duhamel representation of the solution, one then uses that under the free flow, the outgoing part move further away. The incoming part on the other hand, by going backward in time, shows that it came form far away. But, knowing that there was nothing far away earlier in time, controls the size of the incoming part.

 In contrast, by Mourre's method, one proves decay estimates by showing directly a-priory estimates that hold for the Fully Interacting dynamics. These are derived from the key a-priory estimate, known as the Mourre estimate:
 Given a self ajdoint operator $H$ that generates the flow, suppose that one can find an operator $A=A^*$
 such that the {\bf Mourre Estimate} holds:
\begin{align}\label {MOu}
 &E_Ii[H,A]E_I \geq \theta E_I^2 +K\\
&E_I\equiv \chi(H\in I); \quad I\subset \R; \quad \theta >0.\\
& K \text{ is a compact operator,}\\
&H \in C^k(A), \quad k\geq 1.
\end{align}

 Here $C^k(A)$ is the space of operators which have a resolvent that under the action of the group generated by $A$
 is a $C^k$ function (in the strong sense), as a function of the group parameter.
 
 Then \emph{Mourre's method}\cite{M1981,amrein1996commutator,FH} applies to get the standard Local decay estimates for $H,$ as well as the decay of the solution on the support of $P^{-}(A).$
 Later, by time dependent Propagation Estimate, the Mourre estimate was used to get the optimal propagation estimates for the flow generated by $H$, in particular the Minimal and Maximal Velocity bounds for the flow\cite{SS1988,HSS1999}.

 These methods do not work for time dependent potentials and nonlinear equations, since one can not localize the energy away from zero or other bad points. Furthermore, unless the interaction terms are localized in space, it is not known where the interactions vanish. Finally, commuting and other operations of such functions of operators with general nonlinear interaction terms are difficult to do.

 However, the new approach to be described, is based on proving  Propagation Estimate for the full flow for general equations, also nonlinear. I expect some of the constructions used in the linear multichannel scattering in the past, may be applied in this more general situations.
 An interesting estimate needed for solving the three body problem, to demonstrate the usefulness of such tools is  the following:
It is originally based on a beautiful construction due to Mourre.

Suppose we know the interaction vanishes fast in a direction in space, given by $x_0 \in \R^N.$
We want to prove that for large $x$ in the cone around $x_0$, if the momentum (frequency) is localized in a direction $p_0$ pointing out of the cone, the solution vanishes as time goes to infinity.

The idea of Mourre was to construct a \emph{conjugate operator} $A_h=A-h\cdot x$, which has a positive commutator on the energy shell where $E=p^2.$ $h$ is a vector in space, should be chosen such that the following condition is also satisfied: $j_{x_0}(x)J_{p_0}(p\sim p_0)A_h \leq -\delta |x|,$ for large $|x|$, and all the quantities are now defined by their symbols.

Then the result stated above follows from noting that:

$$
j_{x_0}(x)J_{p_0}(p\sim p_0)e^{-iHt}\psi(0)=j_{x_0}(x)J_{p_0}(p\sim p_0)[P^+(A_h)+P^-(A_h)]e^{-iHt}\psi(0),
$$
and
$$
j_{x_0}(x)J_{p_0}(p\sim p_0)P^+(A_h)e^{-iHt}\psi(0)= \mathcal{O}(x^{-n})e^{-iHt}\psi(0)\in L^2_{x,t},
$$ 
and 
$$
j_{x_0}(x)J_{p_0}(p\sim p_0)P^-(A_h)e^{-iHt}\psi(0) \in  L^2_{x,t}.
$$
Both decay estimates follow the PRES derived by the Mourre's method or \cite{HSS1999}.
The tricky part is to show that the product of the localization in the cones kills $P^+(A_h).$
It is clear on the symbol level, but the projection is not a PDO. So a very different proof is required.

A similar situation is encountered if one tries to microlocalize the functions of $\gamma.$
This is generally treated by the \emph{localization Lemma} of \cite{SS1987}.
Basically the idea is to first localize the operator itself (rather than the function of the operator). This is possible by standard methods if $A_h$ or $\gamma$ is a standard PDO.
Then one uses the methods similar to proving the Commutator Expansion Lemma to control the difference
$$
j(x)J(p)\left[F(j(x)J(p)A_h J(p)j(x))-F(A_h)\right]J(p)j(x).
$$

Since in the general scattering problem we cannot localize the energy, and we do not have local decay estimates to begin with, one has to modify this argument to become an exterior propagation estimate. 

\section{The General Scattering Theory: Construction of Channels Decomposition}

The first key step in the new general approach we have, is to show in great generality that one meta-theorem of scattering theory holds:
Solutions of Dispersive wave equations which are global, decompose into a free wave and a remainder; this decomposition is the asymptotic state of the system, the limit is approached in the strong sense.
This requires a new strategy, that applies to time dependent potentials. 

\subsection{The Free Channel Wave Operator}

Recall the "Physicist Circular Argument". It does {\bf not} use the type of the interactions, only that the interaction vanishes fast enough under a free flow.
We approach this problem by trying to directly construct the free (free in the sense of Non-Interacting) wave of the asymptotic solution.

Let us then consider a system with dynamics $U(t)$ which may be linear or not, acting on initial data $\psi(0)$,
 which leads to a global solution $\psi(t)=U(t)\psi(0).$
 
In order to construct the free part of the solution at infinity, we now introduce the \emph{free channel wave operator}
 \eq\label{free}
 \Omega_{F_{\pm}}^{*}\psi(0)=s-\lim\limits_{t \to \pm \infty}e^{iH_0 t}J_{free}U(t)\psi(0).
\eeq

The key new idea is now to choose $J_{free}$ the "right way".
We would like to choose $J_{free}$ as a smooth projection on the region of phase-space where the free solution concentrates.
This is a well known construction, as it was realized that to solve long range scattering problem, we need a sharp estimate on this part of the phase-space.

General Propagation Estimates of this type were first proved in the N-body case in \cite{sigal1993asymptotic,sigal1990long}.
But here this microlocalization is used in a very different way.
When the dispersion relation is known, then we have an explicit formula for the group velocity as a function of the wave number, corresponding to the derivative operator.

A classical particle moves under the free flow according to 
$$
x(t)=x(0)+vt.
$$
 So, we are led to use $|x-v(k)t|\leq t^{\alpha}.$
Then, we choose 
 \begin{align}
 &J_{free}=F_c(\frac{|x-2pt|}{t^{\alpha}}\leq 1).\\
 &\omega(k)=k^2, \quad v(k)=2k, \quad v(p)=-2i\nabla_x,\\
 &0< \alpha < 1.
 \end{align}

In the short-range scattering problems $\alpha$ can be chosen small. But it is not possible in general; in particular in the long-range scattering case, $\alpha$ cannot be small \cite{ifrim2022testing,lindblad2021asymptotics,lindblad2023modified,LS2015,LS2021}.
 First we note that if the limit defining the free channel wave operator exists, it captures {\bf all} the free wave at large times.
 This is easy to prove, by showing that on the complement of the support of $F_c$, the limit \ref{free} is weakly convergent to zero.

 Next, to prove the limit above exists, we use Cook's argument:

 Writing $(\Omega_{F_{+}}^{*}-I)\psi(0)=-\int_0^{\infty}e^{iH_0 t}[F_c\mathcal{N}+\tilde{F}']U(t)\psi(0),$
 
 then, we need to prove that the integral converges absolutely.

 Now notice that by the Heisenberg formulation of QM, we have that:
 $$
 e^{i\Delta t}xe^{-i\Delta t}=x+2pt; \quad \partial_t\{e^{i\Delta t}(x-2pt)e^{-i\Delta t}\}=0.
 $$
 Therefore,
 \begin{align}
 &\partial_t G(x-2pt)=0, G \,\text{arbitrary},\\
 &e^{-i\Delta t}G(x)e^{+i\Delta t}=G(x-2pt),\\
 &e^{-i\Delta t}G(x-2pt)U(t)=G(x)e^{-i\Delta t}U(t)
 \end{align}

 From these identities, we derive
 \eq
 e^{-iH_0 t}[F_c\mathcal{N}+\tilde{F}']U(t)\psi(0)=F_1(\frac{|x|}{t^{\alpha}}\leq 1)e^{+iH_0 t}[\mathcal{N}-\frac{\alpha}{t}F_1']U(t)\psi(0).
 \eeq
 Therefore we need to prove the integrability (in norm) of the above expression.
 It consists of two terms. The first one, coming from the interaction term, is bounded by
 \eq
\|F_1\|_{L^2_x}\|e^{+iH_0 t}\mathcal{N}U(t)\psi(0)\|_{L^{\infty}_x}\lesssim t^{n\alpha/2}t^{-n/2}\|\mathcal{N}U(t)\psi(0)\|_{L^{1,s}}.
 \eeq

 This estimate gives the main condition on the interaction $\mathcal{N}.$
In its abstract form it is: $(t\geq 1)$
$$
\|F_1(\frac{|x|}{t^{\alpha}}\leq 1)F_2(\sum_j |p-\tau_j|>t^{-\beta})U_0(-t)\mathcal{N}\psi(t)\|_{L^2_x}\lesssim t^{-1-\epsilon}.
$$
Here the free dynamics is generated by $H_0=\omega(p),$ with thresholds at $\tau_j.$
In three or more dimensions, with the standard Laplacian generating the free flow, we do not need to localize away from the thresholds. If the interaction term in localized in space (with sufficient decay), then one can get the needed estimate also in 1 and 2 dimensions.
Moreover, since we assume the solution is uniformly bounded in $H^1,$ we can allow one derivative in the interaction, on each side. Therefore terms like $-\nabla_ig_{ij}(x,t)\nabla_j$ can also be incorporated with a suitable martix $g_{ij}.$

  This expression is integrable in time if $\alpha$ is sufficiently small, and the dimension $n\geq 3.$
 The number of derivatives $s$ depends on dimension for the Wave equations of the Hyperbolic type, but is zero for the Schro\"edinger type. Also note that the effective dimension $n$ for some Hyperbolic equations is $n-1.$

 For this to hold we only need to know that $\mathcal{N}U(t)\psi(0)$ is uniformly (in time) bounded in $L^1,$
 and in fact a weaker condition is sufficient for integrability.
 This estimate is very general, and it only uses $L^p_x$ estimates on the solution and interaction terms, but not point-wise decay of the interaction at infinity. Therefore it does not require the assumption of spherical symmetry.
 As we see, up to this point everything was done on the basis of the free, non-interacting flow, including the choice of the partition defining the free channel! Just like the "physicist circular argument".

 The issue is of course that the decomposition of the solution to near and far is {\bf not} stable. A wave can come back from far away to the origin.
 To control this rigorously, you may argue that the far part is mostly outgoing, and under the free flow will move farther. The incoming part should be small... In fact this is the line of proof in the methods of Enss, Mourre etc...
 It required Propagation Estimates to control the incoming waves and other estimates for the outgoing waves.
 In the approach above we see that the boundary effect is reduced to controlling the flow through the boundary, given by the support of the function $F_1'.$ There is an extra factor of $1/t$, which gives the (false) impression that it should be easy to control.

 In fact, it is here that we need an a-priory estimate w.r.t. the {\bf full,} interacting  flow.
 This is now done by using the method described before, of proving Propagation Estimates by an appropriate choice/s of Propagation Observables.
 In this case the answer is very simple, we use as Propagation Observable the operator $F_c(\frac{|x-2pt|}{t^{\alpha}}\leq 1)$ itself.

 To this end we compute:

 \begin{align}
&\partial_t (U(t)\psi(0), e^{+i\Delta t}F_1(\frac{|x|}{t^{\alpha}}\leq 1)e^{-i\Delta t}U(t)\psi(0))=\\
&-(\Omega^*(t)\psi(0), \frac{\alpha}{t}F_1'\Omega^*(t)\psi(0))+\\
&2\Re(\Omega^*(t)\psi(0),F_1 e^{-i\Delta t}\mathcal{N}\psi(t)).\\
&\Omega^*(t)\equiv e^{-i\Delta t}U(t).
 \end{align}

 The first term on the RHS is positive. The second term is integrable as we showed in the previous step.
 Since the integral of the LHS is uniformly bounded (by $L^2$ norms), it follows that the $F_1'$ term is also integrable.
 This is then used to control the term we need for proving the existence of the free channel wave operator.

 A few comments are in order:
 While here the computation is done for the case where the free flow is given by the free Hamiltonian $-\Delta,$
 this argument is very general. In fact it is essentially abstract:
 For a given flow $U_0$ for which one has favorable dispersive/local-decay estimates, then the same construction of the Free channel wave operator will apply, using 
 $$
 J_{free}= U_0(t)F_1(|x|/t^{\alpha} \leq 1)U_0(-t).
 $$
In some cases, the dispersive estimates for the free flow needed, only hold if the state is supported away from zero and infinite frequency.
In this case, one can modify the definition of $J_{free}$ by a cutoff of these frequencies, using
$$
F_2(|p|\leq t^{\beta}) G_2(t^{\beta}|p|\geq 1).
$$
Both of these cutoff functions have non-negative derivative w.r.t. time.
The complement is not contributing to the asymptotic free wave, since a free wave (of a Hamiltonian which is a   constant coefficient differential operator) can not  have a part with vanishing frequency. The high energy part contributes a vanishing term as time goes to infinity, provided the solution is uniformly bounded in some Sobolev norm $H^s$ with $s$ positive. This is used for example to treat one dimensional problems, since dispersion is weak in 1 and 2 dimensions.
It is also used for KG and wave type equations, in order to localize away from high frequencies. 

\section{Weakly localized states}
We have described above that quite generally, solutions of dispersive equations decompose asymptotically to a free wave and the rest, in the strong sense of limits.
This means in particular showing the existence of the free channel Scattering Wave Operator. These operators are the basic building blocks of the Scattering Matrix, from which one derives the properties of the system.
In the case the interaction term is a time independent localized function of x, the only other states of the system are the eigenfunctions of the Hamiltonian, if any. They are all localized functions of x, generically decay exponentially.

This is no longer the case in the time dependent/nonlinear cases. Some solutions are localized, like solitons and breathers, but there are also spreading self similar solutions.
The existence of such global solutions with a {\bf finite} $L^2$ mass for Schrodinger type equations, is a critical issue then; it will break the standard notion of Asymptotic Completeness.
As is explained below, the general theory we have now, can only show that the part of the solution that is not asymptotically free, is weakly localized, so it can spread, but slowly. In the linear  time independent case, extra non-trivial arguments are needed to exclude such solutions. 

Consider the case of localized interactions, including the nonlinear case with radial symmetry.
Inspired by the results of Tao in this case, we want to prove that the weakly localized part is \emph{smooth} even if the initial data is not. We want to prove that these solutions are localized in space, uniformly in time.

That would be Asymptotic Completeness: all solutions are free plus localized solutions (in the localized interaction case).
Each such localized solution is also a solution of the equation, asymptotically, and is therefore a soliton of say the NLS with purely nonlinear terms. But there could be more solutions, which are time dependent, like breathers.
While these are non-generic (since they correspond to embedded eigenvalues of the Floquet Operator in the time almost  periodic case), ruling them out, is a separate issue.

In the work \cite{liu2023large}, it s shown that the weakly localized solution is smooth, and moreover, in the mass supercritical cases considered, must have a massive part around the origin for all times, similar to Tao's results.
Furthermore it is shown that the delocalized part, if it exists is very well localized in a thin part of the phase space, essentially self similar solution which can spread no faster than $t^{1/2+0}.$

This is proved by a long sequence of Propagation Observables of a novel type, as they apply to the weakly localized part, rather than scattering states.
The fundamental difference is that now the asymptotic state is given by scaling, rather than the free flow.
So, rather than looking at the region of phase-space where $|x-2pt|$ is small, we now look at the region where $|x|\sim t^{\alpha}.$ Clearly, one can prove that if the momentum is away from zero, there is no propagation in this region.
So, we study the following Propagation Observable;
$$
B=F_1(\frac{|x|}{t^{\alpha}}\geq 1)F_2(t^{\beta}\gamma \geq 1)F_1, \quad  \alpha >\beta.
$$
One also needs to reverse the sign of $\gamma$, but this is usually the easier direction.
It is easy to see that formally, the leading term in the Heisenberg derivative of this operator is non-negative.
The interaction terms are non issue for localized ones, though is does force a lower bound on $\alpha.$

The fundamental problem is coming from the symmetrization terms. Each such term may not be better than the leading term by $t^{-\alpha+\beta}.$
This is not a technical weakness. It must be like that, since it is an indication that there could be solutions who concentrate in the region of phase-space $\{|x|\sim t^{\alpha}\} \bigcap \{|p|\sim t^{-\alpha}.\} $
Note that $\beta <\alpha$ is forbidden by the uncertainty principle.

Since the Symmetrization terms are only better by a factor by $t^{-\alpha+\beta},$ this may not be sufficient to get integrabilty in time of these terms.
We therefore need to use an iteration scheme:
The estimate we get from the above PROB is the following:

$$
D_H B= t^{-\alpha} \sqrt{F_1F'_1}\gamma F_2 \sqrt{F_1F'_1} +\mathcal{O}(t^{-3\alpha+\beta})\tilde{ F_1}F_2\tilde{F_1}+R.
$$
 Here $\tilde{F_1}$ stands for function on the compact support of $F_1^{(j)}, j=1,2,3.$

 The first term is of order $t^{-\alpha-\beta},$ so it is better then the next term by a factor of $2(\alpha-\beta).$
 So, we can redo the same estimate, this time to control the symmetrization term, by the use of the following Propagation Observable:

 $$
 B_1=t^{-2(\alpha-\beta)}B.
 $$
 The new estimate we get controls the symmetrization term of the first estimate by an error which is of the order of (the integral over time) $t^{4(\alpha-\beta)}.$
 Therefore, IF $\alpha>\beta$, this process terminates after finitely many iterations.
 However, there is an important complication in this argument. The support of the derivatives of $F_1$, move slightly backward, and therefore cannot be controlled by the leading term. This means that in the definition of $B_1$, we must change $F_1(a\geq 1)$ to $F_1(a\geq 1-\delta_1).$ Doing it finitely many times is not a problem, but since we must have $\sum_j \delta_j <1/2,$ we are limited in many cases of interest.

 Another issue results from the fact that the interaction terms do {\bf not} live on the support of $F_1'$, but, in some complicated way of the support of $F_1[F_2,N].$
 So it raises the question of how to get estimates on support $F_1$ itself.
 One way this can be done is by using another type of Propagation Observable, where we multiply $B$ by $|x|/t^{1/2 +0}.$
 This adds extra $t^{-1/2+0}$ to the decay of the interaction term, and also give new estimate on the support of $F_1$
 but with a factor $t^{-1/2-0-\beta}.$

The above analysis tells us that we can squeeze the expanding part of the weakly localized states into the region of phase-space where

$$
\{|x|\sim t^{\alpha}\} \bigcap \{|p|\sim t^{-\alpha +0}\} .
$$

As pointed out above, if the Interaction decays too slowly, we may not be able to iterate with a too small value of $\alpha.$ One can use the above to conclude that iteration is possible if the interaction decays faster than $<x>^{-3}.$

Next we observe that for any $\alpha$ that can be controlled integrably, without the factor of $t^{-1/2-0},$ this implies that the free channel wave operator exists with $J_{free}=B.$ Therefore, for such $\alpha, \beta$ this part of the solution converges to a free wave.
On the other hand, all the estimates derived with the extra factor of $t^{-1/2-0}$ do not converge to a free waves, but decay to zero, in some averaged slow way.
This is because such estimates imply that

$$
\int_1^{\infty} \|F_1 \sqrt{F_2\gamma}\psi(t)\|^2 t^{-1/2-0}dt <c<\infty. \quad \beta <1/2.
$$
 Note however that this proof works only for a solution that is purely weakly localized, since we used the boundedness of the factor $ |x|/t^{1/2+0}$ on the state, in the sense of expectation.

Indeed, to further understand the properties of the weakly localized state, we prove that such states are in the domain of the dilation operator $A.$ Tao proved a similar result for all powers of $A.$

\section{Microlocalization via Functions of Dilations}

We saw in a previous section that the analytic functions of $A$ have explicit expression for the commutator with $x, p.$
We will show some relevant and typical applications of the use of functions of $A$ to get Propagation estimates and localization of incoming/outgoing waves.

 We demonstrate how one can  get a-priory estimates similar to Morawetz, but with a different scale, and more general dispersion relations.
First, we consider estimates with the simplest Propagation Observable, $$S_{M,R}=\tanh{\frac{A-M}{R}}, R>>1. $$
\begin{proposition}
\begin{align*}
& i[S_{M,R}, |p|^a]= \\ 
& |p|^{a/2} [\frac{1}{i}\left(\tanh \frac{A+ai/2}{R}-\tanh \frac{A-ai/2}{R}\right)]|p|^{a/2}\\
&=|p|^{a/2}\frac{1}{i} \frac{\sinh (a i / R)}{\cosh \frac{A+ai/}{R} \cosh \frac{A-ai/}{R}}|p|^{a/2} \\
& =|p|^{a/2}\frac{\sin (a / R)}{\cosh \frac{a A}{R}+2 \cosh \frac{a}{R}}|p|^{a/2} \quad \text { for } R>a / \pi .
\end{align*}
\end{proposition}

This proposition implies local decay estimate for Fractional and other power Hamiltonians, slightly better than the standard Mourre's method result, since instead of localization of energy away from zero, we only need a factor of $|p|^{a/2}$ to cutoff the zero frequency.
Adding a repulsive potential will not change the result, provided it is repulsive \emph{relative to the above Propagation Observable }.
Unlike the Morawetz estimate, this Propagation Observable is bounded on $L^2$ and works in any dimension.

Next, we consider the problem of commuting with a potential which is a function of $x$.
The generic formula makes sense for potentials which are $\epsilon$ level dilation analytic:

\begin{theorem} 

Let $V(x)$ be dilation analytic for all $|s| \leq \beta$. Then

$$
i[V, \tanh A / R]=\frac{+i}{2 \cosh A / R}\left\{V^{[-\beta]}-V^{[+\beta]}\right\} \frac{1}{\cosh A / R}
$$

where $\beta=1/R,$

$$
V^{[\beta]} \equiv e^{\beta A} V e^{-\beta A}=V\left(e^{-i \beta} x\right) .
$$
\end{theorem}
 Further constructions which are relevant are the projections on\newline incoming/outgoing waves, using the above function of $A.$
 Similar Propagation Estimate holds for such projections, in particular an outgoing wave remains outgoing in this generalized sense, without corrections.
 Furthermore, the free flow is \emph{monotonic} decreasing (in size) when projected on the incoming waves.
 \cite{S2011}.
\subsection{The Strategy}
 Many new Propagation Estimates can now be derived, by combining microlocalization of standard observables, like $|x|, |p|, \gamma, A, ..$ with appropriate projections on exterior domains of various sizes. The domain sizes scale  with time or with a large parameter used dyadically.
 
 Then, a key new step is the process of iteration of the Propagation estimates. By this we mean that we control the error terms in an estimate (since it is not $L^1$ in time), using the leading order. This is not always possible. But with further (second) microlocalization of the leading term, it may be possible to control the error terms in some parts of the phase space.
 If this procedure pushed to the limit, there is a leftover thin sets of the phase-space where nothing works.
 In fact,in these thin sets solution concentrates! In some sense one can find the solution of the complicated system, by finding all parts of the phase-space where the solution decays with time.
\subsubsection{Examples}    
 The important examples of Propagation Observables that are used in the above mentioned works are generally constructed by a product of phase-space operators.
 For example,  the \emph{ Exterior Morawetz Estimate} is derived from the following class of Propagation Observables:
 ( $C.C.$ stands for complex conjugate)
 \begin{align}
& F_1(\frac{|x|}{R}\geq 1)\gamma_0 F_1(\frac{|x|}{R}\geq 1), \quad 2\gamma_0=(-ix/|x|)\cdot \nabla_x +C.C.\\
& F_1(\frac{|x|}{R}\geq 1)\gamma_1 F_1(\frac{|x|}{R}\geq 1), \quad 2\gamma_1=(-ig(x)x)\cdot \nabla_x +C.C.\\
& g(x)=1/|x|, |x|>2, g(x) \text{ smooth bounded vanishing at the origin.}\\
& F_1(\frac{|x|}{R}\geq 1)\gamma_1 F_1(\frac{|x|}{R}\geq 1), \quad 2\gamma_2=(-ig(x)x)\cdot \nabla_x +C.C.\\
& xg(x)=\nabla G(x) , \quad -\Delta^2 G(x) \geq c<x>^{-3-0}, \text{Dimension} \, N\geq 3.
\end{align}
One also uses $t^{\alpha}$ instead of $R.$
The resulting estimates from the above have error terms that come from the Interaction term, as well as error terms that come from symmetrization of the leading order term.
The interaction terms are assumed to decay fast enough at infinity, so we get errors of order $R^{-k}.$
The symmerization terms are of order $R^{-3}.$ The leading order terms are of order $(1/R)F_1'\gamma^2.$
While we might hope to be able to iterate by controlling the nonlinear term in terms of the leading term, this cannot be done for the symmetrization term, as $\gamma$ can concentrate at frequency below $1/R.$
 By multiplying the whole estimate by $M^3/T^2$ and double integrating both sides from $1 \to T$ provides a useful bound.
 We also use that on support $F_1',$  we have that $A\sim M\gamma.$

Using this kind of  estimates on a weakly localized state, with the property that $$<\psi_{wb}(t), |x|\psi_{wb}(t)>\leq c\sqrt t,$$
we obtain a bound of the form
\begin{align}
& \frac{1}{T^2}\int_{T/2}^T \int_{T/2}^t \|\sqrt{F_1F_1'}A\psi_{wb}(s)\|^2_{L^2_x} ds dt\\
& \leq \frac{1}{T^2}\int_{0}^T \int_{0}^t \|\tilde F_1(<x>/M)\psi_{wb}(s)\|^2_{L^2_x} ds dt +
\frac{M^3}{T^{3/2}}+\frac{M^2}{T}+ M^{3-k}.
\end{align}

Now we use a dyadic sum: $M_j\equiv 100\, 2^j,$ and let $M_{j_0-1}<T<M_{j_0}$
and sum the above estimate over all $j\leq j_0.$
The result is 
$$
\frac{1}{T^2}\int_{T/2}^T \int_{T/2}^t \|A\psi_{wb}(s)\|^2_{L^2_{(|x|\leq \sqrt T)}} ds dt \lesssim 1.
$$

Therefore, there is a sequence  of times $s_n,$ on which $\|A\psi_{wb}(s)\|^2_{L^2_{(|x|\leq \sqrt T)}} \lesssim 1.$
In the difficult case of one-dimensional nonlinear scattering, there are only very partial results for non-localized (e.g. purely nonlinear) interaction terms \cite{duyckaerts2024profile,soffer2024scattering}.

Such estimates, and more general ones,  show that the weakly localized solution concentrates in to thin regions of phase-space, where $|x|\sim t^{\alpha}, |p|\sim t^{-\alpha}.  $

\subsubsection{Iterations with Functions of $A$}
One can prove estimates as above and other ones, by microlocalizing on incoming/outgoing waves separately.

Typically, the estimates on the incoming waves are stronger. Following \cite{liu2023large} and submitted work,
a typical important example is the following localization:

\eq\label{outgoing prob}
B\equiv F_1(\frac{|x|}{t^{\alpha}}\geq 1)P^+_{M,R}(A)F_1(\frac{|x|}{t^{\alpha}}\geq 1).
\eeq
Here $P^+$ is a projection on outgoing waves, analytic in $A$ and projects on $A>M.$
It is easy to see that the commutator of the Laplacian gives a positive commutator, to leading order, plus symmetrization terms.
Suppose we choose $\alpha$ large enough such that the nonlinear terms drop.
The symmetrization terms come from commuting functions of $\frac{|x|}{t^{\alpha}}$ with functions of $A$.
Since the commutator $[A,x]=c x,$ we do not gain a factor of $t^{-\alpha},$ only a factor of $A$ and a factor of $1/R.$
Therefore, we need to iterate such estimates in order to gain factors of $1/R.$
However, the leading order term controls expressions of the form
$$
\int \| P^+_{M,R}(A)F_1'\psi(s)\|^2 ds/s^{\alpha}+ (1/R)\int \|P^{+'}_{M,R}(A)(A\sim M)F_1 p\psi(s)\|^2 ds.
$$
Here $P^{+'}$ stands for the function we get by commuting the Laplacian with $P^+.$ $p=-i\nabla.$

This estimate is very strong; in particular the second integral has no time decay factor.
We would like to iterate this. The problem is that these expressions are not of the form that can control the interaction term, since the interaction terms have the factor $P^+$ only on ONE side.
We are then faced with the issue of commuting $P^+(A)$ through functions of $x.$ 

Since the interaction terms are arbitrary functions of $x$ or nonlinear, it is not clear why such commutators will be small.
In fact, this is a general problem with microlocalization of nonlinear dispersive equations.

We deal with this problem using two tools:
First, it is easy to verify that we get good estimates if the interaction term is a Dilation Analytic Function.
In this case, we can use:

$$
F(A\leq R)V(x) F(A\geq mR)=
$$

$$
F(A\leq R)e^{+\theta A}V_{\theta}(x) e^{-\theta A}F(A\geq mR)\leq \mathcal{O}(e^{-mR+R}).
$$
 Therefore, we see that if we can commute powers of $A$ through $V$ we are good.
 Next, we note that if a function $f\in L^2,$ then $F(|A|<R)f$ is a "good" function, since
 $$
 [A,F(|A|<R)f]=F(|A|<R)(Af) \sim RF(|A|<R)f.
 $$
Hence
$$
F(A\leq R)(F(|A|<R')f) F(A\geq mR) \lesssim (R+R')^l (mR)^{-l},
$$
by repeatedly commuting the factor $(A+i).$
The second tool we need is to do a high/low decomposition of the Interaction terms. For this, we need to have the nonlinear term be an analytic function in the two variables $\psi, \psi^*.$
The high/low decomposition is done with respect to the spectrum of $A$ ( or $\gamma$).

To this end we break any factor of $\psi:$

$$
\psi=  (F_l(A<-R)+F_h(A>R) +F_0(|A|\leq R))\psi.
$$
Suppose we have good estimates for all parts, and we want to iterate the high estimate.
Then, we need to control expressions of the form 
$$
F_h (\psi^{\#})^k\psi.
$$
By the above decomposition all terms will have at least one factor of $F_h\psi$, except one term
in which all $\psi$ are low.

We need to show that such a term is higher order, due to the projection $F_h$ on the left.
For this we need a new lemma, that states that a product of two functions with localized $A$ is again localized.

In the frequency domain this is trivial, as it follows from $e^{ikx} e^{ik'x}= e^{i(k+k')x}$ where these exponentials are the characters of the derivative operator.

The characters of $A$ are known (on $\R^N$) and they are of the form $|x|^{-a+ik}.$
Therefore the product is not localized, but one can prove it decays fast away from $k+k'.$\cite{liu2023large} and followup work.
Then, Since the product of $k$ low terms is around $kR,$ in each iterate we need to decompose with $R' <<R/k.$
So, the number of possible iterations is limited, usually by orders of log or log-log of the initial $R.$
But this turns out to be sufficient.

In order to controls analytic functions,\newline  we control the generic object $(\lambda+ \psi^* \psi)^{-1}.$
For details of the above analysis see \cite{liu2023large}.
\subsection{Microlocalization of Weakly Localized States}

The method of proof of the existence of the free channel wave operator does not say what is the nature of the leftover.
Other estimates imply that such solutions cannot spread faster than $x\sim \sqrt t,$ at least when the interaction terms decay fast enough.
To get better results, we can use propagation estimates tailored to such states, if they exist.
A key starting point is the estimate that $A$ and powers of $A$ are bounded on WLS (weakly localized states).
One example is the following generalization of \ref{outgoing prob}:

$$
B\equiv F_1(\frac{|x|}{t^{\alpha}}\geq 1)A^2P^+_{M,R}(A)F_1(\frac{|x|}{t^{\alpha}}\geq 1).
$$
 This operator is unbounded, however we know that at least on a sequence of times, its expectation on a WLS is uniformly bounded.
 Therefore we will get an estimate that is specific to WLS.
 In this case we will need to iterate in order to control the lower order terms that come from symmetrization.

 Another method that is used for WLS is based on exploiting the property of boundedness of $|x|/\sqrt t$ on such states.
 If $B$ is a good Propagation Observable, then so is $(|x|/\sqrt t)B+B(|x|/\sqrt t).$
 The advantage of this is that Propagation Estimates on support $F_1'$ are upgraded to hold on the support $F_1,$ at a price of having to use a weight $t^{-1/2}.$
 However, this is necessary if we want to iterate estimates to arbitrary small $\alpha.$ See \cite{liu2023large},
 where it is shown that the weakly localized states may propagate only near $|p|\sim t^{\alpha}, \, \alpha>0$ 
 in the domain where $|x|\sim t^{\alpha}.$
 
 \section{Klein-Gordon equations}
 
 We now show how to adapt the methods above to Klein-Gordon type equations.
 The first natural step is to reformulate the hyperbolic type equation, as a first order equation in time, and with a conservative dynamics.
 For this we need first to choose the proper Hilbert Space.
 
 In space dimension $n\geq1$, we consider the following nonlinear Klein-Gordon equations:
\eq\tag{KG}
\begin{cases}
(\square+1)u(x,t)=\mathcal{N}(u,x,t)u(x,t)\\
\vu(0):=(u(x,0),\dot{u}(x,0))=(u_0(x), \dot{u}_0(x))\in \mathcal{H}
\end{cases}, \quad (x,t)\in \mathbb{R}^n\times \mathbb{R}, \label{KG}
\eeq
where $P:=-i\nabla_x$, $\mathcal{H}:=H^1(\mathbb{R}^n)\times L^2(\mathbb{R}^n)$ denotes a Hilbert space equipped with an inner product $(\cdot,\cdot)_{\mathcal{H}}$ given by:
\eq
(\vec{u}, \vec{v})_{\mathcal{H}}:=(\langle P\rangle u_1, \langle P\rangle v_1 )_{L^2_x(\mathbb{R}^n)}+( u_2,  v_2 )_{L^2_x(\mathbb{R}^n)}\quad \text{ for all }\vec{v}, \vec{u}\in \mathcal{H}.
\eeq
Here, the space dimension $n\geq 1$. We define $\| \vu\|_{\mathcal{H}}^2:=(\vu,\vu)_{\mathcal{H}}$ and write the solution $\vec{u}(x,t)$ as $\vec{u}(x,t)=(u(x,t),\dot{u}(x,t))$. Here, $\Box:=\partial^2_t-\Delta_x$ and $\mathcal{N}$ represents the interaction, the specifics of which will be detailed later.

 \subsection{Free KG equations}

Let $\vu_0(t):=(u_0(t), \dot{u}_0(t))$ be a global solution to a free KG equation
\eq
\begin{cases}
(\square+1)u_0(t)=0\\
\vu_0(0)=\vu(0)=(u(x,0),\dot{u}(x,0))\in \mathcal{H}
\end{cases},\quad (x,t)\in \mathbb{R}^n\times \mathbb{R}.\label{Pfree}
\eeq
Let $H_0:=-\Delta_x$. $u_0(t)$ and $\dot{u}_0(t)$ have following representation
\eq
u_0(t)= \cos(t\sqrt{H_0+1})u(0)+\frac{\sin(t\sqrt{H_0+1})}{\sqrt{H_0+1}}\dot{u}(0)\label{free:rep1}
\eeq
and
\eq
\dot{u}_0(t)=-\sin(t\sqrt{H_0+1})\sqrt{H_0+1}u(0)+\cos(t\sqrt{H_0+1})\dot{u}(0).\label{free:rep2}
\eeq
Let
\eq
\mathcal{A}_0:=\begin{pmatrix} 0& -1\\H_0+1&0\end{pmatrix}.
\eeq
\eqref{Pfree} is equivalent to
\eq
\partial_t[\vu_0(t)]=-\mathcal{A}_0\vu_0(t).
\eeq
Therefore, $\vu_0(t)$ can be rewritten as
\eq
\vu_0(t)=e^{-t\mathcal{A}_0}\vu(0),
\eeq
that is,
\eq
U_0(t,0)=e^{-t\mathcal{A}_0}.\label{Ufree}
\eeq
One uses the following standard dispersive decay estimate for the KG propagator, see for instance H\"ormander \cite{hormander1997lectures}(Corollary 7.2.4) for a proof.
\begin{lemma}\label{dlem1} For all $t\in \mathbb{R},$ 
\eq
\| e^{\pm i t\langle P\rangle}f\|_{\s^\infty_x(\mathbb{R}^n)}\lesssim \frac{1}{\langle t\rangle^{n/2}}\|\langle P\rangle^{\frac{n+3}{2}} f \|_{\s^1_x(\mathbb{R}^n)}.
\eeq
\end{lemma}
 We define
 \eq
\mathcal{F}_\alpha(x,P,t)=\begin{pmatrix}
 \frac{1}{\langle P\rangle} F_\alpha(x,P,t)\langle P\rangle &0 \\
 0 &  F_\alpha(x,P,t)
 \end{pmatrix}.
 \eeq
 Here, $F_c$ and $F_1$ denote smooth cut-off functions and $F_\alpha$ satisfies:
 \begin{enumerate}
\item When $\mathcal{N}$ is a localized interaction and $n\geq 1$, $F_\alpha$ is defined as
\eq\label{Falpha1}
     F_\alpha(x,P,t):=  F_c(\frac{|x|}{t^\alpha}\leq 1)F_1(t^{\beta}|P|\geq 1)F_1(|P|\leq t^{\beta})\quad  
 \eeq
 for $\alpha\in (0,1)$ and $\beta\in (0, \min\{1-\alpha, \frac{\sigma-1}{\sigma}\})$.
 \item When $N$ is a non-local interaction and $n\geq 3$, $F_\alpha$ is defined as
 \eq\label{Falpha2}
     F_\alpha(x,P,t):= F_c(\frac{|x|}{t^\alpha}\leq 1)F_1(|P|\leq t^{\beta})  
 \eeq
 for $\alpha\in (0,1-2/n)$ and $\beta\in (0,\frac{n}{n+3}(1-\alpha)-\frac{2}{n+3})$.
 \end{enumerate}

\begin{lemma}[Dispersive estimates for free flows]\label{lem: free decay}Let $F_\alpha$ be as in \eqref{Falpha1}, \eqref{Falpha2}, depending on whether $\mathcal{N}$ is a local or non-local interaction, respectively. The following estimates hold true: $\sigma >1$
\begin{enumerate}
     \item When $\alpha\in (0,1-2/n)$ and $\beta\in (0, \min\{1-\alpha, \frac{\sigma-1}{\sigma}\})$, the weighted $L^2$ estimates hold:
 \eq\label{decay: eq1}
 \| F_\alpha(x,P,t)e^{\pm it\sqrt{H_0+1}}\langle x \rangle^{-\sigma} \|_{L^2_x(\mathbb{R}^n)\to L^2_x(\mathbb{R}^n)}\in L^1_t[1,\infty)
 \eeq
 and
  \eq\label{decay: eq3}
 \| \partial_t[F_1(|P|,t)]F_c(\frac{|x|}{t^\alpha}\leq 1)e^{\pm it\sqrt{H_0+1}}\langle x \rangle^{-\sigma} \|_{L^2_x(\mathbb{R}^n)\to L^2_x(\mathbb{R}^n)}\in L^1_t[1,\infty),
 \eeq
 where 
\eq
F_1(|P|,t):=F_1(t^{\beta}|P|\geq 1)F_1(|P|\leq t^{\beta}).
\eeq
  \item  When $\alpha \in (0,1-2/n)$ and $\beta\in (0, \frac{n}{n+3}(1-\alpha)-\frac{2}{n+3})$, the $L^1$ estimate holds:
  \eq\label{decay: eq2}
  \| F_\alpha(x,P,t) e^{\pm it\sqrt{H_0+1}}\|_{L^1_x(\mathbb{R}^n)\to L^2_x(\mathbb{R}^n)}\in L^1_t[1,\infty)
  \eeq
  and
    \eq\label{decay: eq4}
 \| \partial_t[F_1(|P|\leq t^\beta)]F_c(\frac{|x|}{t^\alpha}\leq 1)e^{\pm it\sqrt{H_0+1}}\langle x \rangle^{-\sigma} \|_{L^2_x(\mathbb{R}^n)\to L^2_x(\mathbb{R}^n)}\in L^1_t[1,\infty).
 \eeq
 \end{enumerate}
\end{lemma}

 We can now define the Free Channel Wave Operators, the general structure is similar to the Schr\"odinger case. However, note that the dispesrive estimates fail for the free flow at high frequency. We do not want to assume that the solution is uniformly bounded in $H^s$ for large $s.$
 The flexibility of the notion of Channel Wave Operators is then used, by modifying the projection on the channel such that the phase space region with low and high frequency are removed.
 \subsection{Free channel wave operators}\label{subsecL free channel}
 
 In this section we introduce the key notion for separating the free solution from $\vu(x,t)$: the free channel wave operators. 

We define a free channel wave operator~$\Omega_{\alpha}^*$ as follows: 
 \begin{equation}\label{eq: subsecL free channel, eq 1}
 \Omega_{\alpha}^*\vu(0) := s\text{-}\lim_{t\to \infty} \mathcal{F}_\alpha(x,P,t) U_0(0,t) \vu(t)\quad \text{ in }\mathcal{H},
 \end{equation}
 where $P:=-i\nabla_x$ and $\mathcal{F}_\alpha$ denotes a matrix operator:
 \eq
\mathcal{F}_\alpha(x,P,t)=\begin{pmatrix}
 \frac{1}{\langle P\rangle} F_\alpha(x,P,t)\langle P\rangle &0 \\
 0 &  F_\alpha(x,P,t)
 \end{pmatrix}.
 \eeq
 Here, $F_c$ and $F_1$ denote smooth cut-off functions and $F_\alpha$ satisfies:
 \begin{enumerate}
\item When $\mathcal{N}$ is a localized interaction and $n\geq 1$, $F_\alpha$ is defined in \eqref{Falpha1}
 for $\alpha\in (0,1-2/n)$ and $\beta\in (0, \min\{1-\alpha, \frac{\sigma-1}{\sigma}\})$.
 \item When $\mathcal{N}$ is a non-local interaction and $n\geq 3$, $F_\alpha$ is defined in \eqref{Falpha2}
 for $\alpha\in (0,1-2/n)$ and $\beta\in (0,\frac{n}{n+3}(1-\alpha)-\frac{2}{n+3})$.
 \end{enumerate}

 Once the limits defining the wave operators are shown to exist, we note that the cutoffs on the frequency can be removed by continuity, assuming the solution is uniformly bounded in $H^1.$

 We make the following general assumptions on the interaction term:

 \begin{assumption}\label{asp: 2}( $n\geq 1$)  {\bf Local interactions}. For some $\sigma>1$, by writing $\vec{u}(t)=(u(t),\dot{u}(t))$, we assume $\langle x\rangle^\sigma N u(t)$ remains bounded in $L^2_x(\mathbb{R}^n)$ globally and uniformly in time:
\eq
\sup\limits_{t\geq 0}\| \langle x\rangle^\sigma u\mathcal{N}(t) \|_{L^2_x(\mathbb{R}^n)}\lesssim_E1.\label{eq100}
\eeq

\end{assumption}
\begin{assumption}\label{asp: 3}( $n\geq 3$) {\bf Non-local interactions}. $\mathcal{N}u$ remains bounded in $L^1_x(\mathbb{R}^n)$  globally and uniformly in time:
\eq
\sup\limits_{t\geq 0}\| \mathcal{N}u(t)\|_{L^1_x(\mathbb{R}^n)}\lesssim_E 1.
\eeq
\end{assumption}

{\bf Examples include}:

\begin{example}[Local interactions]
Typical examples of localized interactions are
\eq
\mathcal{N}(u,x,t)= V(x,t)u+a(x)u^2+b(x)u^3,  \quad \text{ in $1$ dimension},
\eeq
where $\langle x\rangle^\sigma V(x,t)\in L^\infty_{x,t}(\mathbb{R}^{2}) $ and $ \langle x\rangle^\sigma a(x), \langle x\rangle^\sigma b(x)\in L^\infty_x(\mathbb{R}^1)$ for some constant $\sigma>1$. 
\end{example}
\begin{example}[Non-local interactions]Typical examples of non-local interactions are
\eq
\mathcal{N}(u,x,t)= V(x,t)u+\lambda u^3+\lambda' u^4 ,  \quad \text{ in $3$ or higher dimensions},
\eeq
where $V(x,t)\in L^\infty_tL_x^2(\mathbb{R}^{3+1})$ and $\lambda, \lambda'\in \mathbb{R}$. 
\end{example}
\begin{example}[Charge transfer interactions] More generally, one can control
\eq
(\square+1+V(x,t))u(t)=f(u)u(t), \quad  1+V(x,t)\geq v_0>0,
\eeq
with
\eq
\sup_t\|f(u)\|_{L^2_x} < \infty.
\eeq
Here $V(x,t)$ can be of general charge transfer type, that is, $V(x,t)=\sum\limits_{j=1}^N V_j(x-g_j(t)v_j,t)$.
\end{example}

We then have the following asymptotic behavior for the above equations:

\begin{theorem} Let all assumptions  above be satisfied. The solution $\vu(t)$ has the following asymptotic decomposition: as $t\to \infty$, 
\eq\label{limit}
\| \vu(t)-U_0(t,0)\Omega_\alpha^*\vec{u}(0)-\vu_{wlc}(t)\|_{\mathcal{H}}\to 0  
\eeq 
where $\vec{u}_{wlc}(t)$ satisfies:
\begin{align}
&\vu_{wlc}(t) \text{is a sub-diffusive state:}\\
&(\vec{u}_{wlc}(t), \mathcal{X} \vec{u}_{wlc}(t) )_\mathcal{H} \lesssim_E\max \{ t^{1/2 +0}, t^{1/\sigma}\}.
\end{align}
\end{theorem}

 Here,
  $\mathcal{X}$ denotes a matrix operator:
\eq
\mathcal{X}:=\begin{pmatrix}
   \langle P\rangle^{-1} \langle x\rangle \langle P\rangle & 0\\
   0 & \langle x\rangle
\end{pmatrix}.
\eeq

\subsection{Propagation estimates}  We employ the following propagation estimates as introduced in \cite{SW20221}:

\begin{enumerate}
 \item(\textbf{Propagation Estimate}) ({\bf PRES}). Given a class of matrix operators $\{B(t)\}_{t\geq 0}$ with 
 \eq
 B(t)=\begin{pmatrix}
     B_1(t)& 0 \\
     0& B_2(t)
 \end{pmatrix},
 \eeq
 we define the time-dependent inner product as:
\begin{align}
    \langle B(t), \vec{u}(t)\rangle_t:=(\langle P\rangle u_1(t),\langle P \rangle B_1(t) u_1(t))_{L^2_x(\mathbb{R}^n)}+( u_2(t),B_2(t)u_2(t))_{L^2_x(\mathbb{R}^n)},
\end{align}
where $\vec{u}(t)$ denotes the solution to \eqref{KG}. Then the family $\{B(t)\}_{t\geq 0}$ is termed as a \textbf{Propagation Observable} if it satisfies the following condition: For a family of self-adjoint operators $B(t)$, the time derivative satisfies: there exists $L\in \mathbb{N}^+$ such that
\begin{align}
    & \partial_t \langle B(t), \vec{u}(t)\rangle_t=\left(\pm\sum\limits_{l=1}^L\sum\limits_{j=1}^2(\langle P \rangle^{2-j} u_j(t),C_{j,l}^*(t)C_{j,l}(t)\langle P \rangle^{2-j} u_j(t))_{L^2_x(\mathbb{R}^n)}\right)+g(t)\nonumber\\
    & g(t)\in L^1_{t}[1,\infty),\quad C_{j,l}^*(t)C_{j,l}(t)\geq0, \quad l=1,\cdots, L, j=1,2.
\end{align}
Integrating this over time, we derive the \textbf{Propagation Estimate}: 
\begin{align}
   \sum\limits_{l=1}^L \sum\limits_{j=1}^2\int_{t_0}^T\|C_{j,l}(t)\langle P\rangle^{2-j} u_j(t) \|_{L^2_x(\mathbb{R}^n)}^2dt&\leq \langle B(t), \vec{u}(t)\rangle_t\vert_{t=t_0}^{t=T}+\int_{t_0}^T|g(s)| ds\nonumber\\
\leq& \sup\limits_{t\in [t_0,T]} \left|\langle B(t),\vec{u}(t)\rangle_t\right|+C_g,  
\end{align}
where $C_g:=\|g(t)\|_{L^1_t[1,\infty)}.$ 
 \item(\textbf{Relative Propagation Estimate}) Consider a class of matrix operators $\{ \tilde{B}(t)\}_{t\geq 0}$  with 
 \eq
\tilde{B}(t)=\begin{pmatrix}
     \tilde{B}_1(t)& 0 \\
     0& \tilde{B}_2(t)
 \end{pmatrix}.
 \eeq 
 We denote their time-dependent expectation values as:
 \begin{align}
     \langle \tilde{B}: \vec{v}(t)\rangle_t:=(\langle P\rangle v_1(t), \langle P\rangle\tilde{B}_1(t) v_1(t)  )_{L^2_x(\mathbb{R}^n)}+(v_2(t), \tilde{B}_2(t)v_2(t)  )_{L^2_x(\mathbb{R}^n)},
 \end{align}
 where $\vec{v}(t)$ is not necessarily the solution to \eqref{KG}, but satisfies the condition:
\eq
\sup\limits_{t\geq 0}\langle \tilde{B}:\vec{v}(t)\rangle_t<\infty. \label{phiH}
\eeq
If \eqref{phiH} holds, and if the time derivative $\partial_t\langle \tilde{B}:\vec{v}(t)\rangle_t$ meets the following estimate: there exists $L\in \mathbb{N}^+$ such that
\begin{align}
&\partial_t\langle \tilde{B} : \vec{v}(t)\rangle_t=\pm \sum\limits_{l=1}^L \sum\limits_{j=1}^2(\langle P\rangle^{2-j} v_{j}(t), C^*_{j,l}(t)C_{j,l}(t)\langle P\rangle^{2-j}v_{j}(t))_{L^2_x(\mathbb{R}^n)}+g(t)\nonumber\\
&g(t)\in L^1[1,\infty), \quad C_{j,l}^*(t)C_{j,l}(t)\geq0, \quad l=1,\cdots, L, j=1,2.
\end{align}
Then the family $\{\tilde{B}(t)\}_{t\geq 0}$ is termed as a {\bf Relative Propagation Observable} with respect to $\vec{v}(t)$. Integrating this over time yields the\newline \textbf{Relative Propagation Estimate}: 
\begin{align}\label{CC}
   \sum\limits_{l=1}^L \sum\limits_{j=1}^2\int_{t_0}^T\|C_{j,l}(t)\langle P\rangle^{2-j} v_j(t) \|_{L^2_x(\mathbb{R}^n)}^2dt&= \langle B(t), \vec{v}(t)\rangle_t\vert_{t=t_0}^{t=T}-\int_{t_0}^Tg(s) ds\nonumber\\
\leq& \sup\limits_{t\in [t_0,T]} \left|\langle B(t),\vec{u}(t)\rangle_t\right|+C_g.  
\end{align}
\end{enumerate}
We define $\vec{v}(t)=U_0(0,t)\vec{u}(t)$, and we consider the operators
\eq
B_1(t)=\begin{pmatrix} \langle P\rangle^{-1}F_1F_\alpha(x,P,t)\langle P\rangle  & 0\\
0& F_1F_\alpha(x,P,t)
\end{pmatrix}
\eeq
and 
\eq
B_2(t)=\begin{pmatrix}
   \langle P\rangle^{-1} F_\alpha(x,P,t)F_c(\frac{|x|}{t^\alpha }\leq 1)\langle P\rangle  & 0\\
   0 & F_\alpha(x,P,t)F_c(\frac{|x|}{t^\alpha }\leq 1)
\end{pmatrix}.
\eeq
 Recall here that 
\eq
F_\alpha(x,P,t)=F_c(\frac{|x|}{t^\alpha}\leq 1)F_1
\eeq
where $F_1$ is given by
\eq
F_1=\begin{cases}
    F_1(|P|,t)& \text{ if }\mathcal{N}\text{ is local} \\
    F_1(|P|\leq t^\beta)& \text{ if }\mathcal{N}\text{ is non-local}
\end{cases}.
\eeq

To this, a direct computation of the time derivatives of the Propagation Observables defined above imply the desired Propagation Estimates, in a way similar to the Schr\"odinger case. This implies the existence of the Free Channel Wave Operators.
\cite{soffer2022large}

Next, we briefly recall the argument implying that the weakly localized part is subdiffusive.
The proof is different from the proof that a non-radiative solution spreads no faster than $|x|\sim t^{1/2}.$
The result in the presence of a free wave is a bit weaker, and the solution can spread like $t^{1/2+0}.$

The steps of the proofs are similar in both the Schr\"odinger and KG case:
In both cases the projection used to isolate the non-radiative part, is the projection on the region of phase-space where $|x-\vec{v}t|\geq t^{\alpha}.$ In the KG case $\vec{v}(k)=\frac{k}{\sqrt{1+k\cdot k}}.$

Now writing the solution in terms of the Duhamel Integral representation, we see that the first term, being a free wave is going to zero if we project on the above phase-space.
The Duhamel term is treated by projecting on incoming and outgoing wave separately.
In both cases the Duhamel term is of the form 
$$
U_0(t)\int_0^t U_0(-s)\mathcal{N}\vec{v}(s)ds
$$
that we use for controlling the incoming waves part, and 
$$
U_0(t)\int_t^{\infty} U_0(-s)\mathcal{N}\vec{v}(s)ds- U_0(t)\vec{v}(\infty).
$$
In the first case we project on the incoming waves by the use of localization on $F(|x|>t^{1/2+0})P^-$
which goes to zero when acting on $U_0(t-s)\mathcal{N}\vec{v}, t-s>0.$
Then, a similar observation is applied to the second term projected on outgoing waves.
Here, $P^+$ denote an appropriate projection on outgoing waves, for the KG equation.

\section{N-body Scattering, Real and Quasi Particles}
\subsection{Introduction}

   The scattering theory of time independent interaction terms is mostly considered in the context of Quantum systems.
   Mathematically speaking, given a self-adjoint operator $H$ acting on a Hilbert space $\mathcal{H},$
   understanding the properties of such an operator. In particular, its spectrum, its diagonalization etc... amounts to solving the "linear algebra" problem associated with the linear operator $H.$

   Quantum Mechanics offers a huge number of interesting examples of $H$, corresponding to actual physical systems and other models.
   
   The concrete example given by $H=-\Delta +V(x)$ on $L^2(\R^n)$ or on some other manifold is the classical example.
   This problem shows up in particle scattering in QM, in Spectral Geometry, Analytic Number Theory...
   In the simplest case, we deal with the dynamics of one quantum particle moving under the potential $V(x)$ which is assumed to be not too rough, and decays fast enough at infinity.
   
   But even the eigenvalue problem in this case, is a second order equation in many variables, and with non constant coefficients, and no general solution method exists.
   Solving this problem, and in particular proving AC has occupied a large number of people over a few generations.

   Once the number of particles in more than 2, this problem became fundamentally more complicated, since now the scattering is multi-channel and the potential $V(x)$ is constant on hyper-planes extending to infinity.
   
   A key tool in solving the large time behavior for dispersive equations was proving first Local-Decay estimates, for the full dynamics. See next section.
   
   However, to prove Local-Decay for N-body systems proved to be very difficult. It was achieved in 1981 using Mourre's method. It required proving the Mourre estimate for $N$ particles. This was done in \cite{PSS} following the work on Mourre for 3 particles.
   
   As mentioned before, Enss introduced a method of proving AC without the strong version of Local-Decay, instead relying on a very weak version of it \cite{E1978}.
   Enss  then showed how to solve the three body problem this way.
   However, this method could not be extended beyond three body, and the solution of the AC problem for N-body is using Local-Decay for the N-body Hamiltonian in a crucial way. Same for later proofs.
   This need for a Local-Decay made it difficult to extend the proofs to quasi-particles, since only in special cases the Dilation operator can be used to get a Mourre estimate.

   Using the new approach we present, it is natural to try to deal with the N body quasi particle case.
   Indeed, the three body case was done in quite general form in \cite{soffer2023three} by this approach  discussed below next.
   
\subsection{The Three Quasi-Particle System}

 We consider the scattering theory of a three \emph{quasi-particle system.} The system is described by the following equation:
\eq
\begin{cases}
i\partial_t\psi(x,t)=(H_0+V(x))\psi(x,t)\\
\psi(x,0)=\psi_0(x)\in \s^2_x(\mathbb{R}^9)
\end{cases},\quad x=(x_1,x_2,x_3)\in \mathbb{R}^{9}.\label{SE}
\eeq
Here, we define:
\begin{itemize}
    \item  $x_j\in \R^3 $(with $j=1,\cdots,3$) represents the position variable of the $j$th particle. 
    \item $P_j:=-i\nabla_{x_j}$ represents the {\bf wave number operator} of the $j$th particle, often referred to as the {\bf quasi-momentum.}
    \item $H_0=\sum\limits_{j=1}^3\omega_j(P_j)$, where $\omega_j: \R^3\to \R, \eta\mapsto \omega_j(\eta)$, denotes the kinetic energy operator, the {\bf dispersion relation.}
    
    \item The solution $\psi(x,t)$ is a complex-valued function of $(x,t)\in \mathbb{R}^{9+1}$.
    \item The term $V(x)$, corresponds to the interaction among  the three particles, has the form 
$$
V(x)=\sum\limits_{1\leq j<l\leq 3}V_{jl}(|x_j-x_l|),
$$
for some real-valued functions $V_{jl}, 1\leq j<l\leq 3$ .

\end{itemize}
We present two classical examples of the equation \eqref{SE}:
\begin{enumerate}
    \item When $\omega_j(P_j)=P_j^2=-\Delta_{x_j}, j=1,2,3$, it is the standard $3$-body system, which describes a system of $3$ non-relativistic particles interacting with each other.
    \item When $\omega_j(P_j)=\sqrt{m_j^2+P_j^2}$, $j=1,2,3,$ \eqref{SE} describes the system of 3 relativistic particles.
\end{enumerate}

There are several significant applications involving quasi-particles. For instance, a particle moving through a medium—like a periodic ionic crystal—exhibits a complex and implicit effective dispersion relation. This complexity is similarly observed in particles within Quantum Field Theory (QFT), where the effects of renormalization make the particle mass a momentum-dependent function in a complicated manner. There are also quasi-particles that are not elementary particles but rather derived constructs. A classic example involves the dynamics of the Heisenberg spin model, which can be described in terms of spin-wave excitations, also known as Magnons. In such cases, a typical dispersion relation can be expressed as
\eq
\omega(k)=4(3-\cos k_1-\cos k_2 -\cos k_3).
\eeq
See also \cite{grenier2020linear}.

Now let us introduce the interaction $V(x)$. Define $\langle x\rangle$ as $\sqrt{1+|x|^2}$ for $x\in \R^n$. When dealing with potentials or operators represented by functions, one typically assumes either of the following conditions:
\begin{itemize}
    \item (short-range potentials) For all $1\leq j<l\leq 3$, $|V_{jl}(\eta)|\leq C\langle \eta\rangle^{-1-\epsilon}$ for some $\epsilon>0$ and some constant $C>0$, or
    \item (long-range potentials) For some $1\leq j<l\leq3$, $|V_{jl}(\eta)|\leq C\langle \eta\rangle^{-\epsilon}$ only holds only for some $\epsilon \in (0,1]$.
\end{itemize}
 We focus on systems characterized by short-range interactions. System \eqref{SE} represents a multi-particle or three-body system. 

\subsection{Channel wave operators}

To prove AC, we use induction. Specifically, we reduce the $3$-body problem to several $2$-body problems. We then further reduce each $2$-body system to a one-body system through translation invariance. The one-body problem is well-studied, as demonstrated by  \cite{cycon2009schrodinger}. The main result is:

\begin{theorem}\label{Thm}With the Assumptions above on the decay of the two body potentials, and the conditions on the structure of the threshold set hold, then we have AC for system \eqref{SE}: for all $\psi(0)=P_c \psi(0)\in \s^2_x(\R^9)$, we have
\eq
\lim\limits_{t\to \infty}\|\psi(t)-\sum\limits_{a\in L}e^{-itH_a}\psi_{a,+}(x)\|_{\s^2_x(\R^9)}=0,\label{thm: AC}
\eeq
where for all $a\in L$, $\psi_{a,+}(x)$ are given by:
\begin{subequations}
For $a_0=(1)(2)(3)$, $a_l=(jk)(l)\in L$, 
\eq
\psi_{a_0,+}=\Omega^{a_1,*}\Omega_{a_1}^{*}\psi_0+\Omega^{a_2,*}(1-\Pp_{1}^{2}) \Omega_{a_2}^{*}\psi_0+\Omega^{a_3,*}(1-\Pp_{1}^{3})(1-\Pp_{2}^{3}) \Omega_{a_3}^{*}\psi_0,
\eeq
\eq
\psi_{a_1,+}:=P_{bs}(H_{a_1})\Omega_{a_1}^{*}\psi_0,
\eeq
\eq
\psi_{a_2,+}:=P_{bs}(H_{a_2})( 1-\Pp_{1}^{2})\Omega_{a_2}^{*}\psi_0
\eeq
and 
\eq
\psi_{a_3,+}:=P_{bs}(H_{a_3})(1-\Pp_1^{3})(1-\Pp_{2}^{3})\Omega_{a_3}^{*}\psi_0.
\eeq
\end{subequations}
\end{theorem}
\begin{remark}
Besides AC this theorem also give an important formula: the free channel wave operator of the three-body problem in terms of two-body wave operators. These kind of identities play a role in computational aspects of N-body scattering, going back to the classical work of Faddeev on the standard three body problem \cite{faddeev1963mathematical}.
\end{remark}

\begin{definition}[Channel wave operators] The $H_a$ channel wave operator is defined by 
\eq
\Omega_a^*\psi(0):=s\text{-}\lim\limits_{t\to \infty} e^{itH_a}J_ae^{-itH}P_{sc}\psi(0)\quad \psi(0)\in \mathcal{H}=\s^2_x(\mathbb{R}^9)
\eeq
where $\{ H_a\}$ denotes a set of all some sub-Hamiltonians of $H$ and $\{J_a\}$ is a set of some smooth cut-off functions of $x, P, t$ satisfying
\eq
\sum\limits_{a} J_a=1.
\eeq
\end{definition}
Based on the definition of the channel wave operators, it is evident that the establishment of the existence of all channel wave operators leads to the proof of AC.  Choosing the right $J_a$ is pivotal, and assembling a suitable set of $J_a$ requires ingenuity. In the context that follows, we assume that $\alpha$ is chosen appropriately, and we define $\eta_{jk}$ for $k=1,\cdots, N_{j}$ and $j=1,2,3,$ as the thresholds of $\Omega_j.$  We define
\eq
J_{\alpha,a}(t)=e^{-itH_a}F_{c,l,\alpha}(x_l,t,P_l)e^{itH_a}
\eeq
for $a=(jk)(l)\in L$, and 
\eq
J_{\alpha,free}(t):=e^{-itH_0}\left(\Pi_{l=1}^{l=3} F_{c,l,\alpha}(x_l,t,P_l)\right)e^{itH_0},
\eeq
for $a=(1)(2)(3)$. Here, 
\eq
F_{c,j,\alpha}(x_j,t,P_j)=F_c(|x_j|\leq t^\alpha)\left(\Pi_{k=1}^{N_j}F_1(|P_j-\eta_{jk}|>\frac{1}{t^{\alpha/2}})\right)F_1(|P_j|\leq t^{\alpha/2}),
\eeq
where $\eta_{jk}$ for $k=1,\cdots, N_j,$ are defined as the thresholds of $\Omega_j.$  We define the new $3$-body channel wave operators as follows:
\begin{align}
\Omega_{a,\alpha}^{*}:=&s\text{-}\lim\limits_{t\to \infty} e^{itH_a}J_{\alpha,a}(t)e^{-itH}\nonumber\\
=&s\text{-}\lim\limits_{t\to \infty}F_{c,l,\alpha}(x_l,t,P_l) e^{itH_{a}}e^{-itH}\quad \text{ on }\s_x^2(\mathbb{R}^9)
\end{align}
for $a=(jk)(l)\in L$, and we define the new free channel wave operator as follows:
\begin{align}
\Omega_{free,\alpha}^{*}:=&s\text{-}\lim\limits_{t\to \infty} e^{itH_a}J_{\alpha,a}(t)e^{-itH}\nonumber\\
=&s\text{-}\lim\limits_{t\to \infty}\left(\Pi_{l=1}^{l=3} F_{c,l,\alpha}(x_l,t,P_l)\right) e^{itH_{0}}e^{-itH}\quad \text{ on }\s_x^2(\mathbb{R}^9).
\end{align}

 We  introduce a new definition the of the projections onto the space of all bound states:
\begin{definition}[A new characterization of the projection on the space of all $3$-body bound states]\label{def:bd}Let $\alpha$ be as chosen before. The projection on the space of all $3$-body bound states is defined by
\eq
P_\mu:= s\text{-}\lim\limits_{t\to \infty} e^{itH}[\Pi_{j=1}^{j=3} e^{-it\omega_j(P_j)}(1-F_{c,j,\alpha}(x,t,P))e^{it\omega_j(P_j)}]e^{-itH}\quad \text{ on }\s^2_x(\mathbb{R}^9).
\eeq
\end{definition}
To continue our analysis, we require the following projections and channel wave operators. Given $a=(jk)(l), b=(j'k')(l')\in L$, we define the following:
\eq
\Pp^{l'}_l:=s\text{-}\lim\limits_{t\to \infty}e^{itH_b} e^{-itH_a}F_{c,l,\alpha}(x_l,t,P_l)e^{itH_a}e^{-itH_b},\quad \text{ on }\s^2_x(\R^9),
\eeq 
\eq
\Omega_{\alpha, jk}^*:=s\text{-}\lim\limits_{t\to \infty} F_{c,j,\alpha}(x_j,t,P_j)e^{itH_0} e^{-itH_a},\quad \text{ on }\s^2_x(\R^9),
\eeq
\eq
P_{sc}(H_a):=s\text{-}\lim\limits_{t\to \infty} e^{itH_a}e^{-itH_0}F_{c,j,\alpha}(x_j,t,P_j)e^{itH_0} e^{-itH_a},\quad \text{ on }\s^2_x(\R^9)
\eeq
and 
\eq
P_{bs}(H_a):=1-P_{sc}(H_a).
\eeq
\begin{remark} We would like to remind the reader that $\Omega_{\alpha, jk}^*$ represents the free channel wave operator for the two-body problem. Hence, in the definition of $\Omega_{\alpha, jk}^*$, we use $F_{c,j,\alpha}$ instead of $F_{c,l,\alpha}$.
\end{remark}
The proof of the existence of the Channel wave operators follows the same method as before: Use Cook's method to control the interaction term by dispersive and Local Decay estimates for the sub-hamiltonians;then use Propagation Estimate to prove the integrability of the Heisenberg derivative of the partition in phase space that defines the Channel.

To see this, consider first the free channel wave operator. In this case, the projection $J$ localizes all the distances between the three quasi-particles such that $|x_i-x_j|,$ for every pair of particles $(i,j),$ is growing with time fast enough.
For example, in
\eq
J_{\alpha,free}(t):=e^{-itH_0}\left(\Pi_{l=1}^{l=3} F_{c,l,\alpha}(x_l,t,P_l)\right)e^{itH_0},
\eeq
After conjugating the above operator by the FREE non-interacting 3 quasi-particle hamiltonian, the distance between particle 1 and 2,
will be controlled by (on LHS) by $t^{\alpha},$ using the product of the $F_c(|x_j|\leq t^\alpha).$ 
Then comes the action of the free dynamics, and then it acts on the interaction term (times $\psi(t)).$

Due to the decay of the two-body potentials, $V(x_1-x_2)\psi(t)$ is an $L^1$ function of $x_1-x_2$ and decaying.
Since we are also projecting the frequencies away from the thresholds of the free Hamiltonian (using the functions $F_1$),
we can use decay estimates of the free sub-hamiltonian to get integrability in time.

Similar arguments allow the construction of the Propagation Observables needed to prove the relevant Propagation Estimates.

A serious new problem shows up in linear time independent scattering: We need to show that {\bf all} states in the continuous spectral part of $H$ are scattering states. In the one body case, it is easy to achieve, using the following argument:

Suppose there is a solution that in the continuous spectral part of $H$, which is not asymptotically a free wave.
Then, it is weakly localized and spreads slower than $|x|\sim t.$ This however is in contradiction with the fact that there is a sequence of time along which $|x|\sim t,$ for any initial condition supported away from the thresholds.

This follows for example by observing that the expectation of the Dilation operator has positive commutator with $H,$
in the time average.
The contribution of the potential to the commutator becomes small using the weak local decay that follows from Ruelle theorem or Wiener's Lemma.

The analogous procedure for the three body case is much more complicated, but follows similar ideas.
The difficulties stem from the fact that since the problem is non  symmetric, we cannot prove that weakly localized solutions are spreading slowly; they may do that around each localized state, but we do not know in advance that these localized states are stable. (in general they are bound states of some particles, and they can break apart at any time).

The other problem is that we cannot prove directly that the commutator with $A$ or similar is positive (that would amount to proving the Mourre estimate).

The first step is to use Ruelle's Theorem to show that there is a sequence of times for any $R$ finite, s.t. that solution in a ball of radius $R$, decays in time:

 Choose a scattering state, $\psi\in \s^2_x(\R^9)$. Then $\psi=\psi_{sc}$ with $\|\psi\|_{\s^2_x(\R^9)}=1$.  By Ruelle's Theorem we have that for any $M\geq 1$, that there exists a sequence of time $\{t_n\}_{n=1}^{n=\infty}$ with $t_n=t_n(M,\psi)\uparrow \infty$ as $n\to \infty$, such that 
\eq
\| \chi(|\tilde{x}|\leq M)e^{-it_nH}\psi\|_{\s^2_x(\R^9)}< \frac{1}{n}. \label{Req1}
\eeq
If one can show that for each $\epsilon>0$, there exist $M=M(\epsilon,\psi)\geq 1$ and $T=T(\epsilon, M, \psi)>0$ such that when $t\geq T$, 
\eq
\| \chi(|\tilde{x}|>M) e^{-itH}\psi- e^{-itH_0}\psi_{free}-\sum\limits_{a=(jk)(l)\in L}e^{-itH_a}\psi_a\|_{\s^2_x(\R^9)}<\epsilon,\label{sec 4main:eq}
\eeq
then, we get a sequential AC. The sequential AC implies that 
\eq
e^{it_nH}\left( \Pi_{l=1}^{l=3} e^{-it_n\omega_l(P_l)} (1-F_{c,l,\alpha}(x_l,t,P_l))e^{it_n\omega_l(P_l)}\right)e^{-it_nH}\psi\to 0
\eeq
in $\s^2_x(\R^9)$ for a sequence of time $\{t_n\}_{n=1}^{n=\infty}$. Due to the existence of the Channel wave operators, we get $P_\mu P_{sc}=0,$ where $P_{\mu}$ is the projection on states which are orthogonal to the range of the (union of) the Channel Wave Operators. So it suffices to prove \eqref{sec 4main:eq}.

This last estimate says that the solution breaks into (linear combination of) clusters which are far separated.
Then, we microlocalize each such subspace around the propagation set, and away from it.
The propagation set is where classical trajectories live \cite{SS1987}.
Then we prove that the solution is small away from the Propagation set.

To get such estimates, we need to reduce it to the corresponding estimates on the free flow.
For this we have the following tools: the interaction terms are small since we localize on large distance between the clusters (for each channel separately). The velocities between the clusters should be localized away from thresholds, therefore, classically, the clusters move away from each other.
Finally, we use intertwining property of the channel wave operators to replace the full dynamics by the asymptotic dynamics, up to higher order corrections. These higher order corrections are given by a Duhamel term.

Consider for example the channel  $a=(jk)(l).$

For $a=(jk)(l)\in L$ fix $\epsilon>0$. Recall that
\eq
\psi_{M,a,b}(x,t)=\chi(|x_j-x_k|< M^{\frac{1}{100}}) g_{M,a}(\tilde{x})P_b(H^a)\bar{F}_\tau(H_a,\epsilon_1)e^{-itH}\psi.
\eeq
This is the part of the solution where the distance between the particles $j,k$ is small relative to the total distance of order $M$. Then we project on the bound states subspace of the internal Hamiltonian, $H^a$.
Then we also localize the external Hamiltonian away from its thresholds.

$\psi_{M,a,b}(x,t)$ can be rewritten as 
\begin{align}
    \psi_{M,a,b}(x,t)=\sum\limits_{d=1}^{N^a}\psi_{M,a,b,d}(x,t),\label{psiMab}
\end{align}
where 
\eq
\psi_{M,a,b,d}(x,t):=\chi(|x_j-x_k|< M^{\frac{1}{100}}) g_{M,a}(\tilde{x})P_{b,d}(H^a)\bar{F}_\tau(H_a,\epsilon_1)e^{-itH}\psi.
\eeq
For each $d\in \{1,\cdots, N^a\}$, we define the forward and backward projections with respect to the flow $e^{-it(\omega_l(P_l)+\lambda_{a,d}(P_j+P_k))}$ as:
\eq
P_{a,d}^\pm :=P^\pm (r,v)
\eeq
where $r=x_l-x_j$ and $v=v_l(P_l)-\nabla_{P_j+P_k}[\lambda_{a,d}(P_j+P_k)]$. We then decompose $\psi_{M,a,b,d}(x,t)$ into two components:
\begin{align}
    \psi_{M,a,b,d}(x,t)=& \chi(|x_j-x_k|< M^{\frac{1}{100}}) g_{M,a}(\tilde{x})P^+_{a,d}P_{b,d}(H^a)\bar{F}_\tau(H_a,\epsilon_1)e^{-itH}\psi\nonumber\\
    &+\chi(|x_j-x_k|< M^{\frac{1}{100}}) g_{M,a}(\tilde{x})P^-_{a,d}P_{b,d}(H^a)\bar{F}_\tau(H_a,\epsilon_1)e^{-itH}\psi\nonumber\\
    =:&\psi_{M,a,b,d}^+(x,t)+\psi_{M,a,b,d}^-(x,t).\label{psiMabd}
\end{align}
For $\psi_{M,a,b,d}^+(x,t)$, we approximate it using $P_b(H^a) e^{-itH_a} \Omega_{a}^{*}\psi$. For $\psi_{M,a,b,d}^-(x,t)$, we use $P_b(H^a) e^{-itH_a} \psi$ as its approximation. Their errors decay in $M$ by using forward/backward propagation estimates with respect to flow $e^{-it(\omega_l(P_l)+\lambda_{a,d}(P_j+P_k))}$.
The way the Intertwining is used gives, after projecting on the outgoing waves subspace (the incoming part will go to zero):
 \begin{align}
        &\psi_{M,a,b,d}^+(x,t)\nonumber\\
        =& \chi(|x_j-x_k|< M^{\frac{1}{100}}) g_{M,a}(\tilde{x})P^+_{a,d}P_{b,d}(H^a)\bar{F}_\tau(H_a,\epsilon_1)e^{-itH_a}\Omega_{a}^{*}\psi\nonumber\\
        &+i\int_t^\infty ds \chi(|x_j-x_k|< M^{\frac{1}{100}}) g_{M,a}(\tilde{x})P^+_{a,d}P_{b,d}(H^a)\bar{F}_\tau(H_a,\epsilon_1)e^{-i(t-s)H_a}V_{jl}(x_j-x_l)e^{-isH}\psi\nonumber\\
        &+i\int_t^\infty ds \chi(|x_j-x_k|< M^{\frac{1}{100}}) g_{M,a}(\tilde{x})P^+_{a,d}P_{b,d}(H^a)\bar{F}_\tau(H_a,\epsilon_1)e^{-i(t-s)H_a}V_{kl}(x_k-x_l)e^{-isH}\psi\nonumber\\
        =:& \psi_{M,a,b,d,0}^+(x,t)+\psi_{M,a,b,d,1}^+(x,t)+\psi_{M,a,b,d,2}^+.
    \end{align}

    The projection on bound state $d$ replaces $H^a$ by a number, and therefore, up to a phase, $H_a$
    is two-body hamiltonian, with potential $V_{jl}+V_{kl}.$ Since we localize away from thresholds, standard Propagation Estimates apply, to show that the last two terms decay, and so, as expected the asymptotic behavior is given by the first term.

\section{Local Decay Estimates}

In this section we consider one of the most intricate issues in General Scattering Theory:
How to estimate the rate of convergence of the solution to its asymptotic state, in particular to a free wave.
This kind of estimate can be derived from LD (Local Decay), but LD is an a-priory estimate on the full dynamics.

The methods available to prove such estimates are limited to time independent potentials with the Standard Dispersion Relation, and away from thresholds for any Multi-Channel Scattering Problem.
Even in cases one can prove AC, as in the Enss method, it does not give convergence rate.

In this section we describe a method that uses the results on AC to prove LD in general situations, covering potentials which are quasi-periodic in 5 or more dimensions. These estimates are robust and are stable under small perturbations.
The methods introduced in this section are then used to prove the localization in space of the weakly localized part of the radial case, in 5 or more dimensions.

We demonstrate the method by first sketching the time independent case, in 3 or more dimensions.

\subsubsection{Time-independent cases}Let us begin with the proof for time-independent cases. We split $e^{-itH}P_c\psi$ into four pieces, using incoming/outgoing decomposition(see \eqref{P+} and \eqref{P-} for the definition of incoming/outgoing projections)
\begin{multline}
e^{-itH}P_c\psi=P^+e^{-itH}P_c\psi+P^-e^{-itH}P_c\psi\\
=P^+e^{-itH_0}\Omega_+^*P_c\psi+P^+(1-\Omega_+^*)e^{-itH}P_c\psi+P^-e^{-itH_0}\Omega_-^*P_c\psi+P^-(1-\Omega_-^*)e^{-itH}P_c\psi\\
=P^+e^{-itH_0}\Omega_+^*P_c\psi+P^-e^{-itH_0}\Omega_-^*P_c\psi+Ce^{-itH}P_c\psi
\end{multline}
where
\eq
C_1:=P^+(1-\Omega_+^*),
\eeq
\eq
C_2:=P^-(1-\Omega_-^*),
\eeq
\eq
C:=C_1+C_2
\eeq
and $\Omega_\pm^*$ are the conjugate wave operators
\eq
\Omega_\pm^*:=s\text{-}\lim\limits_{t\to \pm\infty} e^{itH_0}e^{-itH}P_c\quad \text{ on }\s^2_x(\mathbb{R}^3).
\eeq
Here we also used the intertwining property
\eq
\Omega_\pm^*e^{-itH}=e^{-itH_0}\Omega_{\pm}^*\quad \text{ on }\s^2_x(\mathbb{R}^5).
\eeq
\begin{remark}
The above compactness of $C_1, C_2$ to be proved, is similar in some respects to  the compactness estimates used by Enss \cite{E1978} and Davies \cite{D1980}.
But there are differences: Since we already know that AC holds, we prove the compactness for the adjoint of the wave operator.
Moreover, we do not need to localize the Hamiltonian away from the thresholds $0$ and $\infty$. That would not be possible in the time dependent case.
Instead, we use the local smoothing estimates of the free flow to deal with the high energy part. In 3,4 dimensions we assume that $0$ is a regular point of the spectrum of $H.$ In the time dependent case, we consider only dimensions 5 or higher, so, we do not need a regularity assumption.
 \end{remark}

 Next, we need Propagation Estimates for the free flow, including situations where the frequency is localized near zero, and other estimates where we gain derivatives by localizing away from the Propagation Set.
 
 \subsection{ Free Wave Propagation estimates}\label{inout}

We use the notion of incoming/outgoing waves and state the main Propagation Estimates for the free flow.

The incoming/outgoing wave decompositions are similar to the ones initiated by Mourre \cite{M1979}. The dilation generator $A$ is defined by
\eq
A:=\frac{1}{2}(x\cdot P+P\cdot x).
\eeq
\begin{definition}[Incoming/outgoing waves] The projection on outgoing waves is defined by \cite{S2011}:
\eq
P^+:= (\tanh(\frac{A-M}{R})+1)/2\label{P+}
\eeq
for some sufficiently large $R,M>0$ such that Lemma \ref{out/in1} holds. The projection on incoming waves is defined by
\eq
P^-:=1-P^+.\label{P-}
\eeq
\end{definition}
They enjoy the following properties (when the energy is both away from $0$ and $\infty$, such estimates are proved in \cite{HSS1999}).
\begin{lemma}\label{out/in1}When $R\geq R_0$ for some sufficiently large $R_0,$ the incoming and outgoing waves satisfy:
\begin{enumerate}
\item {\bf High Energy Estimate}

For all $\delta> 1, t\geq 0, c>0, N\geq 1$, when the space dimension $n\geq 1$,
\eq
\|P^{\pm}  F(|P|>c)e^{\pm i tH_0}\langle x\rangle^{-\delta}\|_{\s^2_x(\mathbb{R}^n)\to \s^2_x(\mathbb{R}^n)}\lesssim_{c,n}  \frac{1}{\langle t\rangle^{\delta}} .\label{Sep29.1}
\eeq
\item {\bf Pointwise Smoothing Estimate}

For $\delta> 0, t\geq v>0, c>0, l\in[0,\delta)$, when the space dimension $n\geq 1$,
\eq
\|P^{\pm}  F(|P|>c)e^{\pm i tH_0}|P|^l\langle x\rangle^{-\delta}\|_{\s^2_x(\mathbb{R}^n)\to \s^2_x(\mathbb{R}^n)}\lesssim_{c,n,v,l}  \frac{1}{\langle t\rangle^{\delta}} .\label{Sep20.2}
\eeq

\item {\bf Time Smoothing Estimate}

For $\delta> 2,c>0$, when the space dimension $n\geq 1$,
\eq
\int_0^1 dt t^2\|P^{\pm}  F(|P|>c)e^{\pm i tH_0}|P|^2\langle x\rangle^{-\delta}\|_{\s^2_x(\mathbb{R}^n)\to \s^2_x(\mathbb{R}^n)}\lesssim_{c,n} 1 .\label{Oct.1}
\eeq

\item {\bf Microlocal Decay}

There exists some $\delta=\delta(n)>1$ such that when $n\geq 3$,
\eq
\| \langle x\rangle^{-\delta} P^{\pm} e^{\pm i tH_0}\|_{\s^2_x(\mathbb{R}^n)\to \s^2_{x,t}(\mathbb{R}^{n+1})}\lesssim_{c,n} 1.\label{Sep19.5}
\eeq
In particular, when $n=3$, $\delta$ can be any positive number which is greater than $1$.
\item {\bf Near Threshold Estimate}

For $\delta>1$ when $t\geq1,$ ($\epsilon\in (0,1/2)$)
\eq
\| P^{\pm}  F(|P|>\frac{1}{\langle t\rangle^{1/2-\epsilon}})e^{\pm i tH_0} \langle x\rangle^{-\delta}\|_{\s^2_x(\mathbb{R}^n)\to \s^2_x(\mathbb{R}^n)}\lesssim_{n,\epsilon}  \frac{1}{t^{(1/2-\epsilon)\delta}}  .\label{Sep20.1}
\eeq
In particular, when $n\geq 5$, one has ($\epsilon\in (0,n/4)$)
\eq
\| P^{\pm}  e^{\pm i tH_0} \langle x\rangle^{-\delta}\|_{\s^2_x(\mathbb{R}^n)\to \s^2_x(\mathbb{R}^n)}\lesssim_{n,\epsilon} \frac{1}{\langle t\rangle^{n/4-\epsilon}} \in \s^1_t(\mathbb{R})
\eeq
since the $\s^2 $ volume of $F(|P|\leq \frac{1}{\langle t\rangle^{1/2-\epsilon}})f$ in frequency space is controlled by $ \frac{1}{\langle t\rangle^{n/4-n\epsilon/2}}\|f\|_{\s^1_x}$ up to some constant.

\item {\bf Global Time Smoothing Estimate }For $a=0,1,2,$ $\sigma>4$, $n\geq 5$,
\eq
 \int_0^\infty ds s^a\| P^+e^{isH_0}(-\Delta)^aF(|P|\leq 1)\|_{\s^2_{x,\sigma}(\mathbb{R}^n)\cap \s^1_x(\mathbb{R}^n)\to\s^2_x(\mathbb{R}^n) }\lesssim_{\sigma,n} 1.\label{Jan23.1}
 \eeq
\end{enumerate}
\end{lemma}
{\bf Comments on the Proof}
The estimates are for the free hamiltonian dynamics, which on $\mathbb{R}^n$ can be written explicitly by Fourier transform. It can then be estimated by standard Stationary phase methods.
Some of these estimates may be cumbersome to derive in such a way.
But they also follow from the propagation estimates of \cite {SS1988, HSS1999}.
Alternatively, one can follow the arguments of \cite {HSS1999} combined with the analytic construction of propagation observables(PROBs for short) of \cite{S2011},
to get a shorter and direct proof.
Most of these estimates will follow from direct computation of the resulting Propagation Estimates applied to the PROBs:
$ P^{-}(\frac{A-M}{R}), \quad \sum_{j=1}^N x_j P^{-}(\frac{A-M}{R}) x_j, \quad (A-M)P^{-}(\frac{A-M}{R}).$
These PROBs will lead to estimates which hold \emph{for all energies,} including $0,\infty.$
By further localizing the energy(frequency) of the initial data away from the thresholds $0,\infty$
we can then use similar operators as above but with $A(t)\equiv A-bt$ instead of $A$, the usual dilation operator on $\mathbb{R}^n.$
Here we take $b< 2E_m,$ $E_m$ is the lowest energy in the support of the initial data.

\subsection{Compactness} 
Here we start by using $\Omega_\pm^*.$    Indeed the existence of the free channel wave operators $\Omega_{\alpha,\pm}^*$ implies the existence of $\Omega_{\pm}^*$ and they are equal to each other, see \cite{SW20221}.
$C$ is compact on $\s^2_x$(for compactness, see Lemma \ref{Lem1}).

Based on \eqref{Sep19.5} in Lemma \ref{out/in1},
\eq
\int dt\| \langle x\rangle^{-\sigma}P^\pm e^{-itH_0}\Omega_\pm^*P_c\psi \|_{\s^2_x(\mathbb{R}^3)}^2\lesssim_{\sigma} \|\psi\|_{\s^2_x(\mathbb{R}^3)}^2\label{psif}
\eeq
for all $\sigma>1$.
\begin{lemma}\label{cpt}$C$ can be expressed by
\eq
C=C_r+C_m\label{r+m}
\eeq
with $C_r, C_m$ satisfying
\eq
\|C_r\|_{\s^2_x(\mathbb{R}^3)\to\s^2_x(\mathbb{R}^3) }\leq 1/100,\label{FS1}
\eeq
\eq
\|\langle x\rangle^{-\eta} (1-C_r)^{-1}C_r\langle x\rangle^{\sigma}\|_{\s^2_x(\mathbb{R}^3)\to \s^2_x(\mathbb{R}^3)}\lesssim_{\sigma,\eta} 1,\quad \eta>1, \sigma \in (1,101/100),\label{FS2}
\eeq
and
\eq
\int dt\| C_me^{-itH}P_c\psi \|_{\s^2_x(\mathbb{R}^3)}^2 \lesssim \|\psi\|_{\s^2_x(\mathbb{R}^3)}^2\label{FS3}.
\eeq
\end{lemma}
Let
\eq
\psi_f(t):=P^+e^{-itH_0}\Omega_+^*P_c\psi+P^-e^{-itH_0}\Omega_-^*P_c\psi.
\eeq
Using \eqref{r+m},
\eq
e^{-itH}P_c\psi=\psi_f(t)+Ce^{-itH}P_c\psi=\psi_f(t)+(C_r+C_m) e^{-itH}P_c\psi.
\eeq
Based on \eqref{FS1}, $(1-C_r)^{-1}$ exists on $\s^2_x(\mathbb{R}^3)$ and one could rewrite $e^{-itH}P_c\psi$ as
\eq
e^{-itH}P_c\psi=(1-C_r)^{-1}\psi_f(t)+(1-C_r)^{-1}C_me^{-itH}P_c\psi.
\eeq
By using \eqref{FS2}, \eqref{FS3}, \eqref{psif}, using that
\eq
(1-C_r)^{-1}=1+(1-C_r)^{-1}C_r
\eeq
and taking $\sigma=1001/1000$, one has that
\begin{multline}
\int dt \| \langle x\rangle^{-\eta}e^{-itH}P_c\psi \|_{\s^2_x(\mathbb{R}^3)}^2\lesssim\int dt \| \langle x\rangle^{-\eta}\psi_f(t) \|_{\s^2_x(\mathbb{R}^3)}^2+\\
\|\langle x\rangle^{-\eta} (1-C_r)^{-1}C_r\langle x\rangle^\sigma\|_{\s^2_x(\mathbb{R}^3)\to \s^2_x(\mathbb{R}^3)}^2\int dt \|\langle x\rangle^{-\sigma}\psi_f(t) \|_{\s^2_x(\mathbb{R}^3)}^2+\\
\|\langle x\rangle^{-\eta} (1-C_r)^{-1}\|_{\s^2_x(\mathbb{R}^3)\to \s^2_x(\mathbb{R}^3)}^2\int dt\| C_me^{-itH}P_c\psi \|_{\s^2_x(\mathbb{R}^3)}^2\lesssim_{\eta}\|\psi\|_{\s^2_x(\mathbb{R}^3)}^2
\end{multline}
and we finish the proof.

The proof is then reduced to proving the compactness of operators which are in general microlocalization of $(\Omega^*-I) P_c.$
\subsection{Compactness of $C$}
 We prove the compactness of $C$ in this subsection.
\begin{lemma}\label{maincpt}If $V(x)$ satisfies $V(x)\lesssim <x>^{-6-0}$, then $C$ is a compact operator on $\s^2_x(\mathbb{R}^3)$.
\end{lemma}
In order to prove Lemma \ref{maincpt}, it suffices to show the compactness of $C_1$ and $C_2$. Recall that
\eq
C_1=P^+(1-\Omega_+^*)
\eeq
and
\eq
C_{2}=P^-(1-\Omega_-^* ).
\eeq
The proof requires the following lemma.
\begin{lemma}[Representation formula for $\Omega_\pm^*$]\label{formula}Assume $0$ is neither an eigenvalue nor a resonance for $H$ and assume $\langle x\rangle^{2\sigma} V(x)\in \s^\infty_x(\mathbb{R}^3)$ for some $\sigma>1$. For $g(x)$ satisfying $\langle x\rangle^{\sigma} g(x) \in \s^2_x(\mathbb{R}^3)$,
\eq
\Omega_\pm^*g=(1+V(x)\frac{1}{H_0})^{-1}g-i\int_0^{\pm\infty} dt H_0e^{itH_0}\Omega_\pm^*V(x)\frac{e^{-itH_0}}{H_0}(1+V(x)\frac{1}{H_0})^{-1}g.\label{express}
\eeq
\end{lemma}
\begin{remark}In $5$ or higher space dimensions, we don't need this lemma since there are no zero-frequency issues.
\end{remark}
\begin{proof}Let $\s^2_{-\sigma, x}(\mathbb{R}^3)$ denote the weighted $\s^2_x$ space
\eq
\s^2_{-\sigma, x}(\mathbb{R}^3):=\{f: \langle x\rangle^{-\sigma}f\in \s^2_x(\mathbb{R}^3)\}.
\eeq
$(1+V(x)\frac{1}{H_0})^{-1}: \s^2_{\sigma,x}(\mathbb{R}^3)\to  \s^2_{\sigma,x}(\mathbb{R}^3)$, is bounded for $\sigma>1$ since $0$ is neither an eigenvalue nor a resonance for $H$.(Here $1$ comes from resolvent estimates, see Lemma 22.2 in \cite{KK2014} for example).

To verify the above statement we use the following Lemma:
\begin{lemma}
Let $a,b$ non commuting in general operators, which may not be bounded.
If $(I+ab)^{-1}$ is a bounded operator in the Hilbert space, then
$$
(I+ab)^{-1}=I-ab+ab(I+ab)^{-1} ab
$$
\end{lemma}

We use it with $b=\frac{1}{H_0}, \, a=V(x).$

So  the integrand of \eqref{express} is well-defined acting on localized functions.

Now we prove the validity of \eqref{express}. Using Duhamel's formula and employing integration by parts, one has
\begin{align}
\Omega_+^*=&\Omega_+^*\Omega_++\Omega_+^*(1-\Omega_+)\\
=&1+(-i)\int_0^\infty dt e^{itH_0}\Omega_+^*V(x)e^{-itH_0}\\
=&1+e^{itH_0}\Omega_+^*V(x)\frac{e^{-itH_0}}{H_0}\vert_{t=0}^{t=\infty}-i\int_0^\infty dt H_0e^{itH_0}\Omega_+^*V(x)\frac{e^{-itH_0}}{H_0}\\
=&1-\Omega_+^*V(x)\frac{1}{H_0}-i\int_0^\infty dt H_0e^{itH_0}\Omega_+^*V(x)\frac{e^{-itH_0}}{H_0}
\end{align}
which implies
\eq
\Omega_+^*(1+V(x)\frac{1}{H_0})f=f-i\int_0^\infty dt H_0e^{itH_0}\Omega_+^*V(x)\frac{e^{-itH_0}}{H_0}f.\label{express2}
\eeq
Taking $f=(1+V(x)\frac{1}{H_0})^{-1}g$ and plugging it into \eqref{express2}, one gets \eqref{express} for $\Omega_+^*$. Similarly, one gets \eqref{express} for $\Omega_-^*$. We finish the proof.
\end{proof}

\begin{proposition}\label{Lem1}If $\langle x\rangle^\delta V(x)\in \s^\infty_x(\mathbb{R}^3) $ for some $\delta>4$, then $C_{1}$ and $C_{2}$ are compact operators.
\end{proposition}
\begin{proof}Let
\eq
C_{1}(v):=iP^+\int_v^\infty du e^{iuH_0}\Omega_+^*V(x)e^{-iuH_0}.
\eeq
Break $C_1(v)$ into two pieces
\begin{multline}
C_1(v)=iP^+\int_v^\infty du F(|P|>1/\langle u\rangle^{1/2-\epsilon})e^{iuH_0}\Omega_+^*V(x)e^{-iuH_0}+\\
iP^+\int_v^\infty du \bar{F}(|P|\leq 1/\langle u\rangle^{1/2-\epsilon})e^{iuH_0}\Omega_+^*V(x)e^{-iuH_0}=:C_{1,l}(v)+C_{1,s}(v)
\end{multline}
for some $\epsilon >0$ satisfying $ (1/2-\epsilon)(\delta-\epsilon)>2$. Due to Lemma \ref{out/in1}(Pointwise Smoothing Estimate and Near Threshold Estimate), using Duhamel formula to expand $\Omega_+^*$,
\begin{multline}
\|C_{1,l}(v)\|_{\s^2_x(\mathbb{R}^3)\to \s^2_x(\mathbb{R}^3)}\lesssim_v \int_v^\infty \frac{du}{\langle u\rangle^{(1/2-\epsilon)(\delta-\epsilon) }}\| \langle x\rangle^{\delta-\epsilon} \langle P\rangle^{-\epsilon}V(x) \|_{\s^2_x(\mathbb{R}^3)\to\s^2_x(\mathbb{R}^3)}+\\
\int_v^\infty du\int_0^\infty ds \frac{du}{\langle u+s\rangle^{(1/2-\epsilon)(\delta-\epsilon) }}\| \langle x\rangle^{\delta-\epsilon} \langle P\rangle^{-\epsilon}V(x) \|_{\s^2_x(\mathbb{R}^3)\to\s^2_x(\mathbb{R}^3)}\times \|V(x)\|_{\s^\infty_x}\\
\lesssim  \| \langle x\rangle^{\delta-\epsilon} \langle P\rangle^{-\epsilon}V(x) \|_{\s^2_x(\mathbb{R}^3)\to\s^2_x(\mathbb{R}^3)}(1+  \|V(x)\|_{\s^\infty_x})\label{C1l}
\end{multline}
for all $v>0$ and
\eq
\|C_{1,l}(v)\|_{\s^2_x(\mathbb{R}^3)\to \s^2_x(\mathbb{R}^3)}\to 0\quad \text{ as }v\to \infty.
\eeq
Since $ \langle x\rangle^{-\epsilon}\langle P\rangle^{-\epsilon} $ is compact on $\s^2_x(\mathbb{R}^n)$ for all $n\geq 1$, $C_{1,l}(v)$ is compact on $\s^2_x(\mathbb{R}^3)$.  When the space frequency is slightly away from the origin, by using Lemma \ref{out/in1}, we could get enough decay in $u$ when $V(x)$ is well-localized. It will be the same in time-dependent cases. For compactness and integrability of $u$, all we have to check is the part with small space frequency.

The control of $C_{1,s}(v)$ follows a similar set of arguments. 
For $C_{1,s}(v)$, using Lemma \ref{formula}, one has
\begin{multline}
C_{1,s}(v)=iP^+\int_v^\infty du \bar{F}(|P|\leq 1/\langle u\rangle^{1/2-\epsilon})e^{iuH_0}\Omega_+^*V(x)e^{-iuH_0}\\
=iP^+\int_v^\infty du \bar{F}(|P|\leq 1/\langle u\rangle^{1/2-\epsilon})e^{iuH_0}(1+V(x)\frac{1}{H_0})^{-1}V(x)e^{-iuH_0}+\\
P^+\int_v^\infty du \bar{F}(|P|\leq 1/\langle u\rangle^{1/2-\epsilon})\\ \times \int_0^\infty du_1 H_0e^{i(u+u_1)H_0}\Omega_+^*V(x)\frac{e^{-iu_1H_0}}{H_0}(1+V(x)\frac{1}{H_0})^{-1}V(x)e^{-iuH_0}\\
=:C_{1,s,1}(v)+C_{1,s,2}(v).\label{qu1}
\end{multline}

\begin{remark}
The difficulty here is that we do not have any localization bound for $\Omega^*_+$ and the frequency is localized near zero. Therefore we cannot get any decay from the projection $P^+.$    
\end{remark}

For $C_{1,s,1}(v)$, one has an extra localization,
\begin{multline}
\| C_{1,s,1}(v)\|_{\s^2_x\to \s^2_x}\lesssim\\ 
\int_v^\infty du \| \bar{F}(|P|\leq 1/\langle u\rangle^{1/2-\epsilon})\|_{\s^1_x(\mathbb{R}^3)\to \s^2_x(\mathbb{R}^3)}\| (1+V(x)\frac{1}{H_0})^{-1}V(x)e^{-iuH_0}\|_{\s^2_x\to \s^1_x}\\
\lesssim \int_v^\infty du \frac{1}{\langle u\rangle^{3/4-3/2\epsilon}}\| (1+V(x)\frac{1}{H_0})^{-1}V(x)e^{-iuH_0}\|_{\s^2_x\to \s^1_x}\\
\lesssim  \int_v^\infty du \frac{1}{\langle u\rangle^{3/4-3/2\epsilon}}\| (1+V(x)\frac{1}{H_0})^{-1}V(x)\|_{\s^6_x\to \s^1_x}\| e^{-iuH_0}\|_{\s^2_x\to \s^6_x}\lesssim_\epsilon 1\label{C1s1},
\end{multline}
by using Strichartz estimate for the last inequality, choosing $\epsilon \in (0,1/6).$ Then, using that for some $\sigma>1$, $\langle x\rangle^\delta V(x)\in \s^\infty_x$ with $\delta>4$, we have
\begin{multline}
\| (1+V(x)\frac{1}{H_0})^{-1}V(x)\|_{\s^6_x\to \s^1_x}\leq \| V(x)\|_{\s^6_x\to \s^1_x}+\\
\| V(x)\frac{1}{H_0}\langle x\rangle^{-\sigma}\|_{\s^2_x\to \s^1_x}\| \langle x\rangle^\sigma(1+V(x)\frac{1}{H_0})^{-1}\langle x\rangle^{-\sigma}\|_{\s^2_x\to \s^2_x}\times \| \langle x\rangle^\sigma V(x)\|_{\s^6_x\to \s^2_x}\lesssim 1.
\end{multline}
The new estimates which are needed in this case, are the use of the end-point Strichartz estimates and volume estimates in the frequency space.  In three dimensions, the volume in Fourier space of the support of $\bar{F}(|P|\leq 1/\langle u\rangle^{1/2-\epsilon})$ is bounded (up to a constant) by
$u^{-3/2+3\epsilon}.$ The other parts are controlled by similar (though longer, since we need to control a double integral over time variable) estimates.
\end{proof}

Next, we discuss the extension of this analysis to quasi-periodic potentials.

\subsubsection{Time-dependent cases}

By AC the subspace of scattering states is identified by the range of the wave operators $\Omega_{\pm}$.

For $U(t,0)\Omega_+\phi$, using incoming/outgoing decomposition, we split $U(t,0)\Omega_+\phi$ into four pieces:
\begin{align}
U(t,0)\Omega_+\phi=&P^+e^{-itH_0}\phi+P^-e^{-itH_0}\Omega_-^*\Omega_+\phi+\\
& P^+(1-\Omega_{t,+}^*) U(t,0)\Omega_+\phi+ P^-(1-\Omega_{t,-}^*) U(t,0)\Omega_+\phi\\
=:&\psi_{1}(t)+\psi_{2}(t)+ C_1(t)U(t,0)\Omega_+\phi+C_2(t)U(t,0)\Omega_+\phi
\end{align}
with
\eq
P_c(t)=s\text{-}\lim\limits_{s\to \infty}U(t,t+s)F_c(\frac{|x-2sP|}{s^\alpha}\leq 1)U(t+s,t) ,\quad \text{ on }\s^2_x(\mathbb{R}^5),
\eeq
\eq
P^+\Omega_{t,+}^*:=s\text{-}\lim\limits_{a\to \infty} P^+e^{iaH_0}U(t+a,t)P_c(t)\quad \text{ on }\s^2_x(\mathbb{R}^5),
\eeq
\eq
P^-\Omega_{t,-}^*:=s\text{-}\lim\limits_{a\to -\infty} e^{iaH_0}U(t+a,t)P_c(t)\quad \text{ on }\s^2_x(\mathbb{R}^5)
\eeq
and
\eq
C_1(t):=P^+(1-\Omega_{t,+}^*),\quad C_2(t):=P^-(1-\Omega_{t,-}^*) .
\eeq
It is not clear here whether $\Omega_{t,+}^*$ and $\Omega_{t,-}^*$ exist with $P_c(t)$ defined only in one direction($t\to \infty$), but with $P^\pm$, $ P^\pm\Omega_{t,\pm}^*$ exist on $\s^2_x(\mathbb{R}^5)$. Here we also use 
 the following time-dependent intertwining property
\eq
\Omega_{t,\pm}^*U(t,0)=e^{-itH_0}\Omega_{\pm}^*\quad\text{ on }\s^2_x(\mathbb{R}^5).
\eeq
\begin{lemma}\label{demain2}If $V$ satisfies the decay assumption ( with $\delta>6$) and the Assumption on eigenfunctions on the Floquet operator (polynomial decay of order $<x>^{-3}$), then
\eq
C_j(t)=C_{jm}(t)+C_{jr}(t),\quad j=1,2
\eeq
for some operators $C_{jm}(s,u)$ and $C_{jr}(s,u)$ satisfying
\eq
\sup\limits_{t\in \mathbb{R}}\| C_{jr}(t)\|_{\s^2_{x}\to \s^2_{x}}\leq 1/1000,\label{Cjrt}
\eeq
\eq
\sup\limits_{t\in \mathbb{R}} \| \langle x\rangle^{-\eta}(1-C_{1r}(t)-C_{2r}(t))^{-1}C_{jr}(t)\langle x\rangle^\sigma  \|_{\s^2_{x}\to \s^2_{x}}\lesssim_{\sigma} 1\label{Cjrtsigma}
\eeq
for all $\eta>5/2,\sigma\in (1,101/100)$ and
\eq
\sup\limits_{t\in \mathbb{R}} \| C_{jm} U(t,0)\Omega_+ \phi(x)\|_{\s^2_{x,t}}\lesssim \| \phi\|_{\s^2_x(\mathbb{R}^5)}.\label{Cjmt}
\eeq
\end{lemma}
Then, according to a similar argument as we did for the time-independent system, using Lemma \ref{demain2} instead of Lemma \ref{cpt}, we get the same LD estimate as in the time independent case.

In summary, the key results of this approach is getting LD estimates for  (localized in space) potentials, with quasi-periodic time dependence. The result hold for all energies, with no restriction on thresholds (in 5 or more dimensions). This approach will play an important role in what we do next.

\section{ Properties of the Asymptotic Solutions of General NLS type Equations}

We now turn to the general case.

 We use the notation $A\lesssim_s B$ and $A\gtrsim_s B$ to indicate that there exists a constant $C=C(s)>0$ such that $A\leq CB$ and $A\geq CB$, respectively.
When the interaction $\mathcal{N}$ is uniformly bounded in $x$, one can show that the weakly localized part is smooth. Specifically, in Theorem \ref{thm22}, we assume that the interaction $\mathcal{N}(|\psi|,|x|,t)\in \s^\infty_t\s^2_x(\mathbb{R}^{n+1})$ satisfies for all $f\in \s^\infty_tH^1_x(\mathbb{R}^{n+1})$,
\eq
|\mathcal{N}(|f|,|x|,t)|\lesssim_{\|f \|_{\s^\infty_tH^1_x(\mathbb{R}^{n+1})}} \frac{1}{\langle x\rangle^{\sigma}}\quad \text{ for some }\sigma >6:\label{con1}
\eeq
\begin{assumption}\label{asp1} \eqref{con1} holds.
\end{assumption}
Let $V(x,t):=\mathcal{N}(|\psi(t)|,|x|,t)$, $\psi_D(t):=\psi(t)-e^{-itH_0}\psi(0)$, and $V_D(x,t):=\mathcal{N}(|\psi_D(t)|,|x|,t)$. Due to the assumed $H^1$ boundedness on the initial data and the solution, \newline we have $\|\psi_D(t)\|_{\s^\infty_tH^1_x(\mathbb{R}^{n+1})}\leq 2E$.
\begin{assumption}\label{asp4}For all $j=1,2,\cdots,n$, let $P_j:=-i\partial_{x_j}$. One has 
\eq\label{p3:2}
\| P_j \mathcal{N}(|\psi(t)|,|x|,t)\psi(t)\|_{\s^\infty_t\s^2_{x,2}(\mathbb{R}^{n+1})}\lesssim_E 1.
\eeq
\end{assumption}
\begin{lemma}\label{LempsiD}If the $H^1$ norm is uniformly bounded and \eqref{con1} and  \ref{p3:2} are satisfied, then we have
 \eq
 \|\langle x\rangle^{-2}\langle P\rangle^{3/2}\psi_D(t)\|_{\s^\infty_t\s^2_x(\mathbb{R}^{n+1}) }\lesssim_E 1 .
 \eeq
 \end{lemma}

Based on Lemma \ref{LempsiD}, we know that  $\langle x\rangle^{-2}\psi_D(t)\in H^{3/2}_x(\mathbb{R}^n)$. Therefore, in the proof of smoothness, we construct $\psi_{loc}(t)$ in terms of $\psi_D(t)$ instead of $\psi(t)$.

We also need the following assumptions related to $\psi_D(t)$:
\begin{assumption}\label{asp2}For some $\sigma>2$ and $\sigma'\in(0,1/2)$ (here we take $\sigma'<1/2$ since we only have limited decay in $|x|$ from $V(x,t)$), we have
\begin{multline}
\|V(x,t)\psi(t)-V_D(x,t)\psi_D(t) \|_{\s^2_{x,\sigma}(\mathbb{R}^n)\cap \s^1_{x,2}(\mathbb{R}^n)}\lesssim_{E} \\\|e^{-itH_0}\psi(0)\|_{\s^2_{x,-\sigma'}(\mathbb{R}^n)}+\|e^{-itH_0}\psi(0)\|_{\s^2_{x,-\sigma'}(\mathbb{R}^n)}^{2/n}\label{con2}.
\end{multline}
\end{assumption}
\begin{remark}Note that $\psi(t)-\psi_D(t)=e^{-itH_0}\psi(0)$ and both $V(x,t)$ and $V_D(x,t)$ are localized functions due to Assumption \eqref{asp1}. That is how we get a factor $\|e^{-itH_0}\psi(0)\|_{\s^2_{x,-\sigma'}(\mathbb{R}^n)}$. In $V(x,t)\psi(t)-V_D(x,t)\psi_D(t)$, there is another term $(V(x,t)-V_D(x,t))\psi_D(t)$. Here we could get a factor $e^{-itH_0}\psi(0) \times \psi_D(t)$ or $\left(e^{-itH_0}\psi(0)\right)^* \times \psi_D(t)$. When $|x|\leq 1$, $\psi_D(t)$ is not uniformly bounded, so in this case, $\s^2_{x,\sigma}(\mathbb{R}^n)$ cannot be controlled by $\|e^{-itH_0}\psi(0)\|_{\s^2_{x,-\sigma'}(\mathbb{R}^n)}$. On the other hand, $$|(V(x,t)-V_D(x,t))\psi_D(t)|\lesssim_E |\psi_D(t)|.$$ Using interpolation, one can obtain, for example,
\eq
\chi(|x|\leq 1)|(V(x,t)-V_D(x,t))\psi_D(t)|\lesssim_E  \chi(|x|\leq 1)|e^{-itH_0}\psi(0) |^{2/n}|\psi_D(t)|.
\eeq
Using that $\sup\limits_{t\in \mathbb{R}}\|\psi_D(t)\|_{\s^{2n/(n-2)}_x(\mathbb{R}^n)}\lesssim \sup\limits_{t\in \mathbb{R}}\|\psi_D(t)\|_{H^{1}_x(\mathbb{R}^n)}\lesssim_E 1$, by H\"older's inequality $(\frac{1}{n}+(\frac{1}{2}-\frac{1}{n})=\frac{1}{2})$, one has
\eq
\begin{split}
&\|\chi(|x|\leq 1)(V(x,t)-V_D(x,t))\psi_D(t)\|_{\s^2_{x,\sigma}(\mathbb{R}^n)}\\ 
\lesssim_E& \| |e^{-itH_0}\psi(0) |^{2/n}\|_{\s^n_{x,-\sigma'}(\mathbb{R}^n)}\|\psi_D(t)\|_{\s^{2n/(n-2)}_x(\mathbb{R}^n)}\\
\lesssim_E& \|e^{-itH_0}\psi(0)\|_{\s^2_{x,-\sigma'}(\mathbb{R}^n)}^{2/n}.
\end{split}
\eeq
\end{remark}
\begin{assumption}\label{asp3}
\eq
\| \langle P\rangle^{3/2}\langle x\rangle^6 V_D(x,t)\psi_D(t)\|_{\s^2_x(\mathbb{R}^n)}\lesssim_{E}1.\label{con5}
\eeq
\end{assumption}
{Under the aforementioned assumptions, if $\mathcal{N}(|\psi(t)|,|x|,t)\in \s^\infty_t\s^2_x(\mathbb{R}^{n+1}),$  and uniform boundedness of the $H^1$ is satisfied, then the solution decomposes to a free and localized parts and it  holds,
\eq
\| A^2\psi_{loc}(t)\|_{\s^2_{x}(\mathbb{R}^n)}\lesssim_E 1,
\eeq
where $P:=-i\nabla_x$ and $A:=\frac{1}{2}(P\cdot x+x\cdot P)$:
\begin{theorem}[Soffer-Wu]\label{thm22}  Let $\psi(t)$ denote a solution to the system \eqref{SE} that is uniformly bounded in $H^1$  and   satisfies \eqref{con1}. If $\mathcal{N}(|\psi|,|x|,t)\in \s^\infty_t\s^2_{x,2}(\mathbb{R}^{n+1})$ satisfies assumptions \ref{asp1}-\ref{asp3}, then the asymptotic decomposition holds, with $\psi_{loc}(t)$ satisfying \eqref{local} and \eqref{smooth}. In particular, when $n\geq 45$, $\sigma> n/2$ in Assumption \ref{asp1}, one has that for some $\delta>1$,
\eq
\| \langle x\rangle^\delta \psi_{loc}(t)\|_{\s^2_x(\mathbb{R}^n)}\lesssim 1 .\label{locallast}
\eeq
\end{theorem}
\cite{soffer2023soliton}.

Typical examples of Theorem \ref{thm22} are
\eq
\mathcal{N}=- \frac{\lambda |\psi|^p}{1+|\psi|^p}\quad \text{ for all }p> 3, \lambda>0, n=5,
\eeq

\subsubsection{About the proofs}
` Here is the {\bf detailed outline} of the proof for the existence of the free channel wave operators and Theorem \ref{thm22}. The proof of the free wave decomposition is based on the approach initiated in \cite{Sof-W5}. We begin by decomposing the solution into two parts:
\eq
\psi(t)=F(|x|\geq 10)\psi(t)+F(|x|< 10)\psi(t).
\eeq
Note that $F(|x|< 10)\psi(t)$ is localized in $x$. For $F(|x|\geq 10)\psi(t)$, we use an incoming/outgoing decomposition:
\eq
F(|x|\geq 10)\psi(t)=P^+F(|x|\geq 10)\psi(t)+P^-F(|x|\geq 10)\psi(t).
\eeq
To approximate $P^\pm F(|x|\geq 10)\psi(t)$, we use $P^\pm \Omega_{t,\pm}^*\psi(t)$:
\begin{enumerate}
\item A key observation here is that
\eq
\Omega_{t,\pm}^*\psi(t)=w\text{-}\lim\limits_{s\to \pm\infty} e^{isH_0}F(|x|\geq 10)\psi(t+s).
\eeq
 Using the intertwining property, one has
\eq
P^\pm \Omega_{t,\pm}^*\psi(t)=P^\pm e^{-itH_0}\Omega_{\pm}^*\psi(0),
\eeq
which are close to the free flow.
\item We write $P^\pm F(|x|\geq 10)\psi(t)$ as
\eq
P^\pm F(|x|\geq 10)\psi(t)=P^\pm e^{-itH_0}\Omega_{\pm}^*\psi(0)+C_\pm(t)\psi(t)
\eeq
where 
\eq
\begin{split}
C_\pm (t):=&P^\pm F(|x|\geq 10)-P^\pm \Omega_{t,\pm}^*\\
=&\pm i\int_0^\infty ds P^\pm e^{\pm isH_0} F(|x|\geq 10)V(x,t\pm s)U(t\pm s,t)+\\
&(\mp i)\int_0^\infty ds P^\pm e^{\pm is H_0} [H_0, F(|x|\geq 10)] U(t\pm s,t).
\end{split}
\eeq
\item Let $\psi_{free}=\Omega_+^*\psi(0)$ and define $\psi_{loc}(t)$ as $\bar{F}(|x|<10) \psi(t)+(C_+(t)+C_-(t))\psi(t)$.  It is left to show that 
\eq
\| \langle x\rangle^\delta(C_+(t)+C_-(t))\psi(t) \|_{\s^2_x(\mathbb{R}^n)}\lesssim_E 1 \text{ for some }\delta>0.
\eeq
This is mainly accomplished by proving that
\eq
\int_0^\infty ds\| \langle x\rangle^\delta P^\pm e^{\pm i sH_0} \|_{\s^2_{x,\sigma}(\mathbb{R}^n)\cap \s^1_x(\mathbb{R}^n)\to\s^2_x(\mathbb{R}^n) }\lesssim_{\sigma,\delta }1\label{4eq1}
\eeq
for some $\delta >0$, where $P:=-i\nabla_x$.
\end{enumerate}
\begin{remark}Here we use the weak limit for $\Omega_{t,\pm}^*$, since the strong limit of $e^{isH_0}\psi(t+s) $ does not exist in general when there is a soliton.
\end{remark}
\begin{remark}Fortunately, the limit $s\text{-}\lim\limits_{s\to \pm\infty} P^\pm e^{isH_0}\psi(t+s) $ exists in $\s^2_x(\mathbb{R}^n)$, therefore there is no confusion in using $\Omega_{t,\pm}^*$.
\end{remark}
\begin{remark}Observe that the weakly localized part only spreads significantly when $|P|$ is close to zero. Here, $P$ refers to the momentum of the weakly localized part. The most challenging part of the argument is showing that the zero-frequency part does not delocalize, that is,
\eq
\| \langle x\rangle^\delta F(|P|\leq \epsilon)\psi_{wl}(t)\|_{\s^2_x(\mathbb{R}^n)}\lesssim 1
\eeq
for some $\delta>0,\epsilon>0$.
\end{remark}
\begin{remark}In $\int_0^\infty ds \| P^\pm e^{\pm isH_0}\|_{\s^2{x,\sigma}(\mathbb{R}^n)\cap \s^1_x(\mathbb{R}^n)\to \s^2_x(\mathbb{R}^n)}$, the part of the solution with momentum (frequency) of order $\frac{1}{\langle s\rangle^{\epsilon}}$ has a total traveling distance $|x|$ of order $\langle s\rangle^{1-\epsilon}\sim s\times |P|$. Using $P^\pm$ and the method of stationary phase, each integration by parts gains $\frac{1}{\langle s\rangle^{1-\epsilon}}$ decay and loses $\langle s\rangle^\epsilon$ due to the cut-off frequency. Therefore, $\epsilon =\frac{1}{2}$ is the borderline. Fortunately, in $5$ or higher space dimensions ($n\geq 5$), we have 
\eq
\| P^\pm e^{\pm isH_0}F(|P|\leq \frac{1}{\langle s\rangle^{1/2-0}})\|_{\s^2_{x,\sigma}(\mathbb{R}^n)\cap \s^1_x(\mathbb{R}^n)\to \s^2_x(\mathbb{R}^n)}\lesssim \frac{1}{\langle s\rangle^{\frac{n}{4}-n/2\times 0 }}\in \s^1_s(\mathbb{R})
\eeq
by using Hardy-Littlewood-Sobolev inequality. So in this case, $n=4$ is the borderline, which is why this method is not workable when $n\leq 4$.
\end{remark}
\begin{remark}
In general, we cannot use propagation estimates with $P^\pm=P^\pm(A)$ (where $A$ is the dilation operator) when $|P|\sim \frac{1}{\langle s\rangle^\epsilon}$, since in this region we have $|A|\lesssim |P||x|\sim \langle s\rangle^{-\epsilon }\cdot \langle s\rangle^\epsilon =O(1)$.
\end{remark}
\begin{remark}
We can approximate $P^-F(|x|>10)\psi(t)$ by $P^-F(|x|>10)e^{-itH_0}\psi(0)$, and the error term becomes 
\eq
c\int_0^tdsP^-e^{-i(t-s)H_0}V(x,s)\psi(s)
\eeq
for some $c>0$, using Duhamel's formula. Since $t-s \geq 0$, we can use a similar argument based on estimates of the free flow, and the result follows.
\end{remark}
{The proof of Theorem \ref{thm22} is similar. We use incoming/outgoing decomposition of $\psi(t)$:
\begin{align}
\psi(t)=&P^+\psi(t)+P^-\psi(t)\\
=&P^+e^{-itH_0}\Omega_+^*\psi(0)+P^-e^{-itH_0}\Omega_-^*\psi(0)+P^+(1-\Omega_{t,+}^*)\psi(t)+P^-(1-\Omega_{t,-}^*)\psi(t)\\
=:&P^+e^{-itH_0}\Omega_+^*\psi(0)+P^-e^{-itH_0}\Omega_-^*\psi(0)+C_r(t)\psi(0).
\end{align}
The localized part is defined by
\eq
\psi_{loc}(t)=\tilde{C}_+(t)\psi(0)+\tilde{C}_-(t)\psi(0),
\eeq
where
\eq
\tilde{C}_+(t)\psi(0):=i\int_0^\infty ds P^+e^{isH_0}V_D(x,t+s)\psi_D(t+s)
\eeq
and
\eq
\tilde{C}_-(t)\psi(0):=(-i)\int_0^\infty ds P^-e^{-isH_0}V_D(x,t-s)\psi_D(t-s).
\eeq
We use $\psi_D(t)$ instead of $\psi(t)$ because $\psi_D(t)$ is smoother than $\psi(t)$ when we localize in space.

The difference between $\psi(t)$ and $\psi_D(t)$ is the free flow $e^{-itH_0}\psi(0)$, which is easy to control. In fact, $\langle x\rangle^{-2}\langle P\rangle^{1+l}\psi_D\in\s^2_x(\mathbb{R}^5)$ for any $l\in [0,1)$. In other words, the advantages of using $\psi_D(t)$ to define $\psi_{loc}(t)$ are:
\begin{itemize}
\item $\psi_D(t)$ is smoother than $\psi(t)$.
\item It is easy to control $\psi(t)-\psi_D(t)=e^{-itH_0}\psi(0)$, which satisfies  dispersive estimates. 
\end{itemize}
We plan to show that
\eq
\| A^2\psi_{loc}(t)\|_{\s^2_{x}(\mathbb{R}^n)}\lesssim_E 1\label{smooth},
\eeq
and 
\eq
\|\langle x\rangle^\delta \psi_{loc}(t)\|_{\s^\infty_t\s^2_x(\mathbb{R}^{n+1})}<\infty.\label{local}
\eeq
For $\psi_{loc}(t)$, we obtain \eqref{local} and \eqref{smooth} by using the estimates for $P^\pm e^{\pm isH_0}$ (for $s\geq 0$) acting on localized functions. We are able to prove \eqref{smooth} because $\psi_D(t)$ is smooth.

\begin{remark}Using Duhamel's formula (see \cite{T2006}), we can rewrite $\psi_D(t)$ as follows:
\eq
\psi_D(t)=(-i)\int_0^t ds e^{-i(t-s)H_0}V(x,s)\psi(s).
\eeq
The integration over $s$ makes $\psi_D(t)$ smoother than $\psi(t)$ locally in space. One can use Duhamel's formula to iterate it again and achieve even greater smoothness. In other words, the gain of  power of derivatives is not optimal, and in most situations, it can be improved to $k>2$ by utilizing the Duhamel iteration strategy.
\end{remark}

\subsubsection{$A$ localization}

Based on the localization properties we have, we can now control powers of $A$:

{\bf Estimate for $A^2\psi_{loc,1}(t)$:} For $A^2\psi_{loc,1}(t)$, we have
\begin{subequations}
    \eq
    i[H_0,A]=2H_0;
    \eeq
    \eq
    (-i)[H_0,A^2]=-4H_0A-4iH_0;
    \eeq
    \eq
    (-i)[H_0, (-i)[H_0,A^2]]=8H_0^2;
    \eeq
    \begin{align}
    A^2e^{isH_0}=&e^{isH_0}A^2+(-i)\int_0^s du e^{-iuH_0}[H_0,A^2]e^{iuH_0}\nonumber\\
    =&e^{isH_0}A^2+(-i)s[H_0,A^2]+(-i)\int_0^sdu\int_0^udv e^{-ivH_0}(-i)[H_0,[H_0,A^2]]e^{ivH_0}\nonumber\\
    =&e^{isH_0}A^2+(-4H_0A-4iH_0)s+4H_0^2s^2,
    \end{align}
\end{subequations}
which implies that
\begin{align}
A^2\psi_{loc,1}(t)=&i\int_0^\infty ds P^+ A^2 e^{ i sH_0}V_D(x,t+ s)\psi_D(t+ s)\nonumber\\
=&i\int_0^\infty ds P^+\left(4s^2e^{isH_0}(-\Delta)^2-4se^{isH_0}(-\Delta)A+(-4is)e^{isH_0}(-\Delta)A\right.\nonumber\\
&\left.+e^{isH_0}A^2\right)\times V_D(x,t+ s)\psi_D(t+ s).
\end{align}
Break $A^2\psi_{loc,1}(t)$ into two pieces
\begin{align}
A^2\psi_{loc,1}(t)=&i\int_0^\infty ds P^+F(|P|\leq 1)\left(4s^2e^{isH_0}(-\Delta)^2-4se^{isH_0}(-\Delta)A+\right.\nonumber\\
&\left.(-4is)e^{isH_0}(-\Delta)A+e^{isH_0}A^2\right)V_D(x,t+ s)\psi_D(t+ s)\nonumber\\
&+i\int_0^\infty ds P^+F(|P|> 1)\left(4s^2e^{isH_0}(-\Delta)^2-4se^{isH_0}(-\Delta)A\right.\nonumber\\
&\left.+(-4is)e^{isH_0}(-\Delta)A+e^{isH_0}A^2\right)V_D(x,t+ s)\psi_D(t+ s)\nonumber\\
=:&C_{m,+,1}(t)\psi(0)+C_{m,+,2}(t)\psi(0).
\end{align}
For $C_{m,+,1}(t)\psi(0)$, according to \eqref{Jan23.1} in Lemma \ref{out/in1}, one has that when $\sigma>6$, $\sigma+a-2>4$ and $V_D(x,t)\in \s^2_t\s^2_{x,2}(\R^{n+1})$, for all $a=0,1,2,$, we have
\begin{align}
\|C_{m,+,1}(t)\psi(0)\|_{\s^2_x(\mathbb{R}^n)}\lesssim\\
&\sum\limits_{a=0}^2 \int_0^\infty ds s^a\| P^+e^{isH_0}(-\Delta)^aF(|P|\leq 1)\|_{\s^2_{x,\sigma+a-2}(\mathbb{R}^n)\cap \s^1_x(\mathbb{R}^n)\to\s^2_x(\mathbb{R}^n) }\nonumber\\
&\times\|  F(|P|\leq 10) A^{2-a} V_D(x,t+s)\psi_D(t+s)\|_{\s^2_{x,\sigma+a-2}(\mathbb{R}^n)\cap \s^1_x(\mathbb{R}^n)}\nonumber\\
\lesssim & \|V_D(x,t+s)\psi_D(t+s)\|_{\s^2_{x,\sigma}(\mathbb{R}^n)\cap \s^1_{x,2}(\mathbb{R}^n)} \nonumber\\
\lesssim_{E,\|V(x,t)\|_{\s^\infty_t\s^2_{x,2}(\R^{n})}}& 1.\label{Cm+1}
\end{align}
Here we  used the assumption that the solution is uniformly bounded in $H^1,$ 
\eq
\sup\limits_{t\in \mathbb{R}}\| \psi_D(t)\|_{H^1_x(\mathbb{R}^n)}\leq 2E<\infty.
\eeq
Hence,
\eq
\| V_D(x,t+s)\psi_D(t+s)\|_{\s^2_{x,\sigma}(\mathbb{R}^n)\cap \s^1_{x,2}(\mathbb{R}^n)}\lesssim_{E,\|V(x,t)\|_{\s^\infty_t\s^2_{x,2}(\R^{n})}}1
\eeq
by using the decay assumption of the interaction term. For $C_{m.+,2}(t)\psi(0)$, according to \eqref{Sep29.1} and \eqref{Oct.1} in Lemma \ref{out/in1}, using the regularity assumption on the localized part of $\psi_D,$ we have that for some $\epsilon>0$ close to $0$,
\begin{align}
&\|C_{m,+,2}(t)\psi(0)\|_{\s^2_x(\mathbb{R}^n)}\lesssim\\
&\sum\limits_{a=0}^2 \int_0^\infty ds s^{a}\| P^+e^{isH_0}F(|P|>1)|P|^{a+1-\epsilon}\|_{\s^2_{x,3+\epsilon}(\mathbb{R}^n)\to\s^2_x(\mathbb{R}^n) }\nonumber\\
&\times\| |P|^{a-1+\epsilon} A^{2-a}V_D(x,t+s)\psi_D(t+s)\|_{\s^2_{x,3+\epsilon}(\mathbb{R}^n)}\nonumber\\
\lesssim &\sum\limits_{a=0}^2 \int_0^\infty ds s^{a}\| P^+e^{isH_0}|P|^{a+1-\epsilon}\|_{\s^2_{x,3+\epsilon}(\mathbb{R}^n)\to\s^2_x(\mathbb{R}^n) }\nonumber\\
&\times \| \langle P\rangle^{3/2}\langle x\rangle^{5+\epsilon} V_D(x,t+s)\psi_D(t+s)\|_{\s^2_{x}(\mathbb{R}^n)} \nonumber\\
\lesssim_{\epsilon,E} &1,\label{Cm+2}
\end{align}
where we use
$$
\| \langle x\rangle^{3+\epsilon} |P|^{a-1+\epsilon} A^{2-a}\langle x\rangle^{-(5+\epsilon)} \langle P\rangle^{-3/2}\|_{\s^2_{x}(\mathbb{R}^n)\to\s^2_{x}(\mathbb{R}^n)}\lesssim 1.
$$
Based on \eqref{Cm+1} and \eqref{Cm+2}, one has
\eq
\|A^2\psi_{loc,1}(t)\|_{\s^2_x(\mathbb{R}^n)}\lesssim_E 1.\label{Aloc1}
\eeq
Similarly, one has
\eq
\|A^2\psi_{loc,2}(t)\|_{\s^2_x(\mathbb{R}^n)}\lesssim_E 1.\label{Aloc2}
\eeq
Based on \eqref{Aloc1} and \eqref{Aloc2}, one has
\eq
\| A^2 \psi_{loc}(t)\|_{\s^2_x(\mathbb{R}^n)}\lesssim_E 1.
\eeq

\subsection{Asymptotic Completeness}

Asymptotic Completeness is a key result of scattering theory. It means that we know all possible asymptotic states of the system.
Therefore, for a given system, to prove AC one needs to have an Ansatz.

In N-body QM, the Ansatz is well known: All asymptotic states are linear combinations of independently moving stable clusters of subsystems.

Implicit in it is that each cluster is a solution of the subhamiltonian, and that each cluster is moving with a constant speed. Moreover, each cluster is a bound state, that is a linear combination of eigenstates.

This was proved in\cite{SS1987}, for short range potentials. 
For long range potentials, there are modifications of the asymptotic states. Therefore the statement has to be modified.
The free flow needs to be replaces by a free flow plus a divergent phase corrections (which correspond to time delay as compared for the free flow).

AC is a crucial result in Modern Physics. It implies that the scattering matrix that maps incoming scattering states to outgoing is Unitary.

If not, it would mean that crushing two atoms on each other, will result in a collection of fragments, which do not add up to the original state. That is an extra dark state is still there, undetected.
This absurd situation is therefore sometimes called a \emph{proof} of AC.

When it comes to linear \emph{time dependent potentials}, there is no clear Ansatz.
When it comes to nonlinear dispersive and hyperbolic equations, the exact solutions and numerics indicate that the asymptotic states are free waves and solitons. 

So a natural Ansatz is that all asymptotic states are solitons and free waves, This is now called "Soliton Resolution". However, we must know if this is generic or holds for all initial data.

The examples we know from one dimension show that it can only be generic:
KdV equations can have a solution that converges to two soliton state, yet a pure two soliton state is not a solution\cite{martel2011description}. It is however true that this solution is a sum of two exact solutions. NLKG equation can have a breather solution, which is a localized and with time dependent amplitude. The NLS system can have many breathers which are quasi-periodic in time.
So is the linear Schr\"odinger equation: for a localized time dependent potentials, is has localized time dependent solutions.

Similar situations in higher dimension are hard to come by. See the next section where a complete breakdown of AC is demonstrated, and the section on open problems.

On the positive side, let us consider what should be true for general scattering systems of the Dispersive and Hyperbolic type.
Suppose there is a breather solution, which is time periodic or quasi-periodic. Then, the corresponding Floquet operator is a linear operator $K$ which is self-adjoint and with "time-independent" potential.

In the periodic case we have:
$$
K=-\Delta -ci\partial_t +V(x,t),
$$
acting on the Hilbert space $L^2(\R^n_x \times [0,2\pi]) $

Since the spectrum of $-i\partial_t$ on $L^2([0,2\pi])$ is a constant times all integers, the spectrum of $K$ is the union of the positive real lines shifted by any integer.
Therefore the breather solution corresponds to an eigenfunction with an eigenvalue embedded in infinitely many copies of the positive real line.

Such a solution will not be stable under a small perturbation, unless \emph{ infinitely many Fermi Golden Rules are {\bf not} satisfied}. That is, infinitley many non linear relations should be satisfied (sounds like completely integrable system, which we do not know of any in higher dimensions).

    So, it is safe to assume that such solutions are non-generic.
What about non-quasi-periodic solutions? Chaotic?
An Ansatz about that \cite{Sof1}, the Petite Conjecture, states that a localized solution which is also regular enough, of a dispersive equation, must be an (asymptotically) almost periodic function of time.
Since we proved that general solutions in 5 dimensions are localized and smooth for a large class of interactions terms, it is reasonable to expect that soliton resolution will hold in the generic sense for such equations.

\section{Self Similar Solutions}

                                                   { \footnotesize \hspace{6cm} Someone has said, Derivations\\  \hspace*{6cm} in dispersion theory are like \newline \hspace*{6cm} a man's teats; they are\newline  \hspace*{6cm}  neither useful nor ornamental.\newline 
\hspace*{6cm}  Goldberger and Watson in "Collision Theory",\newline
                                                    \hspace*{6cm}  paraphrasing James J. Hill.}\\
\subsection{Introduction}                                                    
In the analysis of the scattering of General NLS equations, we find that there is no clear way to prove that the non-scattering states are localized in space, for spherically symmetric solutions, or in the case the interaction term is localized in space.

So, already in \cite{liu2023large} we saw that the exception can be a self similar solution that is slowly spreading in space.
Tao\cite{T2004} have shown that in the mass supercritical case, the weakly localized solutions must have a heavy core around the origin. We have shown, as explained before, that in 5 or more dimensions, with sufficient decay of the Interaction term, the solutions are indeed localized, at least in 5 or more dimensions.

However, we now find out that there is a large class of equations (in 5 or more dimensions) for which there are self similar solutions, which are \emph{global and stable.}

We show that for potentials that decay like $|x|^{-2}$ at infinity, uniformly in time, there are massive (in $L^2$) solutions which are stable, asymptotically stable and can appear together with localized solutions and free waves at the same time.
They can scale like $|x|\sim t^{\alpha}, 0<\alpha<1/2.$
Perhaps also $\alpha=1/2$ is allowed.
These solutions carry no kinetic energy ($\dot H^1,$ the homogeneous Sobolev norm, converges to zero).
They are thus dark solutions (Zombie states).
This shows that AC breaks down in a fundamental, generic way, for such equations.

Therefore I say, rigorous mathematical analysis has the last word.

\subsection{The self similar models}
We start with a linear model
\eq
\begin{cases}
i\partial_t\psi=H_0\psi+g(t)^{-2} V(\frac{x}{g(t)})\psi \\
\psi(x,g(t_0))=e^{- iA \ln g(t_0)}\psi_b(x)\in \s^2_x(\mathbb{R}^n)
\end{cases}\label{linear}, \quad n\geq 3
\eeq
for some $t_0>0$ \, ($t_0$ will be chosen later), $g(t)\in C^2(\mathbb{R})$ satisfying that there exists two positive constants $c_g\in (0,1), \epsilon\in (0,1/2)$ such that
\eq
\begin{cases}
\inf\limits_{t\in \mathbb{R}}g(t)\gtrsim 1,\\
 g(t)\sim \langle t\rangle^{\epsilon}\text{ as }t\to \infty \\
g(t)\sim tg'(t)\sim t^2g''(t)\text{ as }t\to \infty \\
\lim\limits_{t\to \infty}\frac{g(t)-2tg'(t)}{g(t)}= c_g
\end{cases}\label{g(t)},
\eeq
and $V(x)$ and $H:=H_0+V(x)$ satisfying that $H$ has a unique normalized eigenstate $\psi_b(x)$ with an eigenvalue $\lambda<0$ and
\eq
\begin{cases}
0\text{ is regular for }H\\
\langle x \rangle A\psi_b(x)\in \s^2_x\\
\langle x\rangle V(x) \in \s^\infty_x, V(x)\in \s^2_x
\end{cases}\label{Linear:con}
\eeq
where $P_x:=-i\nabla_x$, $A:=\frac{1}{2}(x\cdot P_x+P_x\cdot x)$ and $P_c$ denotes the projection on the continuous spectrum of $H$. We refer to the system \eqref{linear} as mass critical system(\textbf{MCS}). 

Since $g(t)^{-2}V(\frac{x}{g(t)})\in \s^\infty_t\s^2_x(\mathbb{R}^3\times \mathbb{R})$ when $\inf\limits_{t} g(t)>0$, due to \cite{SW20221}, the channel wave operator
\eq
\Omega_\alpha^*:=s\text{-}\lim\limits_{t\to \infty}e^{itH_0}F_c(\frac{|x-2tP_x|}{t^\alpha}\leq 1)U(t,0)\label{wave}
\eeq
exists from $\s^2_x(\mathbb{R}^3)$ to $\s^2_x(\mathbb{R}^3)$ for all $\alpha \in (0, 1/3)$, where $F_c$ denotes a smooth characteristic function.

Based on \eqref{linear}, we also consider a class of mixture models
\eq
\begin{cases}
i\partial_t\psi=H_0\psi+W(x)\psi+g(t)^{-2}V(\frac{x}{g(t)})\psi\\
\psi(x,t_0)=\psi_d(x)+e^{-i A\ln(g(t_0) )}\psi_b(x)\in \mathcal{H}^1_x(\mathbb{R}^n)\\
\sup\limits_{t\in \mathbb{R}}\|\psi(t)\|_{H^1_x}\lesssim 1
\end{cases}\quad (x,t)\in \mathbb{R}^n\times \mathbb{R}\label{NP0}.
\eeq
$W(x)$ satisfies that
\eq
\begin{cases}
H_0+W(x)\text{ has a normalized eigenvector }\psi_{d}(x)\text{ with an eigenvalue }\lambda_0<0\\
W(x)\in \s^2_x(\mathbb{R}^n)
\end{cases}\label{W}.
\eeq
We showed in \cite{Sof-W3} that the weakly localized part, asymptotically, has at least two bubbles: a non-trivial self-similar part and a non-trivial localized part near the origin.

\begin{enumerate}
\item \label{A}Let
\eq
  \tilde{a}(t):= (\psi_b(x), e^{ iA \ln (g(t)) }\psi(t))_{\s^2_x}.\label{a}
\eeq
\eq
\tilde{A}(\infty):=\lim\limits_{t\to \infty} e^{i\lambda T(t)}\tilde{a}(t)\label{tildeA}
\eeq
exists.
\item Furthermore, there exists $ t_0>0$ such that with an initial condition
\eq
\psi(t_0)=e^{-i A\ln(g(t_0) )}\psi_b(x),
\eeq
\eq
|\tilde{A}(\infty)| >0
\eeq
which implies
\eq\liminf\limits_{t\to \infty}|c(t)|\gtrsim 1.\label{eq10}
\eeq
\end{enumerate}

\begin{theorem}\label{thm10}Let $\tilde{a}(t)$ be as in \eqref{a}. If $W(x),V(x), H$ satisfy \eqref{NP0} and $g(t)$ satisfies \eqref{g(t)}, then when $n\geq 5$, $\epsilon\in(2/n, 1/2)$,
\eq
\tilde{A}(\infty):=\lim\limits_{t\to \infty} e^{i\lambda T(t)}\tilde{a}(t)
\eeq
exists and
\eq
\psi_{w,l}(x,t)=c(t)e^{-i A\ln(g(t) )}\psi_b(x)\oplus\psi_{c}(x,t)
\eeq
\eq
c(t):=(e^{-i A\ln(g(t) )}\psi_b(x),  \psi_{w,l}(x,t))_{\s^2_x},
\eeq
\eq
( e^{-i D\ln(g(t) )}\psi_b(x), \psi_c(x,t))_{\s^2_x}=0,
\eeq
where $|c(t)|\gtrsim 1$  and there exists $M>1$ such that
\eq
\liminf_{t\to \infty} | ( \psi_{c}(x,t),\psi_d(x))_{\s^2_x}|\geq c.\label{cd0}
\eeq
Moreover, the $g(t)$-self-similar channel wave operator
\eq
\Omega_{g}^*\psi(0):=w\text{-}\lim\limits_{s\to \infty}e^{isH}e^{ iA\ln (g( T^{-1}(s)))}\psi(T^{-1}(s))
\eeq
exists in $\s^2_x$ and
\eq
\Omega_{g}^*\psi(0)=\tilde{A}(\infty)\psi_b(x).
\eeq
\end{theorem}
 
These results extend to nonlinear perturbations (mass super-critical) of the model above:

\eq
\begin{cases}
i\partial_t\psi=H_0\psi+g(t)^{-2}V(\frac{x}{g(t)})\psi+\mathcal{N}(|\psi|)\psi\\
\psi(x,t_0)=\psi_s(x)+e^{-i A\ln(g( t_0) )}\psi_b(x)\in \mathcal{H}^1_x(\mathbb{R}^n)\\
\text{There is a global }H^1_x \text{ solution }\psi(t)\\
\sup\limits_{t\in \mathbb{R}}\|\psi(t)\|_{H^1_x}\lesssim_{\psi(t_0)} 1\\
\text{Both }\psi(t_0)\text{ and }V(x)\text{ are radial in }x\\
\mathcal{N}\leq 0
\end{cases},\quad (x,t)\in \mathbb{R}^n\times \mathbb{R}\label{NP},
\eeq
when $n\geq 5$. Assume that $V(x),\psi_{b}(x)$ satisfy \eqref{Linear:con}($H:=H_0+V(x)$) and
\eq
\begin{cases}
\| A\psi_b(x)\|_{\s^\infty_x}\lesssim 1\\
\|\frac{1}{\lambda -H}P_c\|_{\s^\infty_x\to \s^\infty_x}\lesssim 1
\end{cases},\label{Linear:con100}
\eeq
and $\psi_s(x)$ is a soliton of
\eq
i\partial_t\phi=H_0\phi+\mathcal{N}_{F,0}(|\phi|)\phi,\label{NF00}
\eeq
where
\eq
\mathcal{N}_{F,0}(k):=\int_0^k dq q \mathcal{N}(q)/k^2<0,
\eeq
that is,
\eq
(H_0+2\mathcal{N}_{F,0}(|\psi_s(x)|))\psi_s(x)=E\psi_s(x), \text{ for some }E<0.
\eeq
If the nonlinearity satisfies that
\begin{enumerate}
\item there exists $T\geq 1$ such that for all $t_0\geq T$,
\eq
(\psi(t_0),(H_0+\frac{1}{g(t_0)^{2}}V(\frac{x}{g(t_0)})+2\mathcal{N}_{F,0}(|\psi(t_0)|))\psi(t_0)  )_{\s^2_x}\leq \frac{E}{2}\|\psi_s(x)\|_{\s^2_x}^2\label{Rt0}
\eeq
with $\psi(t_0)=\psi_s(x)+e^{-i A\ln(g(t_0))}\psi_b(x)$,
\item $\mathcal{N}$ satisfies that
\eq
| \mathcal{N}_{F,0}(k)|\lesssim |k|^\beta,\text{ for some }\beta>0\label{NFandN}
\eeq
and for $f\in H^1_x$,
\eq
\begin{cases}
\|\mathcal{N}(|f(x)|)\|_{\s^2_x}\leq C(\|f(x)\|_{H^1_x})\\
\| \mathcal{N}(|f(x)|)f(x)\|_{\s^2_x}\leq C(\|f(x)\|_{H^1_x})\label{con:N}
\end{cases},
\eeq
\end{enumerate}
{\bf then there are at least two bubbles in $\psi(t)$, a solution to the system \eqref{NP}.}

 \begin{theorem}\label{thm2}Let $\tilde{a}(t)$ be as in \eqref{a}. If $\mathcal{N},V(x), H$ satisfy \eqref{NP} and $g(t)$ satisfies \eqref{g(t)}, then when $n\geq 5$, $\epsilon >(2/n,1/2)$,
\eq
\tilde{A}(\infty):=\lim\limits_{t\to \infty} e^{i\lambda T(t)}\tilde{a}(t)
\eeq
exists and
\eq
\psi_{w,l}(x,t)=c(t)e^{-i A\ln(g(t) )}\psi_b(x)\oplus\psi_{c}(x,t)
\eeq
\eq
c(t):=(e^{-i A\ln(g(t) )}\psi_b(x),  \psi_{w,l}(x,t))_{\s^2_x},
\eeq
\eq
( e^{-i A\ln(g(t) )}\psi_b(x), \psi_c(x,t))_{\s^2_x}=0,
\eeq
with $c(t)$ satisfying \eqref{eq10}. Furthermore, there exists some large number $M\geq 1$ such that
\eq
\liminf\limits_{t\to \infty} \|\chi(|x|\leq M) \psi(t)\|_{\s^2_x}\geq c'\label{con100}
\eeq
for some $c'>0$. Moreover, based on the orthogonality of the self-similar part to the free wave (see \ref{wself4}), we have that the $g(t)$-self-similar channel wave operator
\eq
\Omega_{g}^*\psi(0):=w\text{-}\lim\limits_{s\to \infty}e^{isH}e^{ iA\ln (g( T^{-1}(s)))}\psi(T^{-1}(s))
\eeq
exists in $\s^2_x$ and
\eq
\Omega_{g}^*\psi(0)=\tilde{A}(\infty)\psi_b(x).
\eeq
\end{theorem}
\textbf{Typical example }of Theorem \ref{thm2} is
\eq
\mathcal{N}(|\psi(t)|)=-\lambda \frac{|\psi(t)|}{1+|\psi(t)|^2},\quad g(t)=\langle t\rangle^{\epsilon}
\eeq
for some $\epsilon\in (2/5,1/2)$ and some sufficiently large $\lambda>0$ in $5$ space dimensions. In this case, by taking $\lambda>0$ large enough, there is a soliton to \eqref{NF00}. By using standard iteration scheme, there is a global $L^2$ solution to \eqref{NP} for any initial $H^1_x$ data. The $H^1_x$ norm of the solution is uniformly bounded in $t$ since this system has an asymptotic energy. See section 5 of\cite{Sof-W3}  for more details.
\eq
w\text{-}\lim\limits_{s\to \infty}P_ce^{isH}e^{ iD\ln (\langle T^{-1}(s)\rangle)}\psi(T^{-1}(s))=0\quad \text{ in }\s^2_x.\label{wself4}
\eeq

\subsection{Outline of the proof}\label{outline}

\begin{theorem}\label{thm1}Let $\tilde{a}(t)$ be as in \eqref{a}. If $V(x), H$ satisfy \eqref{Linear:con} and $g(t)$ satisfies \eqref{g(t)}, then when $n\geq 3$,
\eq
\tilde{A}(\infty):=\lim\limits_{t\to \infty} e^{i\lambda T(t)}\tilde{a}(t)
\eeq
exists and
\eq
\psi_{w,l}(x,t)=c(t)e^{-i D\ln(g(t) )}\psi_b(x)\oplus\psi_{c}(x,t)
\eeq
where
\eq
c(t):=(e^{-i D\ln(g(t) )}\psi_b(x),  \psi_{w,l}(x,t))_{\s^2_x},
\eeq
\eq
( e^{-i D\ln(g(t) )}\psi_b(x), \psi_c(x,t))_{\s^2_x}=0
\eeq
with $c(t)$ satisfying \eqref{eq10}. Moreover, based on \eqref{wself4}, the $g(t)$-self-similar channel wave operator
\eq
\Omega_{g}^*\psi(0):=w\text{-}\lim\limits_{s\to \infty}e^{isH}e^{ iD\ln (g( T^{-1}(s)))}\psi(T^{-1}(s))
\eeq
exists in $\s^2_x$ and
\eq
\Omega_{g}^*\psi(0)=\tilde{A}(\infty)\psi_b(x).
\eeq
\end{theorem}

For the linear problem, the proof scheme of Theorem \ref{thm1} is first to set
\eq
\tilde \phi(t):=e^{i A\ln(g(t) )}\psi(x,t),\label{tphi}
\eeq
with $\tilde \phi(t)$ satisfying
\eq
\begin{cases}
i\partial_t\tilde \phi=g(t)^{-2}H \tilde \phi-(\partial_t[g(t)]g(t)^{-1})A \tilde \phi\\
\tilde \phi(t_0)=\psi_b(x)
\end{cases}\label{Linear:1},
\eeq
see Lemma \ref{Lem1}. Secondly, using change of variables from $t$ to $s=T(t)$, \eqref{Linear:1} can be rewritten as
\eq
\begin{cases}
i\partial_s\phi=H\phi+f(s)A\phi\\
\phi(s_0)=\psi_b(x)
\end{cases}\label{Linear:2}
\eeq
by setting
\eq
\phi(s):=\tilde \phi(T^{-1}(s)), \quad t_0:=T^{-1}(s_0)
\eeq
where
\eq
f(s):=-(\partial_t[g(t)]g(t))\vert_{t=T^{-1}(s)}.\label{fs}
\eeq
Based on \eqref{g(t)}, we have
\eq
f(s)\sim  \frac{1}{\langle s\rangle}\quad\text{ and }\quad f'(s)\sim \frac{1}{\langle s\rangle^2},\label{fs2}.
\eeq
 So up to here, the problem is reduced to study the ionization problem ( see \cite{soffer1999ionization}). To be precise, it is reduced to study the asymptotic behavior of $a(s)$ with
\eq
a(s):=(\psi_b, \phi(s))_{\s^2_x}.\label{defa}
\eeq
Indeed
\eq
\tilde{a}(t)=a(T(t)).
\eeq
Let
\eq
A(s):=e^{i\lambda s }a(s).
\eeq
In the end, we show that the limit
\eq
A(\infty):=\lim\limits_{s\to \infty}e^{i\lambda s }a(s)
\eeq
exists, which implies
\eq
\tilde{A}(\infty)=A(\infty)
\eeq
exists since
\eq
\tilde{A}(T^{-1}(s))=A(s).
\eeq
And if we choose $s_0$ wisely(large enough),
\eq
|A(\infty)|\geq \frac{1}{2}>0
\eeq
and finish the proof.

For the nonlinear problem or the mixture problem, it is similar to the linear one except for the nonlinear term and $W(x)$ term. For these terms, we use
\eq
|(g(t)^{2} e^{-i A\ln(g(t) )}\psi_b(x), \mathcal{N}(|\psi|)\psi)_{\s^2_x}|\leq g(t)^{-(n/2-2)}\|\psi_b(x)\|_{\s^\infty_x}\| \mathcal{N}(|\psi|)\psi \|_{\s^1_x},
\eeq
\eq
|(g(t)^{2} e^{-i A\ln(g(t) )}\psi_b(x), W(x)\psi)_{\s^2_x}|\leq g(t)^{-(n/2-2)}\|\psi_b(x)\|_{\s^\infty_x}\| W(x)\psi \|_{\s^1_x},
\eeq
and
\eq
g(T^{-1}(u))^{-(n/2-2)}\sim \langle u\rangle^{-(n/2-2)\epsilon/(1-2\epsilon)}\in L^1_u[1,\infty)
\eeq
when $n\geq 5$ and $\epsilon\in(2/n, 1/2)$.

This follows stability analysis of coherent structures \cite{soffer1990multichannel,soffer2005theory}.

In order to prove the existence of another bubble near the origin, we use the fact that for these systems, there is an asymptotic energy which is negative.
Since the self-similar part carries no energy, there must be a part of the solution localized on the support of the potential $W.$

\subsection{Outlook}

While the current method of construction applies only  in five or more dimensions, it is anticipated that similar results hold in any dimension. In any dimension, the mass-critical equation possesses a self-similar solution. One must properly saturate the nonlinearity to obtain a stable solution. The decomposition into a free wave and a weakly localized part remains valid, at least in three dimensions and higher (in the spherically symmetric case). However, one must demonstrate that certain solutions do not collapse to a purely localized form. It may be simpler to generate more examples by considering systems of equations.

The physical significance of such solutions hinges on the type of models exhibiting such behavior. This also raises questions about the validity of the Inverse Scattering Meta-Theorem.

\section{Open Problems- Direct Scattering}

In this section, I discuss open problems inspired by the approach to scattering of general-type equations with large data. First, I address problems that hold general importance and are strategic to scattering theory. These problems are likely to require some new tools. Then, I describe some problems in inverse scattering theory for general-type equations. Finally, I mention some technical yet open questions that hopefully can be resolved with current tools.

 \subsection{Direct Scattering of General systems}

 \textbf{The Nature of Localized Solutions: Time Dependence}

This problem is, in some sense, the most important.

Simply stated: \emph{Does there exist a breather solution in higher dimensions?}

The standard "soliton resolution conjecture" says no. But of course, this can only be generically true. In one dimension, we have periodic breathers for the Nonlinear Klein-Gordon (NLKG) equation, see e.g., \cite{Sig} and cited references, Korteweg-de Vries (KdV) \cite{munoz2019breathers} and cited references, and systems of Nonlinear Schrödinger (NLS) equations \cite{MSW-1}. Moreover, we have linear time-dependent potentials that have time-dependent, periodic, and quasi-periodic localized solutions \cite{MSW-1}. More examples include NLKG equation on a lattice. Other time-dependent weakly localized solutions were described above in higher dimensions. Further results in this direction can be found in \cite{Pyke, Sig}.

However, we are missing a construction of a time-dependent localized solution in higher dimensions. As explained before, such a periodic solution would correspond to an embedded eigenvalue of the Floquet operator, and so it is expected to be unstable. This indicates non-genericity. But finding such a solution would point to a completely integrable equation in higher dimensions.

The problem can be further simplified by asking for a linear problem with a time-dependent potential and a localized time-dependent spatial solution, in higher dimensions, see also \cite{PW2022}.

 Here are some approaches to tackle this problem:

\textbf{KAM Theory:} A solution that is localized and time-dependent can be represented as a linear combination of two or more eigenfunctions of the linear Schrödinger or Wave equation. By introducing a small nonlinear term, KAM theory can be utilized to demonstrate the existence of a nearby solution. However, there are several obstacles to overcome. Firstly, the Nonlinear Fermi Golden rule must fail (to all orders). Additionally, stability issues arise where the state tends to collapse to a stable ground state when nearby, while near an excited state tends to run away due to instability \cite{Sig,soffer2005theory,S-Wei1,S-Wei4,Ts2002}.

\textbf{Self-Similar Ansatz:} The breathers we are familiar with may exhibit self-similarity in a general sense: they are functions of the form $B(w(t)x,t)$. When $w(t)\sim t^{\alpha}$, we obtain the usual slowly spreading weakly localized solutions. To achieve a breather, it is necessary to find a solution with $w(t)$ being a uniformly bounded function.

\textbf{Systems of Equations:} It might be feasible to find a system of coupled equations with such a Breather solution. By eliminating one equation, an effective higher-order equation with breathers may be obtained. This equation would not necessarily be the familiar NLS or NLKG but would be a part of our general theory.

\textbf{Change the Dispersion Relation:} Building on the aforementioned approaches, it may be possible to use alternative dispersion relations along with a simple interaction. For instance, higher-order relations like $\Delta^2+c\Delta$ could arise from a system of equations. Perturbations of such an operator can readily have embedded eigenvalues for general potentials. Alternatively, one might consider a dispersion relation with gaps in the spectrum. This is inspired by the observation that there are breathers on a lattice due to gaps in the spectrum. In such cases, the eigenvalue of the Floquet operator is stable if it falls within a gap.

It's worth noting the interest in such equations in various contexts. Periodic time dependence arises naturally in nonlinear optics, where the $z$ direction in an optical device serves as the time variable for NLS. Hence, etching the $z$ direction reflects a choice of time dependence. Other examples include effective equations of wave dynamics embedded in active or other physical systems. 

Mathematically, focusing on the simplest type of interactions, like a monomial nonlinear term, reveals extra scaling symmetries of the equation, which can restrict the type of solutions. It is perhaps possible to exclude time-dependent solutions using this additional scaling symmetry of the equation.

\subsection{Quantum Ping-Pong}

This is another fundamental question with a straightforward formulation: \emph{Can a positive potential, which is localized in space, bounded, and a regular function, have bound states?}

Such a potential must be time-dependent. Initially, this problem may seem technical, but it demands a nuanced understanding of Quantum Dynamics with time-dependent potentials (a topic often absent from textbooks on QM). This problem, and its slightly more general forms, frequently arise in the analysis of large-time behavior of systems.

The broader question is: assuming the operator $-\Delta+V(x,t)$ has no bound states for each fixed $t$, does it follow that the equation
$$
i\frac{\partial\psi}{\partial t}=-\Delta \psi+V(x,t)\psi
$$
has NO localized solutions?

This remains an open problem in any dimension. In some special cases, known tools such as adiabatic theory \cite{FS} may be applicable to handle slowly (or rapidly) changing time potentials.

A simplified version of this problem is as follows: prove that the equation below has no localized solutions for a \emph{time-independent} $V$ that is positive as described. For small data, this is known, utilizing the $L^p$ estimates for the linear part. Our analysis demonstrates that for large data, all bounded solutions of this equation converge to a free wave plus a weakly localized solution in three or more dimensions. It's noteworthy that the nonlinear term is \emph{defocusing}.
$$
i\frac{\partial\psi}{\partial t}=-\Delta \psi+V(x)\psi +|\psi(x,t)|^2\psi.
$$
There is an important physical aspect to these problems: one aims to localize a quantum particle using laser beams, on a surface or other domain with no localizing potential. This endeavor essentially encapsulates the aforementioned problem. Naturally, a by-product of this inquiry would be to identify an optimal $V(x,t)$ capable of sustaining particle localization for extended periods.

\subsection{System of Equations}

The importance of having an abstract condition that implies the decomposition of the solution into a free wave and a weakly localized part lies in the ability to control equations with implicit interaction terms. For instance, consider an equation of the form
$$
i\frac{\partial\psi}{\partial t}=-\Delta \psi+V(x)\psi +F(x,t,u)\psi
$$

where $u$ is a solution of another equation.

Then, the asymptotic behavior of $\psi$ is understood if one can prove that (in three or more dimensions), for instance, $F(x,t,u) \in L^2_x$ uniformly in time. This raises the question of how to identify interesting models of systems of equations and establish a priory estimates tailored to the system. In particular, it may be necessary to devise a suitable definition of channel wave operators for each of the unknown functions $\psi, u$.

Examples to consider include a wave equation of a charged particle coupled to the Maxwell equation, wave dynamics in the presence of a reactive medium, or coupled to an ordinary differential equation (e.g., modulation equations of soliton + radiation). For instance, the problem of a single quantum or classical particle coupled to NLS modeling a quantum fluid \cite{chen2019mean,frohlich2011some,soffer2018dynamics}.

\subsection{Self-Similar Solutions}

The fundamental question to address is whether self-similar solutions exist for a given dispersive equation with a mass-critical term that is also saturated with an extra term. The same question can be posed about energy-critical terms with saturation, see e.g. \cite{sulem2007nonlinear,fibich2015nonlinear,rodnianski2003asymptotic} . We know that blow-up solutions exist in the above cases without saturation, highlighting the distinction between pure monomial terms and more general equations. It is possible to have saturated equations with no blow-up for all initial conditions in $H^1$, for example. Since one can choose the nonlinear term to be asymptotic to a mass-critical equation after scaling, the question arises: can there be a self-similar solution in this scenario?

{\bf Very Long Range N-body Scattering}

In the study of N-body scattering, it is observed that if the two-body interaction termsvanish as $|x|^{-\mu}$, where $\mu\leq \sqrt{3}-1$, the proofs of Asymptotic Completeness (AC) break down. This occurs because the asymptotic dynamics of some channels are governed by time-dependent potentials, resulting in a self-similar spreading part due to the concentration of the solution at zero frequency over time. The community generally believes that AC holds in these cases but lacks a formal proof.

Based on the analysis outlined here, it is conceivable that in the very long-range case, new scattering channels may emerge, featuring self-similar weakly localized solutions.

\subsection{N-Quasi-Particles}

We have demonstrated that the three quasi-particle case is manageable using the new approach \cite{soffer2023three}. It is also possible that the proof of the free channel wave operator exists in the N-body case, leaving the task of proving the absence of weakly localized solutions for initial conditions in the Hilbert subspace of the continuous spectrum of the Hamiltonian. This aspect may necessitate the development of numerous new propagation estimates.

\subsection{Other Classes of Equations}

Many other systems may exhibit large-time behavior amenable to scattering theory methods. However, in these cases, one needs a working hypothesis about the form of asymptotic solutions, and an idea of how to introduce appropriate channel wave operators. As demonstrated in the section on self-similar solutions, the corresponding "asymptotic dynamics" may involve a simple scaling transformation with "time" represented by $\alpha \ln t$.

Interesting cases include Boltzmann-type equations, with or without singular (i.e., realistic) Kernels. Other notable examples involve dispersive equations in the Fourier space of the equation, as encountered in fluid dynamics. A "simple" example would be:
$$
i\frac{\partial\psi}{\partial t}=-\Delta \psi+h\cdot p\psi +F(|\psi(x,t)|)\psi, \quad \quad p=-i\nabla_x.
$$
$h\in \R^n$ is a fixed vector. In one dimensions this term reads as $c(-i\frac{\partial \psi}{\partial x}),$
$c$ is a constant.
\section{Open Problems - Inverse Scattering}

The inverse scattering problems we are considering stem from scenarios involving dynamic targets. These models involve equations with time-dependent potentials, corresponding to moving targets, as well as nonlinear equations with large data and multichannel scattering. We also delve into problems with partial data.

 The analysis of inverse scattering problems based on time-independent methods is a standard and highly developed approach; many questions in this domain have been posed and partially resolved in the past, as seen in \cite{uhlmann2014inverse,cakoni2014qualitative,cakoni2022inverse,kian2017unique}. Also, refer to \cite{guillarmou2020eigenvalue,arnaiz2022stability}.

An important distinction lies in the setup we adopt: We consider general interaction terms, which may be non-compactly supported, time-dependent, and nonlinear. Thus, we formulate the results and questions in terms of understanding the asymptotic states at temporal and spatial infinity.

Notions such as the Resolvent, Green's function, Generalized Eigenfunctions, and solutions of the Helmholtz's equation are not as prominent in the time-dependent approach. They are crucial in formulating the inverse scattering problem concerning data on the boundary surrounding the support of the interactions. Indeed, part of the questions we seek to address involves merging ideas from both approaches to tackle new or more general classes of inverse scattering problems.

\subsection{Linear Problems}

The simplest problem to consider is finding a potential that changes in time under the action of the Galilean group $\mathcal {G}$ on its argument. Consider the Schrödinger or Wave equation with the following class of potentials:

$$
V(x,t)= V(\mathcal {G}(t) x).
$$
Here, $\mathcal {G}(t)$ corresponds to an element of the group acting on the finite dimensional space $\R^n.$

First, consider boost transformations: $V=V(x-vt)$, where $x,v\in \mathbb{R}^n$ \cite{beceanu2019semilinear}. The question here is how to deduce $v$ from the Scattering Matrix. The potential $V$ is assumed to be a localized and smooth function for simplicity.

In this case, the Scattering matrix exists, and by a coordinate transformation, it can be reduced to a time-independent problem. However, this problem now features a new dispersion relation dependent on the unknown $v$. Yet, since we have a scattering matrix, it should be possible to determine $v$.

{\bf Problem 1- Moving Potential}

Let's delve into a real-life scenario to contextualize the questions:

Imagine a mosquito buzzing around in a room. At time zero, we switch on a light, emitting white (or red) color. We proceed to measure the reflected wave in all possible directions. If we employ multiple sources emitting different colors, we can determine the position and velocity of the mosquito.

However, this approach might not always be practical. Particularly, we would prefer to solve the inverse scattering problem solely using receivers. In this case, we can measure the incident wave and the reflected wave in all directions. Here, we have knowledge about the asymptotic state of the system and significant information about the initial state. The relationship between these states is determined by the Free Channel Wave Operator (not the Scattering matrix):

$$
U(t)u(0) \simeq U_0(t)\Omega^*u(0)
$$
for all $t$ large enough. Our measurements determine $\Omega^*u(0).$
We also know (say) $V(x).$ How to find $v?$

This problem gets more complicated if the velocity is changing:

{\bf Problem 2 Short-Lived Potential}

Suppose first that the target is an object centered at the origin, which at time zero start moving to the point (1,1,1)
and then come back after a finite time and is at rest forever.
This is an example of a short lived potential. So, the general question is, given a potential of the form
$$
V(x,t)=V_0(x)+ W(x,t), \quad W=0 \quad \quad \text {for } |t| >1.
$$
 Can the scattering matrix determine $W?$ See \cite{kian2017unique}.

 {\bf Problem 3- Quasi-Periodic Potential}

 If the time dependence is known, but not linear, then if it is periodic or quasi-periodic, we can turn the problem into a time independent one, using Floquet theory. See e.g. \cite{S-Wei3}. 
 We can prove now the scattering for such problems. How Inverse scattering looks like?

 {\bf Problem 4- Combined Terms Potential}

If one adds a localized time independent potential $W$ to the linearly moving potential, can one use the extra reflections from $W$ to solve the problem with less measurements?

Problem 5: Self-Similar Potentials

This problem stands out as possibly the most significant: We have demonstrated that solutions to a linear problem with a self-similar potential can harbor stable weakly localized states with a non-zero $L^2$ norm. This contravenes the conventional assumption of Asymptotic Completeness. These states neither qualify as bound states nor scattering states. Moreover, the potential is "short-range," decaying uniformly in time like $|x|^{-2}$ at infinity. Consequently, the scattering matrix, which  maps scattering states to scattering states, is unlikely to be unitary. While this remains to be formally proven, it is expected, due to the fact that a weakly bound state, likely not conforming to an elastic coherent state akin to a soliton of a completely integrable system. Thus, the inquiry arises: Can this constrained scattering matrix determine the potential?

Consider the equation:
 $$
    i\frac{\partial\psi}{\partial t}=-\Delta \psi+W(x)\psi +t^{-a}V(|x|/t^{\alpha})\psi,
 $$

Where $W$ and $V(y)$ are smooth and localized functions of $x$ and $y$ respectively. Both parts of the potential support bound states when added to $-\Delta$. By selecting $\alpha <1/2$, stability can be ensured within a certain range. While we have only established the existence of stable self-similar solutions in five or more dimensions, we anticipate this holds true in lower dimensions as well. The power $a$ is uniquely determined by the scaling of $V$, aligning with $x^{-2}$ akin to the Laplacian's scaling (Mass Critical situation). Hence, $a=2\alpha$.

The question then arises: Can the channel  scattering operator from free waves to free waves, determines $W$, $V$, and $\alpha$? If so, how? It is worth noting that a similar question arises in the scattering of radiation on a non-integrable Soliton.

The presence of such solutions in a realistic physical model may bear significant applications. There may also be systems of equations where the solution of one equation leads to a self-similar potential.

 \subsection{ Nonlinear Inverse Scattering}

In recent years, there has been remarkable progress in the mathematical analysis of inverse nonlinear scattering. Specifically, the nonlinear Schrödinger equation (NLS) with a potential term and polynomial nonlinearities (with variable coefficients) has been extensively studied.

In many of these examples, it has been demonstrated that by considering a well-defined class of small initial data in a Banach space, the corresponding Scattering Matrix exists and uniquely determines the potential and the nonlinear terms.
\cite{killip2024determination,C-Mur,weder2001inverse,ardila2023threshold}.

Common to these works are several very useful facts: it is sufficient to solve the problem for small data, within a small ball in a properly chosen Banach space. Furthermore, all calculations required to retrieve the interaction terms are derived from expressions where the full solution at time $t$ is replaced by the free flow, akin to the nonlinear analog of the Born approximation. Therefore, if it can be shown that the Born approximation is uniquely determined by the Free Channel Wave Operator or by the free-to-free Scattering matrix, the results follow.

A crucial class of estimates is needed to show that the difference between the free flow and the exact solution is controlled in a suitable norm. This emphasizes the importance of smallness and imposes additional conditions on the allowed nonlinearities in non-homogeneous situations.

The new possibilities of constructing the Free channel wave operator in the case of large data may be helpful. This should be applied to the free wave part of the scattering, as it is then possible to show that $F_c(|x-vt|\leq t^{\alpha})\psi(t) \sim F_c(|x-vt|\leq t^{\alpha}) e^{-iH_0 t}\psi_+$, where $\psi_+$ is determined by the Wave Operator/Scattering Matrix. We leave this as a forward-looking proposal at this time.

The first question of interest is whether the above results hold in sufficient generality. This would be significant, implying that small data scattering experiments can determine the complex structure of nonlinear systems. It should be noted that explicit classes of initial data that are large and lead to global existence, even for equations with focusing interactions, are known to be constructible, see e.g. \cite{beceanupositivity}.

The second issue, which appears to be quite open and vast, involves addressing the second meta-theorem of scattering theory: whether the scattering matrix (or wave operators) determines \emph{everything} about the system. This entails finding the \emph{properties} of the system.

In linear scattering theory, the S matrix can be employed to identify bound States and resonances, for instance. This necessitates a separate theory for both linear and nonlinear equations. Therefore, it seems natural to pose the following problems: \newpage

{\bf Problem 1}

How can the energies (range) of localized solutions be determined? Is it possible to ascertain this based on the construction of the matrix corresponding to free-to-free channel scattering?

In general, we are dealing with multichannel scattering, which is complicated even in the linear case. Consider the scattering of a photon from a two-particle atom. There are states where three particles are in and two out. Even if our observations are restricted to cases with a photon in and a photon out, and even if we limit ourselves to the subspace of the continuous spectrum (of the three-body problem), some states will have three free particles out and some only a bound state and a free photon. Is it feasible to determine everything solely by observing the photon?

Let us then consider a specific model: the scattering of small radiation off an NLS localized state in a time-independent potential. In this scenario, it is known that AC holds for all such initial data \cite{soffer1990multichannel}. 

For the equation
$$
i\frac{\partial\psi}{\partial t}=-\Delta \psi+W(x)\psi +\lambda |\psi|^p\psi,
$$

in dimension 3, with initial data given by a nonlinear bound state of the above equation plus a small localized perturbation, the asymptotic states are a nearby nonlinear bound state plus a free wave.

 So, the scattering matrix and the wave operator exist; yet, each asymptotic state is given by a free wave {\bf and} a number, $E_{\pm}$, representing the asymptotic energy of the soliton/nonlinear bound state. There is also an asymptotic phase!
By conservation of energy, knowing only the free wave, we can determine the energy of the bound state. We can also determine the mass ($L^2$ norm) of the localized state.

For this procedure to work out, we need to know that Asymptotic Completeness holds, and that the localized part of the asymptotic state with energy $E$ and mass $m$ is either unique, or belongs to a finite discrete set. If we know that the bound states are time independent (up to a phase), then we can find the full scattering matrix.
If we are not close to a stable bound state, it is not known in general if there are breathers or not. So, it seems one cannot know the full scattering matrix from measurements of the free wave only.

Why do we need to know the Full Scattering Matrix?
In the presence of potentials with bound states, smallness of the initial data does not exclude bound states.
Therefore the scattering matrix used in the proofs in this situation \cite{weder2001inverse} uses asymptotic states with the property that \emph{no bound state comes out.}

It is not clear how to ensure that experimentally.
One way of doing it is to use the constructions we introduced of the projections on the scattering states for general nonlinear equations. That means using these projections to find proper initial conditions.
There are also works which give explicit conditions on the initial data that ensure only free wave as asymptotic states \cite{beceanu2021large}.

{\bf Problem 2}

Let us consider a simpler situation where there are no bound states solutions to the in the equation.
For example, it is expected (still open) that if a non-negative localized smooth potential is perturbed by a \emph{defocusing} nonlinear term, there will be global existence and no localized solutions in many cases.
Certainly, this is known for small initial data. So, there should be no problem in constructing the Scattering Matrix and finding the Potential and the nonlinear term as well.

However, we are interested in other properties of the system, such as the resonances.
In the linear time-independent case, we have a good understanding of the resonances. They are eigen-vectors of the analytic continuation of the resolvent of the Hamiltonian and emerge similarly from the scattering matrix.\cite{dyatlov2019mathematical}
See, for example \cite{cakoni2020duality} in the context of Inverse scattering.
For a time-dependent and more general theory of resonances, refer to \cite{S-Wei2,CS2001}.
Another approach to understanding them is to find initial conditions that remain localized for a long time before spreading to infinity.

If we take a double well potential in any dimension and add a nonlinear term, how can we find the corresponding resonances?

This problem involves complex dynamics where the structure of the potential governs and produces solitons \cite{BP3S,dekel2007temporal,dekel2009nonlinear}.
In direct scattering, we know how to handle it in some cases, using the notion of the nonlinear Fermi Golden Rule.
But for inverse scattering, how can we find the value of the FGR from the S-matrix?

One probable approach is to demonstrate that the asymptotic states approach some linear Hamiltonian for certain energies, and then find the resonances of the linear part.

{\bf Problem 3}
The question above raises the possibility that there are resonances which are purely nonlinear: driven by purely nonlinear interaction terms.
How can this happen? Imagine a nonlinearity which has soliton solutions. It can occur that two or more solitons interact for a long time with each other, before separating or merging.
The lifetime of such a process will show up as a resonance in the scattering.
How can we find it from the Scattering Matrix?

{\bf Problem 4- Adding a Background Potential to a Given System}

Adding an extra potential or changing the background from a vacuum to something active may enhance the quality of solving the inverse problem, as well as reduce the amount of data needed \cite{barsi2009imaging}.
This may seem highly dependent on the problem, but some general principles may be possible to prove:
The way to think about it is like adding a (crooked) mirror to the environment, which adds extra information to the observer by watching the (distorted) image reflected by the mirror.
In the model below, one can consider the nonlinear term with power 2, for example, as immersing or coupling a single atom to a quantum fluid. If $\lambda$ is a localized function, then it is a coupling to a droplet of a quantum fluid (the "mirror").
$$
i\frac{\partial\psi}{\partial t}=-\Delta \psi+W(x)\psi +\lambda |\psi|^p\psi,
$$

The crucial point is that the S matrix also determines the mirror!
Therefore, by separating (mathematically or experimentally) the waves coming from the nonlinearity and the mirror $W$, one gets two observations of the target from different directions, in just one measurement.

A rigorous construction of such an example would be great!

    {\bf problem 5- Transparency} 

    Consider first the basic case of Schr\"odinger equation with a localized time independent potential.
    In this case we can prove the the Scattering wave operators exist, complete and bounded on $L^p$, including $L^{\infty}.$
    $$
    \Omega^{\pm} = s-\lim_{t \to \pm \infty}e^{iHt}e^{-iH_0 t}.
    $$
     $H=H_0+V(x)=-\Delta +V(x).$
     
 We can then construct the outgoing generalized eigenfunction of $H$ at frequency $k$

 $e(x,k)= \Omega^+ e^{-ik\cdot x},\quad k\cdot x>0.$

 Now, in general we write $$e(x,k)=e^{-ik\cdot x+if(x,k)} +\mathcal{O}(1/<x>^a), \, a\geq 1.$$

 Transparency at frequency $k$ is the possibility that $f(k,x)=0.$
 Equivalently,
 $$
 (\Omega-I)e(x,k)=\mathcal{O}(1/<x>^a).
 $$
We now localize in a finite interval around $k$ and projects on outgoing states: we then expect the following to hold,
\begin{equation}\label{Enss}
(\Omega-I)\chi(H_0\approx |k|)P^+(A) e(x,k)=\mathcal{O}(1/<x>^a).
\end{equation}

By a key result of Enss \cite{E1978} we know that the above operator is compact on $L^2.$
We also know that $e(x,k)$ is a bounded and differentiable function \cite{beceanu2020structure,S2018,yajima1995wk}.
The question is how to use these observations to relate it to the notion of \emph{transmission eigenvalues}?
Can one use the above compactness to prove the discreteness of the set of transmission eigenvalues?
Knowing $\Omega$, even approximately, how to get an approximate candidate for transparent $k$?
We know how to construct $\Omega$ for time dependent and nonlinear interactions. How to use this knowledge for the construction of transparency? See \cite{cakoni2014qualitative,cakoni2024lack}

\section{Other Problems}
I list here a few problems which are directly and technically related to the topics of the review.

{\bf Problem 1- Exponential Decay of Localized Solutions}

Can one prove exponential decay of the localized solutions of a general class of nonlinear equations and time dependent potentials?
We have already shown that in 5 or more dimensions the solutions have some pointwise decay in $x.$
We also know, quite generally that $A^n\psi$ is uniformly bounded in $L^2$ for all $n$ for purely nonlinear equations (e.g. Inter-critical in 3 or more dimensions) and that the solution is smooth.
It should be pointed out, that even in the standard QM problem, exponential decay is a non-trivial problem.

{\bf Problem 2- The Projection Operator on Scattering States}
 We have shown that the projection operator $P_{sc}$ on scattering states exists for general scattering problems.
  On its range, we could prove Local Decay estimate in 5 or more dimensions.
  
 This projection operator is an important object in many ways. Its range is the set of all initial data that are  pure scattering states for large time. On its range we expect Local Decay, Strichartz estimates and more to hold.

 Now, suppose that we know that the asymptotic energy operator exists, in the \emph{linear} time dependent case,

 $$
 \lim_{t \to \infty} (U(t)\phi, H(t)U(t)\psi(0)) \equiv (\phi, H^+ \psi(0)).
$$

Is it true that $ P_{sc}= P(H^+\geq 0)-P(H^+=0)$?

Find sufficient conditions for $P(H^+=0)=0,$ implying the absence of weakly localized states.

Find necessary conditions for  $P(H^+=0)=0.$
 The above limit defines a self-adjoint operator $H^+.$

Can one prove that Local Decay holds on the range of $P_{sc}$ also in three and four dimensions?
What would be a useful analog in the nonlinear case?

In the linear case, one may be able to use localization of $H^+$ to develop Mourre estimate and prove Propagation Estimates. \cite{HSS1999,Ger}. Can this be pushed to the nonlinear level?

{\bf Problem 3- The nature of Asymptotic States}

In linear multi-channel scattering theory, part of the notion of AC is that the asymptotic parts are moving independently. In particular, they are solutions of the equation defined by the asymptotic dynamics. This is not true for general nonlinear equations. Can one find sufficient conditions implying that the asymptotic states are also exact solutions of the equation? One soliton plus radiation satisfies this condition, but 2 solitons plus radiation may not \cite{martel2011description}, though in this case the asymptotic state is a linear combination of two exact solutions. Find necessary conditions. In particular, suppose the asymptotic state is reached fast as 
$t$ goes to infinity. Is it sufficient?

{\bf Problem 4- Back to Leonardo Da Vinci}

Do we need to know the \emph{flows} in the medium to solve the inverse problem? Using E.M. waves, as they can move in a vacuum, this question does not apply. We consider then, an object immersed in a fluid and being observed by acoustic waves. Do we need to know the flow? Clearly, if the object is moving in the fluid, the position of the object says something about the fluid state. Can we "see" the object without solving for the waves (and vortices, etc.) of the fluid?

 \medskip

DECLARATIONS
\emph{Availability of data and materials:} All relevant information is published or posted online.
\emph{Competing Interest:} N/A
\emph{Funding:} Supported by NSF; see below.
\emph{Authors' contributions:} All by the Author.

\noindent{\textbf{Acknowledgements}: }A.S. was partially supported by an  NSF grant DMS-2205931.
A.S. would like to thank Fioralba Cakoni for many important discussions on Inverse Scattering Theory.

\bibliographystyle{plainnat}
\bibliography{bib}

\vspace{2mm}


\end{document}